\documentclass[prx,10pt,twocolumn,superscriptaddress,floatfix,showpacs]{revtex4-2}

\usepackage{appendix}
\usepackage{comment}
\usepackage[utf8]{inputenc}  
\usepackage[T1]{fontenc}     
\usepackage[british]{babel}  
\usepackage[sc,osf]{mathpazo}
\usepackage[scaled=0.86]{berasans}  
\usepackage[colorlinks=true, allcolors=blue, urlcolor=blue]{hyperref}  
\usepackage{graphicx} 
\usepackage[babel]{microtype}  
\usepackage{amsmath,amssymb,amsthm,bm,amsfonts,mathrsfs,bbm} 

\usepackage{xspace}  
\usepackage{pgfplots}
\usepackage{xcolor,colortbl}
\usepackage{array}
\usepackage{bigstrut}
\usepackage{tcolorbox}
\usepackage{mathdots}

\usepackage{bm}
\usepackage{booktabs}

\newcommand{\tr}{\mathrm{Tr}}

\newcommand{\bs}[1]{\boldsymbol{#1}}

\newtheorem{lemma}{Lemma}
\theoremstyle{definition}

\newcommand{\bei}{\begin{itemize}}
\newcommand{\eei}{\end{itemize}}
\newcommand{\ket}[1]{|#1\rangle}
\newcommand{\bra}[1]{\langle#1|}

\def\<{\langle}
\def\>{\rangle}
\newcommand{{\Cn}}{{\mathbb{C}^4}}
\newcommand{{\CN}}{{\mathbb{C}^{2n}}}
\newcommand{{\BC}}{{\mathcal{B}(\mathbb{C}^n)}}
\newcommand{{\BBC}}{{\mathcal{B}(\mathbb{C}^{2n})}}
\def\oper{{\mathchoice{\rm 1\mskip-4mu l}{\rm 1\mskip-4mu l}{\rm 1\mskip-4.5mu l}{\rm 1\mskip-5mu l}}}

\usepackage{physics}
\usepackage{amsmath}
\usepackage{mathtools}
\usepackage{amsthm}
\usepackage{amssymb}
\usepackage{bm}

\usepackage[normalem]{ulem}
\usepackage{physics}
\usepackage{mdframed}
\newmdtheoremenv{theo}{Theorem}
\usepackage{amssymb} 

\definecolor{mycolor}{rgb}{0.7, 0.75, 0.71}

\newmdenv[innerlinewidth=0.5pt, roundcorner=4pt,linecolor=mycolor,innerleftmargin=6pt,
innerrightmargin=6pt,innertopmargin=2pt,innerbottommargin=6pt]{mybox}

\newcommand{\B}{\mathcal{B}}

\newcommand{\I}{\mathbbm{I}}

\newcommand{\set}[1]{\left\{ #1 \right\}}

\newcommand{\be}{\begin{equation}}
\newcommand{\ee}{\end{equation}}
\newcommand{\bea}{\begin{eqnarray}}
\newcommand{\eea}{\end{eqnarray}}
\newcommand{\bes}{\begin{equation*}}
\newcommand{\ees}{\end{equation*}}
\newcommand{\beas}{\begin{eqnarray*}}
\newcommand{\eeas}{\end{eqnarray*}}
\newcommand{\bpm}{\begin{pmatrix}}
\newcommand{\epm}{\end{pmatrix}}
\newcommand{\bmb}{\begin{mybox}}
\newcommand{\emb}{\end{mybox}}

\newcommand{\sep}{\mathrm{sep}}

\newcommand{\nn}{\nonumber}

\newcommand{\id}{\mathbbm{I}}

\renewcommand{\H}{\mathcal{H}}
\newcommand{\opn}{{\mathrm{op}}}

\theoremstyle{remark}
\theoremstyle{definition}
\newtheorem{proposition}{Proposition}
\newtheorem{theorem}{Theorem}
\newtheorem{corollary}{Corollary}

\usepackage{xcolor}

\begin{document}

\title{ Mirrored Entanglement Witnesses for Multipartite and High-Dimensional Quantum Systems }


\author{Jiheon Seong}  
\email{equally contributed }
\affiliation{School of Electrical Engineering, Korea Advanced Institute of Science and Technology (KAIST), 291 Daehak-ro, Yuseong-gu, Daejeon 34141, Republic of Korea }

\author{Anindita Bera}
\email{equally contributed }
\affiliation{Department of Mathematics, Birla Institute of Technology Mesra, Jharkhand 835215, India}

\author{Beatrix C. Hiesmayr}
\affiliation{University of Vienna, Faculty of Physics, W\"ahringerstrasse 17, 1090 Vienna, Austria}

\author{Dariusz Chru{\'s}ci{\'n}ski}
\affiliation{Institute of Physics, Faculty of Physics, Astronomy and Informatics,
Nicolaus Copernicus University, Grudzi\c{a}dzka 5/7, 87--100 Toru{\'n}, Poland}

\author{Joonwoo Bae}
\affiliation{School of Electrical Engineering, Korea Advanced Institute of Science and Technology (KAIST), 291 Daehak-ro, Yuseong-gu, Daejeon 34141, Republic of Korea }

\begin{abstract}
Entanglement witnesses (EWs) are a versatile tool to detect entangled states and characterize related properties of entanglement in quantum information theory. The verification of entangled states via EWs relies on the fact that separable states form a convex set; this also means that the framework presented by EWs generally applies to other quantum resources where free resources contain the convexity. A witness $W$ corresponds to an observable satisfying $\tr[W\sigma_{\sep}]\geq 0$ for all separable states $\sigma_{\sep}$; entangled states are detected once the inequality is violated. Recently, mirrored EWs have been introduced by showing that there exist non-trivial upper bounds to EWs, 
\bea
u_W\geq \tr[W\sigma_{\sep}]\geq 0. \nonumber 
\eea
An upper bound to a witness $W$ signifies the existence of the other one $M$, called a mirrored EW, such that $W+M = u_W \I\otimes \I$. The framework of mirrored EWs shows that a single EW can be even more useful, as it can detect a larger set of entangled states by lower and upper bounds. 

In this work, we develop and investigate mirrored EWs for multipartite qubit states and also for high-dimensional systems, to find the efficiency and effectiveness of mirrored EWs in detecting entangled states. We provide mirrored EWs for $n$-partite GHZ states, graph states such as two-colorable states and tripartite bound entangled states. We also show that optimal EWs can be reflected with each other. For bipartite systems, we present mirrored EWs for existing optimal EWs and also construct a mirrored pair of optimal EWs in dimension three. Finally, we generalize mirrored EWs such that a pair of EWs can be connected by another EW, i.e., $W+M =\mathbbm{K}$ is also an EW. Our results enhance the capability of EWs to detect a larger set of entangled states in multipartite and high-dimensional quantum systems. 
\end{abstract}
\maketitle

\section{Introduction}

In quantum information processing, each element of quantum theory, such as states, measurement, and dynamics, is leveraged for various applications \cite{Neumann:2018aa}. The fact that an unknown quantum state alone cannot be copied has been a cornerstone of quantum communication that achieves a higher level of security without relying on computational assumptions \cite{Wootters:1982aa, BENNETT20147}. Measurement-based quantum computing realizes quantum logic gates through quantum states and local measurements only \cite{PhysRevLett.86.5188}. Quantum dynamics corresponds to quantum logic gates that enable efficient computation, such as prime-number factoring or amplitude amplifications \cite{365700, 10.1145/237814.237866}. Quantum measurements can also demonstrate probabilities that cannot be achieved by classical systems, such as contextuality \cite{KOCHEN:1967aa, PhysRevA.71.052108} and incompatibility \cite{Heinosaari_2016}, which are closely connected to Bell nonlocality \cite{Bell, chsh, RevModPhys.86.419}. All of these phenomena can be used to certify the quantum properties of systems. 

Entanglement is a resource behind the quantum information applications. The monogamy of entanglement, or the shareability of entangled states, is another statement of the no-cloning theorem \cite{Werner:1989aa, PhysRevA.69.022308, PhysRevA.73.012112} and the precondition for secure quantum communication \cite{PhysRevLett.92.217903, PhysRevLett.94.020501,ROY2021127143}. Two-body interactions that can create entangled states are a key to universal quantum computation \cite{PhysRevA.52.3457,PhysRevA.102.062431}. Nonclassical correlations such as Bell nonlocality and quantum steering certify the presence of entanglement \cite{PhysRevLett.98.140402, RevModPhys.92.015001}; they do not occur without entanglement. Note also that quantum steering is isomorphic to the measurement incompatibility \cite{PhysRevLett.115.230402}. 

Therefore, the question lies in verifying entangled states in practice \cite{RevModPhys.81.865, GUHNE20091,PhysRevA.98.062304,dense-coding,giovanni2025, PhysRevA.101.012341}. Once entangled states are detected, they can be applied to various tasks in quantum information processing. However, verifying entangled states remains a challenging task in general. To describe the problem, let $\B( \H)$ denote the set of bounded linear operators on a Hilbert space $\H$, and $S(\H)$ the set of quantum states; $\rho\in S(\H)$ is non-negative and of a unit trace, $\tr[\rho] =1$ and $\rho\geq 0$. 

\subsubsection*{Entanglement detection after state tomography}

On the one hand, let us consider a scenario where a quantum state $\rho \in S(\H \otimes \H)$ may be provided after state verification, e.g., quantum tomography. The computational complexity of deciding whether it is entangled is known as {\it NP-Hard} \cite{10.1145/780542.780545}. As a theoretical tool, the so-called partial transpose was proposed, from which it follows that the state must be entangled if it is no longer a quantum state after the partial transpose, i.e., after taking the transpose on a subsystem only \cite{PhysRevLett.77.1413, HORODECKI19961}. However, there are entangled states which remain positive after partial transpose (PPT), called PPT entangled states \cite{PhysRevLett.82.1056}, see also a recent review \cite{doi:10.1142/S0219749925300037}. There have been other theoretical tools, including the realignment criteria \cite{10.5555/2011534.2011535}, moments of matrices \cite{PhysRevLett.122.120505}, and symmetric extensions \cite{PhysRevA.69.022308, PhysRevA.80.052306, PhysRevLett.109.160502}. 

The precise characterization of entangled states can be achieved by referring to allowed and disallowed transformations on quantum states \cite{HORODECKI19961}. Quantum operations are generally described by positive (P) and completely positive (CP) maps on quantum states,
\bea
&&(\mathrm{P}): \Lambda \geq 0 \iff \Lambda[ \rho] \geq 0,~ \forall \rho \in S(\H) \ \nonumber \\
&&(\mathrm{CP}): \id \otimes \Lambda \geq 0 \iff \id\otimes \Lambda[ \rho] \geq 0,~\forall \rho \in S(\H\otimes \H).~~~~ \nonumber 
\eea 
Positive but not completely positive maps, which do not correspond to legitimate quantum operations, can define entangled states. That is, a state $\rho\in S(\H\otimes \H)$  is entangled if and only if there exists a positive but not completely positive map such that $\id\otimes \Lambda[\rho] \ngeq 0$.

\subsubsection*{Entanglement detection before state tomography}

On the other hand, let us consider a scenario where one attempts to determine whether the state is entangled or not, even before a quantum state is verified. The scenario may share similarities with the Bell experiment, for instance the Clauser–Horne–Shimony–Holt inequality \cite{chsh}; an expectation value larger than $2$ certifies the existence of entanglement. That is, Bell tests demonstrate that it is possible to detect the presence of entanglement from measurements alone, even when these measurements are tomographically incomplete, i.e., before the quantum state is verified \cite{TERHAL2002313}. 

Entanglement witnesses (EWs) are exactly those observables that realize the detection of entanglement for unknown quantum states in general \cite{TERHAL2002313, PhysRevA.62.052310}. An EW can be constructed as a dual to a positive but not completely positive map $\Lambda$, 
\bea
W = (\id \otimes \Lambda)[Q],\nonumber
\eea
for some $Q\geq 0$. Therefore, EWs characterize the set of separable states; equivalently, a state $\rho$ is entangled if and only if there exists an observable $W=W^{\dagger}$ such that 
\bea
\tr [ W \rho] < 0 ~ \mathrm{whereas} ~ \tr[W \sigma_{sep} ]  \geq 0~\forall \sigma_{sep} \in S_{sep},
\label{eq:defw}
\eea
where $S_{\sep}$ denotes the set of separable states. Note also that the detection scheme is given by the relation, $\tr[W\rho] = \tr[Q (\id \otimes \Lambda)^{\dagger} [\rho]]$ where $\Lambda$ is not completely positive. Moreover, an EW can be factorized into local observables $\{ A_i\}$ and $\{ B_j\}$
\bea
W = \sum_{ij} c_{ij} A_i \otimes B_j,\label{eq:wlocc}
\eea
for constants $\{c_{ij}\}$ \cite{Guhne:2003aa}. Thus, EWs are also feasible with experiments performing local measurements only. 

\begin{figure}[t]
\centering
\includegraphics[width=0.8\columnwidth]{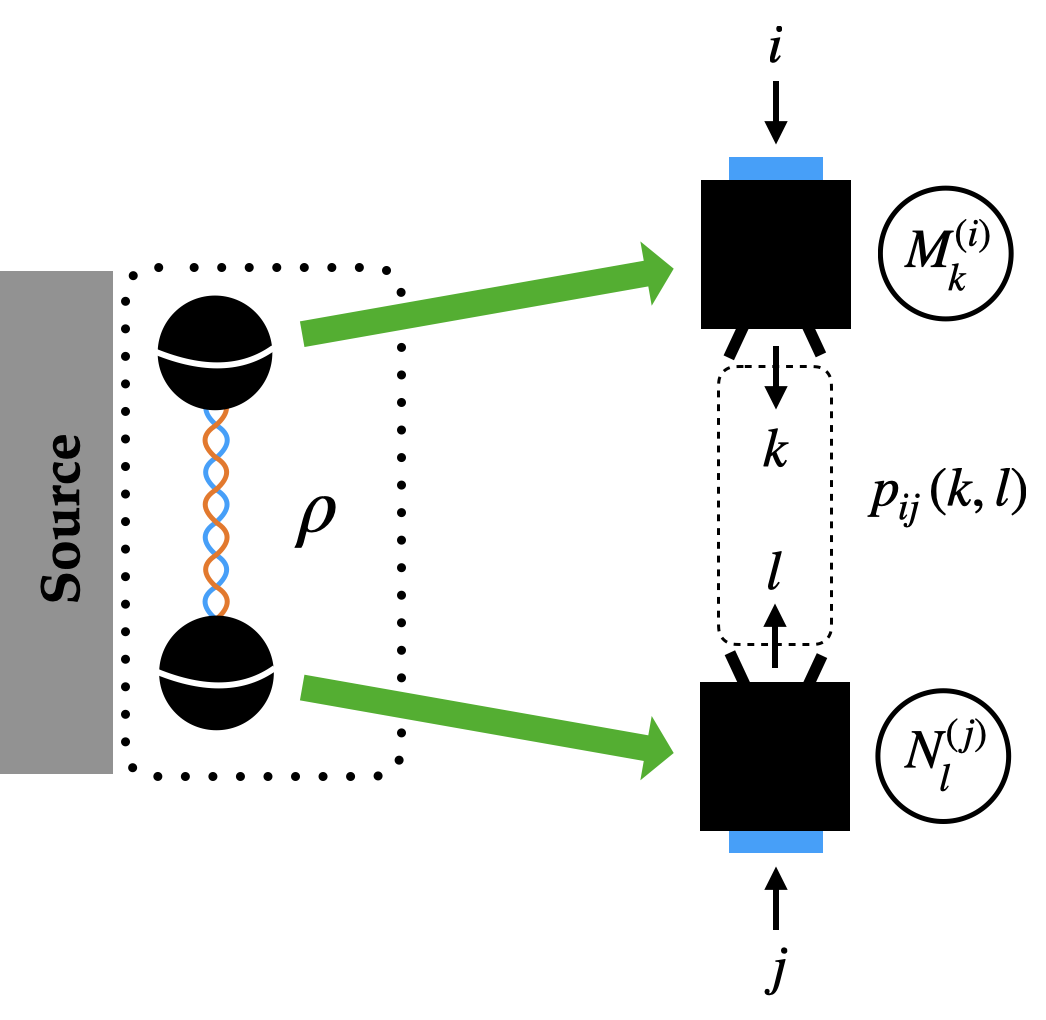}
\caption{ An entangled state $\rho$ can be detected by collecting detection events and finding an estimation, see Eqs. (\ref{eq:pij}) and (\ref{eq:fund}). A state $\rho$ produced from the source (grey rectangle) is distributed to two laboratories (green arrows). Measurement settings labeled $i$ (input from the top) and $j$ (input from the bottom) are chosen, and provide outcomes $k$ and $l$, respectively. Assuming that measurement settings are trusted, that is, given as $\{M_k^{(i)}\}_k$ and $\{N_l^{(j)}\}_l$ according to a decomposition of an EW in Eq. (\ref{eq:wlocc}), two parties achieve an estimation in Eq. (\ref{eq:defw}) from which they may find if a state is entangled. In an MDI framework, a measurement setting is replaced by a state preparation and relaxes the assumption to trust a measurement device, see Eq. (\ref{eq:mdiew}). EWs present the fundamental framework of detecting entangled states. }
\label{fig:mdi}
\end{figure}

\subsubsection*{Entanglement witnesses and their implementation: standard and measurement-device-independent constructions }

{ The fact that an EW can be factorized into local observables presents the fundamental framework that allows one to verify entangled states in experiment. One can decompose observables in terms of positive-operator-valued-measure (POVM) elements, $A_i = \sum_k a_{ki}M_{k}^{(i)}$ and $B_j = \sum_l b_{lj}N_{l}^{(j)}$ with POVM elements $\{M_{k}^{(i)} \}$ and $\{N_{l}^{(j)} \}$ and real numbers $\{ a_{ki}, b_{lj}\}$. Let $p_{ij}(k,l)$ denote the probability of having outcomes $(k,l)$ for measurement settings for local observables $A_i$ and $B_j$
\bea
p_{ij}(k,l) = \tr[\rho M_{k}^{(i)} \otimes N_{l}^{(j)}] \label{eq:pij}
\eea 
where we recall that $M_{k}^{(i)}$ and $N_{l}^{(j)}$ are POVM elements for local observables $A_i$ and $B_j$, respectively. From Eq. (\ref{eq:defw}), an EW can be estimated as follows,
\bea
\tr[W\rho] = \sum_{ijkl} c_{ij}a_{ki} b_{lj} p_{ij}(k,l),  \label{eq:fund}
\eea
which may tell if a state $\rho$ is entangled. For instance, Bell operators can be expressed as above \cite{chsh}. }

{As shown in Eq. (\ref{eq:fund}), an EW signfies a set of measurements and coefficients to construct measurement settings from which probabilities $p_{ij}(k,l)$ are obtained to conclude if a state is entangled. Since measurements in practice may contain imperfections such as non-unit detection efficiencies leading to detection loopholes as well as readout errors, extensive effort has been devoted to attempting to relax assumptions on measurement settings, meaning that probabilities $p_{ij}(k,l)$ are obtained without trust in measurements. For instance, it is fairly straightforward to convert an EW to its measurement-device-independent (MDI) counterpart \cite{PhysRevLett.110.060405, PhysRevLett.108.200401}; for a local decomposition of an EW, see Eq. (\ref{eq:wlocc}), its  MDI counterpart may be given by,
\bea
W_{\mathrm{MDI}} = \sum_{ij} c_{ij} A_{i}^T \otimes B_{j}^T = \sum_{ijkl} c_{ij}a_{ki} b_{lj} M_{k}^{(i)T} \otimes N_{l}^{(j)T}\label{eq:mdiew}
\eea
where $T$ denotes a transpose. Device-independent verification of entangled states, which does not assume specific properties of quantum states either such as a Hilbert space dimension, can also be derived \cite{PhysRevLett.121.180503}. Thus, EWs in Eqs. (\ref{eq:defw}) and (\ref{eq:fund}) presents the fundamental framework for detecting entangled states, and the level of certification of entanglement may be adjusted based on the assumptions made about implementations; standard EWs assume trust in measurement settings.}


We emphasize that the fundamental structure that enables EWs to characterize the set of separable states is convexity. A convex mixture of separable states is also separable. An element outside the set, i.e., an entangled state, can be distinguished by a hyperplane, which corresponds to an EW. EWs for multipartite systems can be constructed by exploiting the convexity and can be used to find a rich structure of multipartite entangled states, such as genuine multipartite entangled states.

\subsubsection*{ Entanglement witnesses for noisy states }

Let us consider a measurement setting that realizes an EW that aims to detect a state $\rho$. If the state suffers from noise but remains entangled even after, it is natural to ask whether the measurement setting can be used to verify the entanglement of the resulting state, see also \cite{PhysRevX.8.021072}. We here point out that an EW is robust against non-unitary noise, whereas it is not under unitary errors.

\begin{itemize}
    \item { Non-unitary noise: As an instance of non-unitary noise, let us consider a state that suffers from a depolarization channel 
\begin{equation}
\mathcal{N}^{(p)}(\cdot) \;=\; (1-p)\,(\cdot) \;+\; p\,\frac{\id}{d}, \qquad p\in[0,1].\nonumber
\end{equation}
with a noise parameter $p$. Against a noisy channel, a measurement setting for a witness $W$, one achieves an estimation,
\bea
&& \tr\!\big[W\,\mathcal{N}^{(p)}(\rho)\big] = \nonumber \\
&& (1-p)\,\tr \big[ W \rho \big] \;+\; \frac{p}{d} \tr W. \nonumber
\eea
Hence, entanglement can be detected for depolarization $p < p^\star$, where
\begin{equation}
p^\star \;:=\; \frac{d \abs{\tr [ W \rho ]}}{d \abs{\tr [ W \rho ]} +   \tr W}. 
\label{eq:white-noise-threshold}
\end{equation}
An EW is robust against non-unitary noise. }
\item 
{  Unitary errors: Let us consider unitary errors on a state,
\begin{equation}
\mathcal{N}^{(p)}_U(\cdot) \;=\; U\,(\cdot)\,U^\dagger,~\mathrm{where}~ \| U - \I \|_\opn = p,
\end{equation}
and $\| \cdot \|_{\opn}$ denote the operator norm $\| A \|_\opn = \sup_{\|\psi\|=1} |\langle \psi | A | \psi \rangle|$. We have the following inequality, for which the detailed derivation is shown in the Appendix: 
\bea
\Big|\tr\!\big[W \mathcal{N}(\rho) \big] -  \tr \! \big[ W \rho \big]\Big| &=& \Big|\tr\!\big[W(U\rho U^\dagger - \rho)\big]\Big| \nonumber \\
 &\;\le\;& 2\,\|W\|_\opn \, \|U-\id\|_\opn\nonumber
\eea
from which we obtain the robustness condition,
\bea
\|U-\id\|_\opn = p \;<\; p_U^\star \;:=\; \frac{|\tr [ W \rho ]|}{2\,\|W\|_\opn}.
\label{eq:unitary-threshold}
\eea
Hence, for unitaries with $p< p_{U}^*$, a resulting state $U\rho U^{\dagger}$ can be detected since,
\bea
\tr  \big[W\,\mathcal{N}_U(\rho)\big]
\le -\abs{\tr [ W \rho ]} + 2\,\|W\|_\opn \, \|U-\id\|_\opn\nonumber
\eea
where the right-hand side is negative. }
{ 
\item Local unitaries: 
A shortcoming of EWs for detecting entangled states is that EWs are not robust against local unitaries. A straightforward example is two Bell states $ |\phi^{\pm}\rangle = \frac{1}{\sqrt{2}} (|00\rangle \pm |11\rangle)$, which are equivalent to local unitaries. They are equally entangled as they are interconvertible by local unitaries. It holds that 
\bea
\tr[W|\phi^+\rangle \langle \phi^+|] + \tr[W|\phi^- \rangle \langle \phi^- |] \geq 0,\nonumber 
\eea 
since 
\bea
\tr[W \frac{1}{2} (| 00\rangle \langle 00| + |11\rangle \langle 11|) ] \geq 0. \nonumber
\eea
One can conclude that no EW detects both of them: at best, an EW detects either of the Bell states. The shortcoming may also be rephrased as the limitation of the size $D_W$ for an EW. 

For local unitaries for qubits, $U=\bigotimes_{j=1}^N U_j$ with $U_j=e^{i(\theta_j/2)\,\bs{n_j} \cdot \bs{\sigma_j}}$, then we have the following
\bea
\|U-\id\|_\opn \le \tfrac12\sum_{j}|\theta_j|,\nonumber
\eea
for which the derivation is shown in the Appendix. Then, an EW can detect local unitarily equivalent states, i.e., $\tr\!\big[W\,U\rho U^\dagger\big] < 0$
whenever
\bea
\sum_{j=1}^N |\theta_j| \;<\; p_U^\star
\eea
since it holds $\sum_{j=1}^N |\theta_j| \;<\; p_U^\star =  \abs{\tr [W \rho]} /\|W\|_\opn$.}
\end{itemize}

In general, EWs can also be improved to detect highly noisy entangled states \cite{PhysRevA.62.052310}. Suppose that for some $\epsilon >0$ and $P\geq 0$, we have
\bea
W^{'} = W-\epsilon P.\nonumber
\eea
Here $W^{'}$ is also an EW. Let $D_{W }$ denote the set of states detected by an EW $W$
\bea
D_W = \{  \rho \in S (\H \otimes \H) ~:~ \tr[W \rho] <0~\}. \label{eq:dw}
\eea
It holds that $D_W\subset D_{W^{'}}$. In this case, two EWs can be compared and we call $W^{'}$ is {\it finer} than $W$. EWs which cannot be improved anymore, that is, {\it finer} than any other, are called optimal EWs.

\subsubsection*{ Mirrored entanglement witnesses } 

Recently, a non-trivial upper bound to an EW has been presented \cite{Bae2020}. For an EW, there exists $u_W>0$ such that
\bea
u_W \geq \tr[W\sigma_{\sep}] \geq 0,~~\forall \sigma_{\sep}\in S_{\sep}.
\label{eq:upper}
\eea
where the upper bound satisfied by all separable states may be obtained by $u_W = \max_{\sigma \in S_{\sep}} \tr[W \sigma] $. {Note that, in fact, all EWs have upper bounds; hence, the construction in Eq. (\ref{eq:upper}) works for all EWs for high-dimensional and multipartite quantum systems. }

{To derive an upper bound for an EW, it suffices to restrict the optimization over product states $|ef\rangle \in \H\otimes \H$. For bipartite $d\otimes d$ systems, a product state $|ef\rangle$ can be expressed by $4(d-1)$ parameters, which are linear with respect to a dimension $d$. One should optimize the parameters to find a maximal value $u_W$. Note that the maximization is different from the eigen-decomposition that finds eigenvalues of an EW. }

{
Although the optimization may be feasible for small values of $d$, the problem is also NP-hard \cite{SHI201551}. There may also be an error in a bound, denoted by $\epsilon$, obtained from numerical optimization. One can safely choose an upper bound to an EW as $u_{W} = \widetilde{u}_{W} +\epsilon$ for a numerically calculated one $\widetilde{u}_{W}$; in this way, an upper bound can be used to unambiguously discriminate entangled states from separable states. } 

The upper bound can also be re-expressed by introducing an EW, 
\bea
\tr[M \sigma_{\sep} ] \geq 0, ~~\forall \sigma_{\sep}\in S_{\sep}.
\label{eq:lower}
\eea
where $M=   u_W \I -W  $. Entangled states detected by the upper bound in Eq. (\ref{eq:upper}) can be identified by the detection by a standard EW in Eq. (\ref{eq:lower}), and vice versa; for an entangled state $\rho$, it holds that
\bea
u_W < \tr[W \rho] \iff   \tr[ M  \rho] <0. \label{eq:equ}
\eea
Then, since $W = u_W\I - M$ and the lower bound in Eq. (\ref{eq:upper}), it holds that 
\bea
u_W \geq \tr[M \sigma_{\sep}] \geq 0. \nonumber
\eea
Similarly, we have for an entangled state $\rho$,  
\bea
 \tr[ W  \rho] <0 \iff  \tr[M \rho] > u_W . \label{eq:eql}
\eea
We call $W$ and $M$ mirrored EWs. {In Ref. \cite{Bae2020}, a mirrored EW can detect two entangled states $(|00\rangle \pm |11\rangle)/\sqrt{2}$ that are local unitarily connected. Hence, the framework of mirrored EWs enables entanglement detection more robust against unitary errors. }

Let $D_{W}^{(U)} $ denote the set of states detected by an upper bound of an EW $W$ and $D_{W}^{(L)}$ by the lower bound. We introduce the set of states detected by mirrored EWs, denoted by $D_{W}^{(M)}$, which is equal to the union,
\bea
D_{W}^{( M )} =  D_{W}^{(U)} \cup D_{W}^{(L)}. \label{eq:dwm} 
\eea
Note that $D_{W}^{(U)} = D_{M}^{(L)}$ from Eq. (\ref{eq:equ}) and $D_{W}^{(L)} = D_{M}^{(U)}$ from Eq. (\ref{eq:eql}). Remarkably, mirrored EWs can detect two Bell states $|\phi^{\pm}\rangle$; one is by a lower bound, and the other by an upper bound \cite{Bae2020}. Thus, we have 
\bea
D_{W} = D_{W}^{(L)} \subsetneq D_{W}^{(M)},\nonumber
\eea
i.e., a strictly larger set of entangled states can be detected by mirrored EWs.  

We reiterate that implementing mirrored EWs does not require an additional experimental cost. The experimental setup and measurement data are used to estimate the expectation value of a single EW with respect to unknown states. Then, mirrored EWs offer opportunities to detect entangled states twice, by lower and upper bounds, whereas a standard one does so once with a lower bound. Hence, a large set of entangled states can be verified by mirrored EWs, without an extra experimental effort. Mirrored EWs are also of theoretical interest. An EW and its upper bound $(W, u_W)$ produce its mirrored one as a non-trivial EW. Equivalently, a positive but not completely positive map can be generated. 

The ultimate usefulness and significance of mirrored EWs lie in their ability to detect entangled states by upper and lower bounds, beyond the standard ones with a lower bound. To this end, there have been extensive efforts to find mirrored operators of optimal EWs. In fact, it is highly non-trivial to have optimal EWs after mirroring. For two-qubit states, mirroring an optimal EW does not result in an EW but a quantum state; hence, the conclusion is drawn that an EW must not be optimal to have an EW after mirroring \cite{Bera:2023aa}. Numerous examples in $3\otimes 3$ and $4\otimes 4$ systems have shown that mirroring optimal EWs, one cannot detect PPTES \cite{Bera:2023aa}. Non-decomposable EWs can be a mirrored pair of EWs while neither of them is optimal. Remarkably, there is an instance of a mirrored pair of optimal EWs in $3\otimes 3$. \\

\subsubsection*{ Summary of the contributions }

In this work, we investigate mirrored EWs for multipartite and high-dimensional quantum systems. We provide a pedagogical review on mirrored EWs and present the structure of mirrored EWs via structural physical approximations and general observables. For multipartite systems, we show mirrored EWs for $n$-partite Greenberger–Horne–Zeilinger (GHZ) states \cite{GHZ}, a class of graph states such as cluster states, and tripartite bound entangled states. In all of the cases, EWs in a mirrored pair are equivalent up to local unitaries, signifying that two EWs detect distinct sets of entangled states that are equivalent up to local unitaries. Remarkably, we also present a mirrored pair of optimal EWs for three-qubit states.

For bipartite systems, we provide a fairly comprehensive review on mirrored EWs, including a pair of optimal EWs in dimension three. We then present a generalization of mirrored EWs such that EWs are mirrored via another EW, not an identity operator. In this case, it is constructive to mirror optimal EWs with each other, and they may also detect a fraction of entangled states in common. 

We remark that all of the mirrored pairs of optimal EWs obtained so far are local unitarily equivalent; that is, mirrored EWs are particularly useful to verify entangled states under rotations of local bases. The result resolves a drawback of standard EWs, which are robust against noise from interactions with an environment but not for rotations of local bases. 

This work is structured as follows. In Sec. \ref{sec:general}, we introduce EWs and their properties. The general structure of mirrored EWs is provided, as well as their relations ot structural physical approximations. In Sec. \ref{sec:multi}, we show mirrored EWs for multipartite states and investigate their properties, the separability windows, and the equivalence of local unitaries of mirrored EWs. A pair of optimal EWs are also provided for three-qubit states. In Sec. \ref{sec:bio}, mirroring of optimal EWs is investigated. A mirrored pair of optimal EWs is shown. In Sec. \ref{sec:big}, generalized mirrored EWs are presented. Mirroring EWs with respect to another EW is shown. In Sec. \ref{sec:open}, we list open problems for mirrored EWs, and in Sec. \ref{sec:con}, we summarize our findings.\\

\section{General properties}
\label{sec:general}

To begin with, we here summarize notations and terminologies to be used throughout. Let $W$ denote a standard EW and $M$ its mirrored one. We write $S_{\sep}$ as a set of separable states. Then, the set of states detected a witness $W$ with its lower bound is denoted as $D_W$ or $D_{W}^{(L)}$
\bea
D_W = D_{W}^{(L)} = \{\rho~:~\tr[W\rho] <0 \}.\nonumber
\eea
Then, $D_{W}^{(U)}$ denotes the set of states detected by an upper bound of $W$,  
\bea
D_{W}^{(U)} = \{\rho~:~\tr[W\rho] >u_W, \}\nonumber
\eea
where an upper bound $u_W$ relies on a given witness. Then, we also write by $D_{W}^{(M)} =D_{W}^{(U)} \cup D_{W}^{(L)}$, i.e., all of the states detected by both upper and lower bounds of a witness $W$.  

For multipartite systems, Pauli matrices are denoted by $\I$, $X$, $Y$, and $Z$. We use the superscript for the label of a qubit. For instance, we may write a Pauli $X$ gate on the second qubit among three qubits as $
\I^{(1)}X^{(2)} \I^{(3)}$.

\subsection{Entanglement witnesses}

We here collect definitions, results, and properties of EWs. As mentioned above, with connections to positive but not completely positive maps, we consider bipartite systems. 

Let us begin by describing a general optimality of EWs. Two EWs $W_1$ and $W_2$ can be compared by the sets of states $D_{W_1}$ and $D_{W_2}$, respectively. $W_1$ is finer than $W_2$ if $D_{W_2}\subset D_{W_1}$, similarly to entanglement, which can be compared by local operations and classical communication. It follows that a witness 
\bea
W^{'} = W -\epsilon P, \label{eq:optimizationd}
\eea
for some $\epsilon>0$ and $P\geq 0$ is finer than $W$. Once $W$ cannot be optimized anymore and thus is finer than any other one, it is called an optimal EW. 

The spanning property is useful for finding whether an EW is optimal \cite{PhysRevA.62.052310}. We write by product states $|\psi_k \otimes \phi_k\rangle$ that vanish on a witness $W$: $\langle \psi_k\otimes \phi_k| W |\psi_k \otimes \phi_k \rangle=0$. Then, we call a witness $W$ has the spanning property if 
\bea
\mathrm{spanning~property:~}\mathrm{span}\{|\psi_k\otimes \phi_k\rangle \} = \H \otimes \H,
\label{eq:opt}
\eea
and $W$ is optimal. However, the converse is not generally true. Several examples of EWs satisfying (\ref{eq:opt}) (hence optimal) already exist in the literature~\cite{S71,Hall,Breuer,Justyna1,BERA2023131,ani23,Scala_2024,BERA2023131} (cf. also review papers \cite{Kye-Rev,Hansen}). Since the condition \eqref{eq:opt} is only sufficient, optimal EWs can be identified without the spanning property \cite{Remik}. 
 
 Then, the properties of EWs follow from the classification of positive maps. A positive and completely positive map is called decomposable if it can be written in terms of the transpose $T$ such that $\Lambda = \Lambda_1 +T\circ \Lambda_2$ where $\Lambda_1$ and $\Lambda_2$ are completely positive. An EW $W\in \B(\H\otimes \H)$ is also called decomposable if 
\begin{equation}
    W = P + Q^\Gamma,\label{eq:decomposable}
\end{equation}
where $P,Q \geq 0$ and $\Gamma$ denotes the partial transposition. An optimal decomposable EW is given $W= Q^{\Gamma}$ for some $Q\geq 0$. It is clear from the definition that a decomposable EW is unable to detect PPT entangled states. 

EWs detecting PPT entangled states are dual to non-decomposable (ND) maps; therefore, they are referred to as ND EWs. Consider an ND EW, denoted by $W$, and suppose that it remains an EW after subtracting a decomposable operator $P_D$, see Eq. (\ref{eq:decomposable})
\bea
W^{'} = W-\epsilon P_D,
\label{eq:optimizationnd} 
\eea
for some $\epsilon >0$. We call $W^{'}$ ND-finer than $W$ since it detects a larger set of PPT entangled states. An EW is ND-optimal when no decomposable operator can be subtracted anymore; it is ND-finer than any other. 

NDEWs can also be described as those EWs which remain EWs after the partial transpose, i.e., both $W$ and $W^{\Gamma}$ are EWs. It turns out that $W$ is an ND-optimal EW if and only if $W$ and $W^{\Gamma}$ are both optimal EWs.

From Eqs. (\ref{eq:optimizationd}) and (\ref{eq:optimizationnd}), one can also consider an optimization with some other EWs, i.e.,
\bea
W^{'} = W-\epsilon W^{''}, \nonumber 
\eea
for some $\epsilon>0$ and {and} EW $W^{''}$. In other words, a witness $W$ is a convex combination of other EWs $W^{'}$ and $W^{''}$. Once the subtraction is not possible, i.e., $W^{'}$ is no longer an EW, an EW is called extremal. Extremal EWs are optimal, but optimal EWs are generally not extremal.


    
\subsection{ Mirrored entanglement witnesses}

EWs, denoted by $W$ and $M$, are mirrored with each other if they are related as mentioned below:
\bea
W + M = \mu \I_A \otimes \I_B, \label{eq:mr}
\eea
for some $\mu\geq 0$. Note that $\mu=2/d^2$ for $\mathrm{dim}\H=d$ if EWs are normalized, i.e., $\tr[W]=\tr[M]=1$. Note also that normalizing EWs has nothing to do with the ability to detect entangled states: for $W$ and $\alpha W$ for $\alpha >0$, it holds $D_W = D_{\alpha W}$. 

\begin{proposition}
    Mirrored EWs detect distinct sets of entangled states, i.e.,
\bea
D_{W}\cap D_{M} = \emptyset. \label{eq:emp}
\eea
\end{proposition}

\begin{proof}
From Eq. (\ref{eq:mr}), $\tr[W\rho] + \tr[M\rho] = {\mu \geq 0 }$ for a state $\rho$. This means that if $\rho \in D_W$, we have $\rho\notin D_M$. Similarly, if $\rho \in D_M$ we have $\rho\notin D_W$.
\end{proof}

From the relation in Eq. (\ref{eq:mr}), non-trivial upper bounds to mirrored EWs can be found,
\bea
\mu \geq \tr[W\sigma_{\sep}] \geq 0,~~\mathrm{and}~~\mu \geq \tr[M\sigma_{\sep}] \geq 0, \label{eq:lulu}
\eea
for all separable states $\sigma_{\sep}\in S_{\sep}$. Note that entangled states detected by the upper bound of a witness $W$ violate the lower bound of its mirrored one $M$. Similarly, entangled states detected by the lower bound of a witness $W$ do not satisfy the upper bound of its mirrored one. Thus, it holds that,
\bea
D_{W}^{(U)} = D_{M}^{(L)}~\mathrm{and}~D_{W}^{(L)} = D_{M}^{(U)}. \nonumber
\eea
We also recall that all entangled states detected by upper and lower bounds of a witness $W$ are denoted by $D_{W}^{(M)} = D_{W}^{(U)}\cup D_{W}^{(L)}$. \\ 

\begin{figure}
    \centering
    \includegraphics[width=0.52 \textwidth]{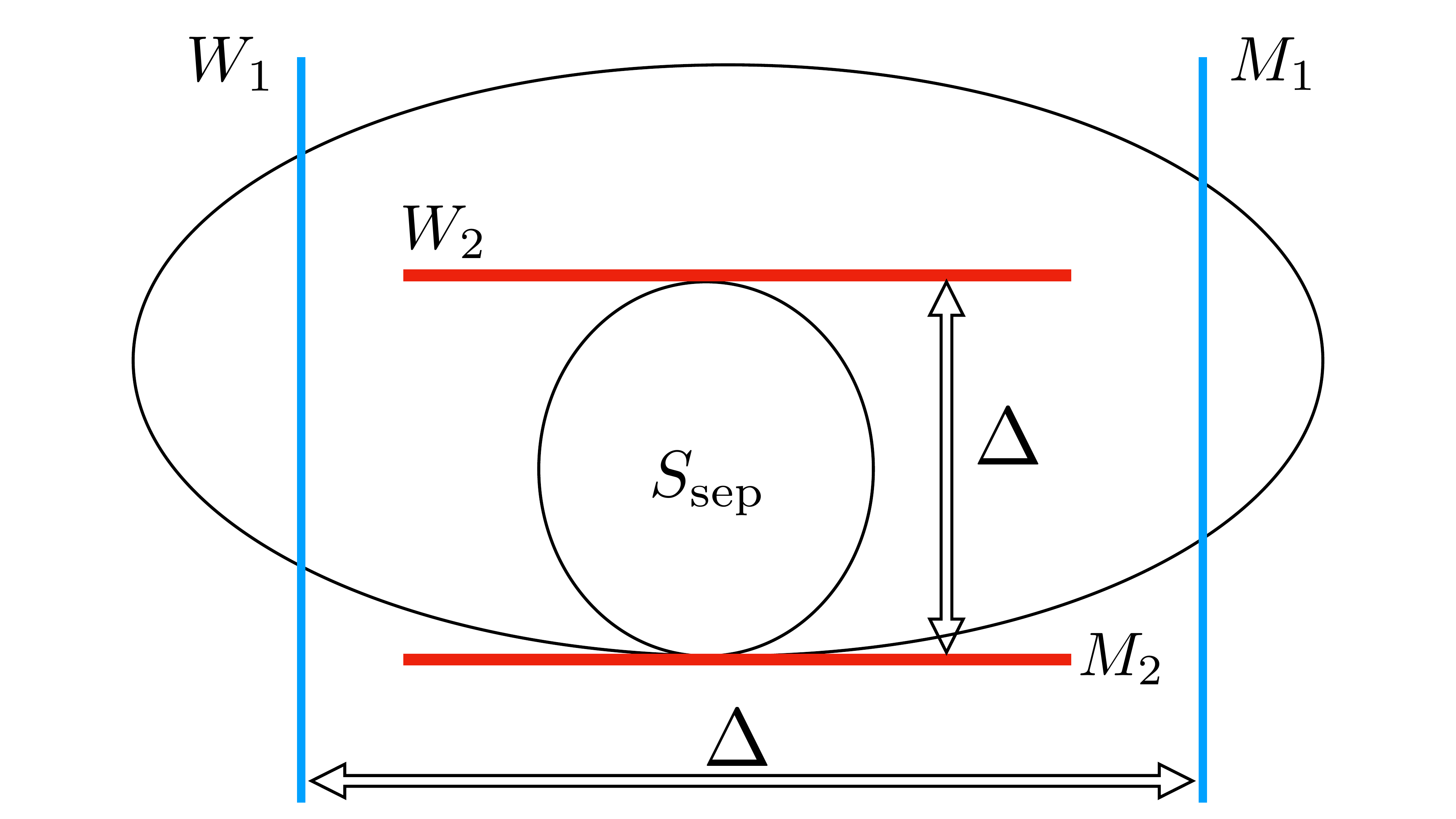}
    \caption{ Mirrored EWs for two-qubit states from Examples 1 and 2 are depicted. (Blue) A pair of non-optimal EWs $W_1$ and {$M_1$} are mirrored with each other, see Eqs. (\ref{eq:w2}) and (\ref{eq:m2}) in Example 1. They detect distinct sets of entangled states. (Red) A mirrored operator to an optimal EW is not an EW, see Eqs. (\ref{eq:wopt2}) and (\ref{eq:mopt2}) in Example 2. In both cases, the separability window is denoted by $\Delta = [0,1/2]$. }
    \label{fig:mr}
\end{figure}

{\it Example 1}.  An EW for detecting a Bell state $|\phi^+\rangle$ can be constructed as follows, 
\begin{align}
    W &= \frac{1}{4} \mathbb{I}
            - \frac{1}{4} \bigl( X^{(1)} X^{(2)} + Z^{(1)} Z^{(2)} \bigr) \nonumber \\
        &= \frac{1}{4} \mathbb{I}
            - \frac{1}{2} \bigl( \dyad{\phi^+} - \dyad{\psi^-} \bigr) \nonumber \\
        &= \frac{1}{4} \begin{pmatrix}
            0 & 0 & 0 & -1 \\
            0 & 2 & -1 & 0 \\
            0 & -1 & 2 & 0 \\
            -1 & 0 & 0 & 0
        \end{pmatrix}, \label{eq:w2}  
\end{align}
where $|\psi^-\rangle = (|01\rangle - |10\rangle)/\sqrt{2}$. Note that $W$ is normalized and thus $\mu=1/2$ in Eq. (\ref{eq:mr}) since $d=2$. It follows that
\bea
\frac{1}{2}\geq \tr[W\sigma_{\sep}]\geq 0,~~\mathrm{for}~~\sigma_{\sep}\in S_{\sep}. 
\nonumber 
\eea
It is straightforward to find its mirrored EW,
\begin{align}
    M &= \frac{1}{4} \mathbb{I}
            + \frac{1}{4} \bigl( X^{(1)} X^{(2)} + Z^{(1)} Z^{(2)} \bigr)\nonumber  \\
        &= \frac{1}{4} \mathbb{I}
            + \frac{1}{2} \bigl( \dyad{\phi^+} - \dyad{\psi^-} \bigr) \nonumber \\
        &= \frac{1}{4} \begin{pmatrix}
            2 & 0 & 0 & 1 \\
            0 & 0 & 1 & 0 \\
            0 & 1 & 0 & 0 \\
            1 & 0 & 0 & 2
        \end{pmatrix}, \label{eq:m2} 
\end{align}
which detects another Bell state $|\psi^{-}\rangle$. Note that EWs satisfy
\bea
W+M = \frac{1}{2} \I \otimes \I. \nonumber
\eea
Similarly, mirrored EWs can be constructed for other pairs of Bell states.\\

{\it Example 2}. An illuminating example is an optimal EW
\bea
W  = |\psi^-\rangle \langle \psi^-|^{\Gamma}  = \frac{1}{2}\I\otimes \I - |\phi^+\rangle \langle \phi^+|. \label{eq:wopt2} 
\eea
The mirrored operator can be found from Eq. (\ref{eq:mr}) with $\mu=1/2$,
\bea
M = |\phi^+\rangle \langle \phi^+|, \label{eq:mopt2} 
\eea
which is not an EW.  {Examples are depicted in Figure \ref{fig:mr}.  EWs for two-qubit states are decomposable, in which mirrored EWs to optimal EWs correspond to non-negative operators, that is, no longer valid EWs \cite{Bera:2023aa}. \\

For mirrored EWs, the separability window is defined as a range between the lower and upper bounds of a normalized witness $W$. From Eq. (\ref{eq:lulu}), assuming that a witness $W$ is normalized, the separability window is given as
\bea
\mathrm{separability~window:~} \Delta= [0,\mu].  \label{eq:swin} 
\eea
The caveat when computing the separability window is that although a witness $W$ is normalized, its mirrored one is not necessarily normalized. One can also rephrase the separability window as an upper bound to a normalized EW.

An optimization of mirrored EWs and the separability window is related as follows. For a non-optimal EW, let us apply an optimization, 
\bea
W^{'} = W - \epsilon P \nonumber 
\eea
for some $\epsilon\geq 0$ and non-negative operator $P\geq 0$. Let $M$ denote the mirrored operator to $W$. From Eq. (\ref{eq:mr}), it follows that $W^{'} + \epsilon P + M = \mu \I \otimes \I$. This means that the mirrored operator to $W^{'}$ is given by
\bea
M^{'} = M + \epsilon P, \nonumber 
\eea
while the separability window is kept, i.e., 
\bea
W^{'} + M^{'} =\mu\I\otimes \I. \nonumber 
\eea
That is, $W^{'}$ is finer than $W$ whereas $M$ is finer than $M^{'}$, or vice versa. Hence, while the separability window is kept, optimization of an EW decreases the capability of its mirrored EW in detecting entangled states, see Fig. \ref{fig:mropt}. 

{Note also that a separability window $\mu$ can be chosen such that an EW has its mirror as a valid EW. For all $|ef\rangle \in \H \otimes \H$, we have
\bea
\langle ef | M | ef\rangle = \langle ef| ( \mu \I \otimes \I - W) |ef\rangle \geq 0. \nonumber
\eea
A mirrored operator $M$ can be a valid EW if there exists an entangled state $\tau$ such that $\tr[M\tau]<0$, that is, $\tr[  (\mu\I\otimes \I - W)\tau] <0$. From two conditions, it holds that
\bea
\max_{|ef\rangle} \langle e f | W | e f \rangle \leq \mu < \tr [W\tau]. \label{eq:cmu}
\eea
Note that an entangled state $\tau$ is not detected by $W$ but by $M$; thus, $\tr[W\tau ] \geq 0$. let us summarize the result as follows.} 

{
\begin{proposition}
For a witness $W$, if there exists an entangled state $\tau$ not detected by $W$ such that it satisfies the following,
\bea
\max_{|ef\rangle} \langle e f | W | e f \rangle <  \tr [W\tau] \label{eq:coon}
\eea
then a separability window $\mu$ can be found such that a mirrored operator $M$ is a valid EW.  
\end{proposition}
}

{A witness $W$ has its mirrror as a positive operator if Eq. (\ref{eq:coon}) does not hold. }

\begin{figure}
    \centering
    \includegraphics[width=0.52 \textwidth]{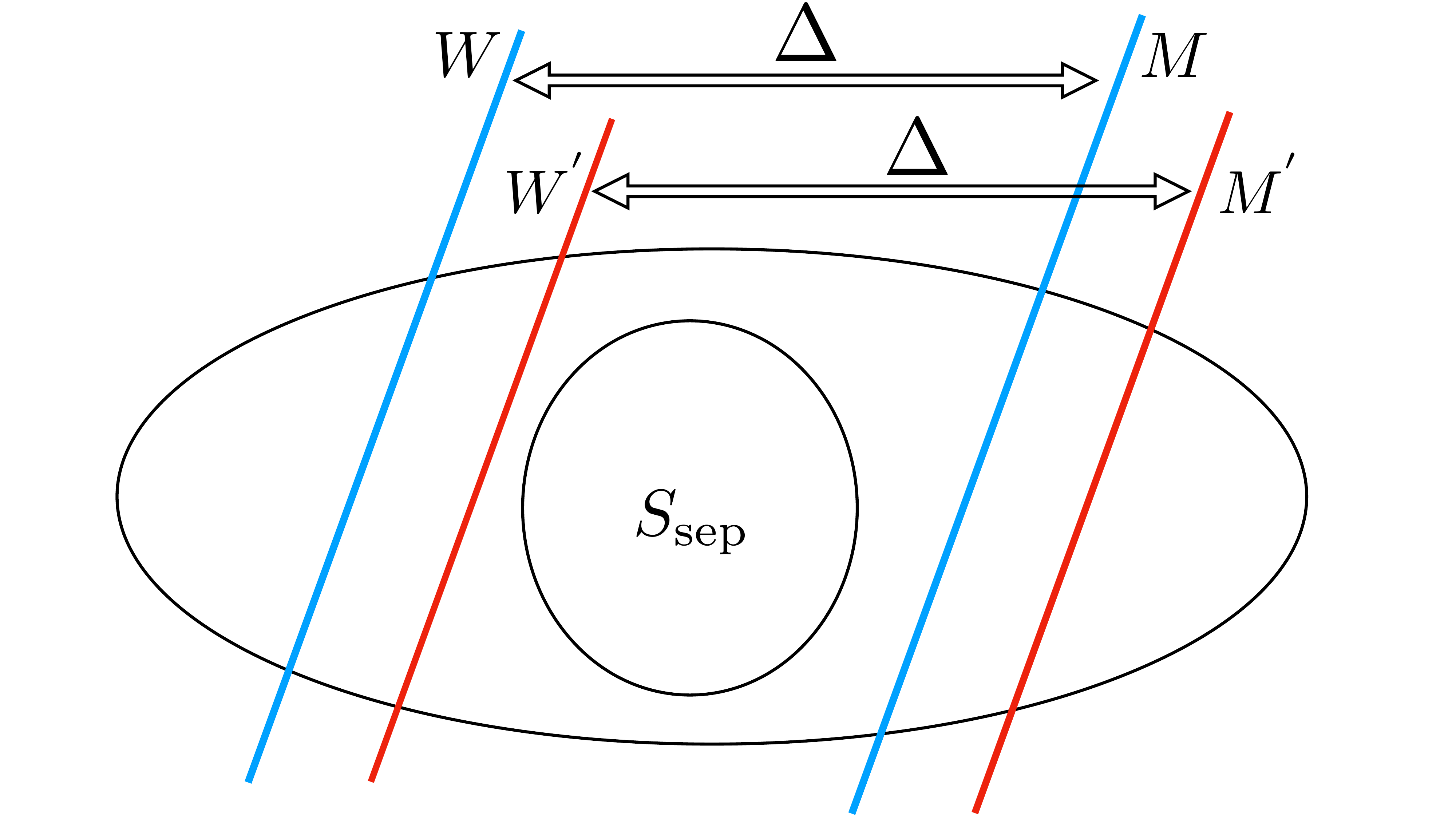}
    \caption{ A witness $W$ is optimized by subtracting a positive operator, and thus $W^{'}$ is finer than $W$. The consequence is that its mirror detects a smaller set of entangled states, $M$ is finer than $M^{'}$, while the separability window is kept constant. }
    \label{fig:mropt}
\end{figure}

\subsection{Structural physical approximation}

The structural physical approximation (SPA) transforms positive maps into quantum operations. It can be equivalently formulated as a transformation of EWs to positive operators, which can be interpreted as quantum states. For convenience, we here omit the normalization of resulting positive operators, although EWs are normalized, and write the SPA as follows, 
\bea
\mathrm{SPA}:~W\mapsto \widetilde{W} = W + p \I\otimes \I, \nonumber 
\eea
with minimal $p$ such that $\widetilde{W}\geq0$. The SPA admixes an identity operator, which is of full rank, to wash out negative projections of an EW so that the resulting operator is non-negative for all states. Since optimal EWs tightly characterize the set of separable states, SPA of optimal EWs are closely connected to separable states \cite{PhysRevA.78.062105, Bae_2017, ende2025simplesufficientcriteriaoptimality}. 

Note that negative projections are used to detect entangled states, while positive ones are used to maintain non-negative expectation values for separable states; for a decomposition $W = W_+ - W_-$ where $W_{\pm}\geq 0$, a witness $W$ may detect entangled states with a negative projection $W_-$. 

Let us consider another witness $M$ and utilize its positive projections for detecting entangled states, i.e., $-M$. We introduce the mirrored SPA as follows, 
\bea
\mathrm{mSPA}:~M\mapsto \widetilde{M} = -M + q \I\otimes \I, \nonumber 
\eea
with a minimal $q$ such that $\widetilde{M}\geq 0$
Note that SPA and mSPA can be applied to arbitrary EWs.

{We thus arrive at the following derivations.}
\begin{itemize}
    \item For an EW, i.e., $W=M$, one applies the SPA and the mSPA. It holds that 
    \bea
    \widetilde{W} + \widetilde{M} = (q+p) \I \otimes \I. \nonumber 
    \eea
    \item For two EWs, the SPA of one and the mSPA of the other coincide, i.e., $\widetilde{W} = \widetilde{M}$. Then, two EWs are mirrored with each other
    \bea
    W + M  = (q-p) \I \otimes \I. \nonumber 
    \eea
\end{itemize}
Note that $q>p$ in general, and thus the separability window is given by $\Delta = [0,q-p]$. Both $W$ and $M$ have an upper bound $\mu = q-p$ satisfied by all separable states.

\subsection {Mirrored EWs as observables: POVM clouds}

One of the lessons from mirrored EWs is that it is, in fact, handy to construct a scheme for verifying entangled states without referring to positive but not completely positive maps directly. Let us begin with an arbitrary observable $\mathbbm{O} \in \B(\H\otimes \H)$. That is, one can describe $\mathbbm{O} = \sum_{ij} c_i \Pi_i \otimes \Pi_j$ with POVM elements $\{\Pi_i \}$, $\{\Pi_j \}$ and $\{c_i \}$ so that it can be estimated in an experiment. 

For an operator $\mathbbm{O}$, we compute,
\bea
u_{\mathbbm{O}} & = & \max_{\sigma\in S_{\sep}} \tr[\mathbbm{O} \sigma], ~~\mathrm{and}~~ l_{\mathbbm{O}}=\min_{\sigma\in S_{\sep}} \tr[\mathbbm{O} \sigma],  \label{eq:olu}
\eea
where it suffices to run the optimization over pure states. With the bounds above, we construct a detection scheme with lower and upper bounds satisfied by all separable states,
\bea
u_{\mathbbm{O}} \geq \tr[\mathbbm{O} \sigma_{\sep}]\geq l_{\mathbbm{O}}, ~~~\forall \sigma_{\sep} \in S_{\sep}. \nonumber 
\eea
The detection scheme above is equivalent to that of mirrored EWs in Eq. (\ref{eq:mr}) by deriving EWs as follows,
\bea
W = \mathbbm{O} - l_{\mathbbm{O}}\mathbbm{I} \otimes \mathbbm{I}, ~~\mathrm{and}~~M = u_{\mathbbm{O}} \mathbbm{I} \otimes \mathbbm{I} - \mathbbm{O}, \nonumber 
\eea
which are mirrored with each other since it holds that
\bea
W + M = (u_{\mathbbm{O}} - l_{\mathbbm{O}})\mathbbm{I\otimes I}, \label{eq:ms} 
\eea
where $\mu = u_{\mathbbm{O}} - l_{\mathbbm{O}}>0$. Hence, mirrored EWs are obtained. We remark that EWs can be constructed with optimizations in Eq. (\ref{eq:olu}). 

Note that the construction of the detection scheme in Eq. (\ref{eq:olu}) generally applies to resource theories in which free resources form a convex set \cite{RevModPhys.91.025001}. In fact, the derived mirrored structure in Eq. (\ref{eq:ms}) exploits the only fact that the set of separable states is convex.

\begin{figure}[t]
\centering
\includegraphics[width=0.45\textwidth]{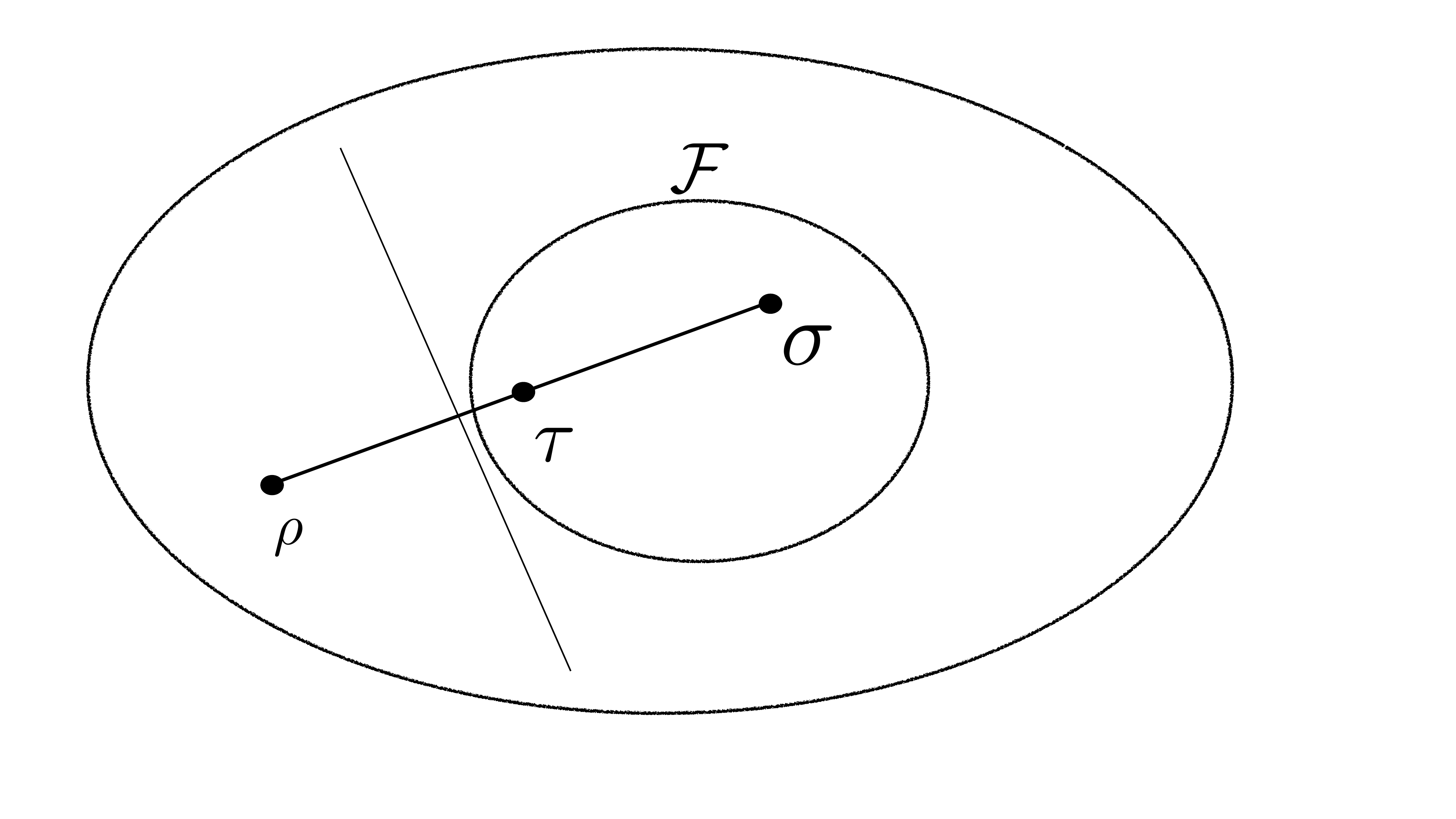}
\caption{ {A resource state $\rho\notin \mathcal{F}$ can be transformed to a free one by admixing a free state $\sigma\in\mathcal{F}$ such that $\tau = (1-\mu) \rho +\mu \sigma \in \mathcal{F}$ for $\mu\in (0,1)$ and $\mu$ may give a quantification, which can be equivalently described by witnesses} \cite{PhysRevX.9.031053}.  }
\label{fig:resource}
\end{figure}

\subsection{Generalization to resource theories }

{ Resource theories generalize the structure of the entanglement theory by identifying free resources and free operations to separable states and local operations and classical communication (LOCC) \cite{RevModPhys.91.025001}; entangled states are resourceful. Free resources from a convex set, analogously to separable states. Then, resources that are not free may be quantified by finding how far they are from the convex set of free resources. Since free resources are convex, the Hahn-Banach theorem tells that there exist separating hyperplanes that distinguish resources from the convex set of free ones. }

{So far, the generalized robustness, also analogous to the robustness of entanglement \cite{PhysRevA.59.141, PhysRevA.67.054305}, applies to the quantification in resource theories. Robustness measures in a resource theory are also equivalently described by witnesses separating resources from free ones \cite{PhysRevX.9.031053}; the result is analogous to the entanglement theory, where EWs are converted to the generalized robustness \cite{PhysRevA.72.022310, Eisert_2007}. Note also that witnesses in resource theories are closely connected to discrimination tasks \cite{PhysRevX.9.031053, Bae_2015, PhysRevLett.122.140404}.  } \\


\subsection{EWs for graph states}

In this subsection, we summarize graph states, an important class of multipartite entangled states.

\subsubsection{Graph states}
To introduce graph states, let $G = (V, E)$ denote a graph consisting of a set $V$ of $n$ vertices and a collection of edges $E$ connecting pairs of these vertices. For a vertex $i$ in the graph $G$, we write its neighborhood $I_i(G)$ as the set of vertices directly connected to the vertex. For a vertex $i$, we introduce an operator $g_i$,
\begin{equation}
g_i = X^{(i)} \bigotimes_{j \in N(i)} Z^{(j)} \nonumber
\end{equation}
An $n$-qubit graph state $\ket{G}$ is defined as an eigenstate, $g_i \ket{G} = \ket{G}$ for $i = {\mu}, \ldots, n$ with eigenvalue $+1$. 

The {stabilizer} of the graph state is defined as,
\begin{equation}
\mathcal{S} (G) = \left\{ s_j \right\}~~\mathrm{where}~~
\quad s_j = \prod_{i \in I_j(G)} g_i, \\ ,
\label{stabdef}
\end{equation}
for a subset of the vertices $I_j(G)$. If a generator $g_k$ appears in the product defining $s_j$, i.e., $k \in I_j(G)$, we say that $s_j$ {contains} $g_k$. It is clear that $s_j \ket{G} = \ket{G}$.

A useful property of graph states and their stabilizers is that the projector onto $\ket{G}$ can be expressed as the sum of all stabilizer elements:
\begin{equation}
\frac{1}{2^n}\sum_{i=1}^{2^n} s_i = \ketbra{G}.
\label{stab}
\end{equation}
Note that the stabilizer $\mathcal{S}(G)$ forms a commutative group and has $2^n$ \cite{Toth2005}. 

 \subsubsection{ Two-colorable graph states} 

Two-coloring of a graph, or a bipartite coloring, is a way to color the vertices of a graph using two colors, which we write by red and blue, such that no two adjacent vertices share the same color. A graph is considered two-colorable if such a coloring exists. We now introduce {stabilizer projectors}, which are in turn employed in the construction of EWs for two-colorable graph states. 


\begin{figure}[t]
\centering
\includegraphics[width=0.45\textwidth]{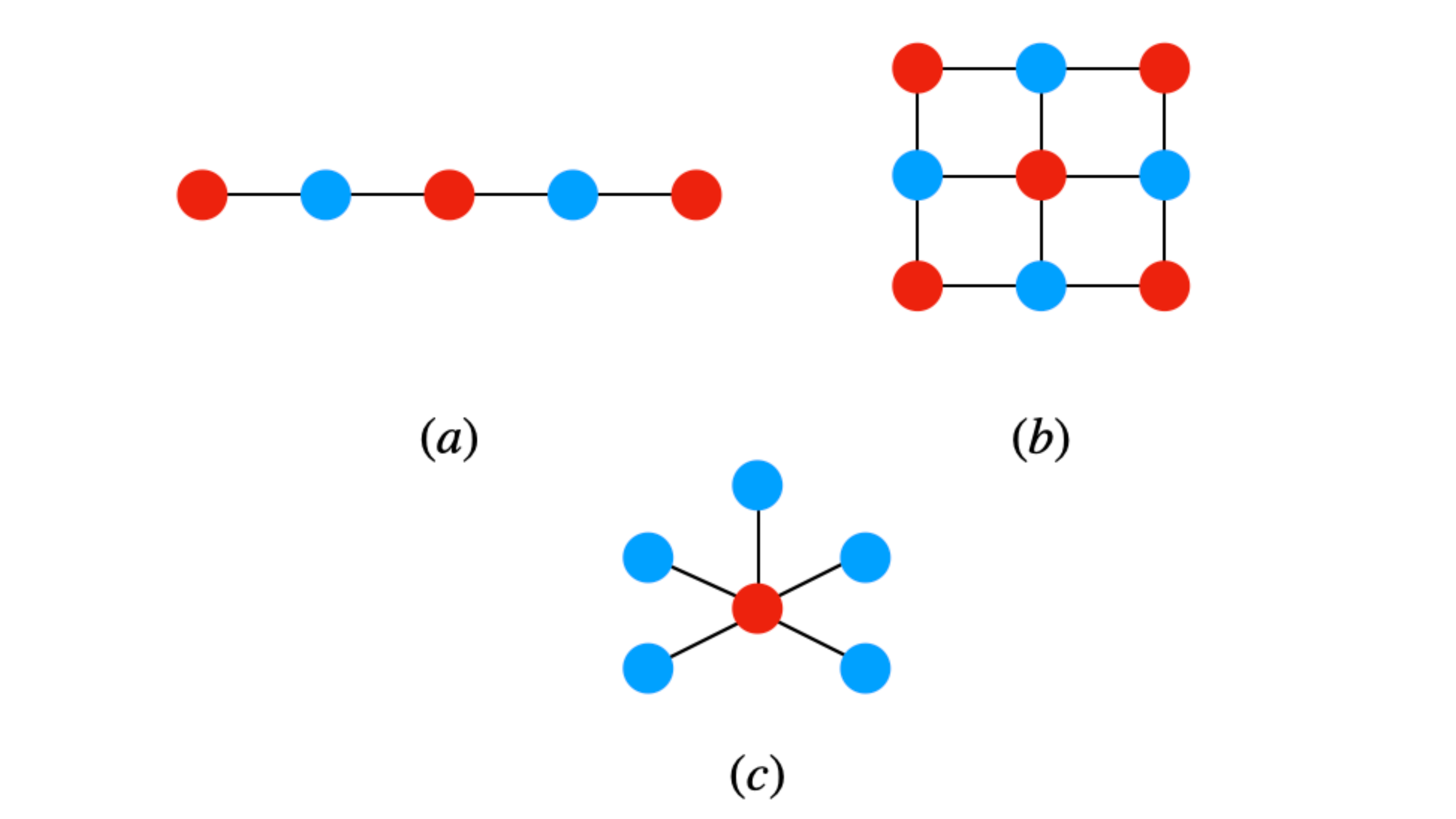}
\caption{Examples of two-colorable graph states are shown: (a) a one-dimensional cluster state and (b) a two-dimensional cluster state (c) a GHZ state. }
\label{fig:2c}
\end{figure}

Examples of two-colorable states include one-dimensional cluster states, two-dimensional cluster states, $n$-partite GHZ states, see Fig. \ref{fig:2c}.

\subsubsection{ Constructing EWs}

For graph states, we consider instances of EWs, called canonical, alternative, and two-measurement ones. 

Suppose that one wants to construct an EW for detecting a particular $n$-partite entangled state $\ket{\psi}$. A straightforward construction may be found as follows,
\be \label{canEW}
 W  = \alpha\, \id^{\otimes n} - \ket{\psi}\bra{\psi}, \nonumber 
\ee
where $\alpha$ may be computed such that an EW is non-negative for all separable states. For instance, graph states $\ket{G}$, which are genuinely multipartite entangled, have $\alpha \geq 1/2$. An EW constructed in this manner is called a canonical EW. For an $n$-partite stabilizer state, we also have $\alpha\geq 1/2$. EWs can be expressed in terms of generators~\cite{Muller:PhysRevA.101.012317},
\be
\mathrm{canonical~EW:}~W_c = \frac{1}{2}\, \id^{\otimes n} - \prod^n_{i=1} \frac{\id^{\otimes n} + g_i}{2} \label{eq:can}
\ee
Note that an EW above can be estimated by measurements on $2^n - 1$ stabilizers. 

Measurement settings can be reduced in EWs in the following~\cite{Guhne:PhysRevA.72.022340, Guhne:PhysRevLett.94.060501, Muller:PhysRevA.101.012317}:
\be 
\mathrm{alternative~EW:}~ W_a = \frac{n - 1}{2}\, \id^{\otimes n} - \frac{1}{2} \sum^n_{i=1} g_i, \label{altEW}
\ee
which are called alternative EWs. These can detect genuinely multipartite entangled states. An alternative EW can be estimated by measurements on $n$ generators.

For 2-colorable graph states, two-measurement EWs can be defined as follows~\cite{Muller:PhysRevA.101.012317}, 
\bea 
&&\mathrm{two}-\mathrm{measurement~EWs:}~ \nonumber \\
&& W_{2m} = \frac{3}{2}\, \id^{\otimes n} - \prod_{{\rm red \;}i} \frac{\id^{\otimes n} + g_i}{2} - \prod_{{\rm blue \;}i} \frac{\id^{\otimes n} + g_i}{2},\label{2mEW}
\eea
where the vertices are colored in two different colors, red and blue, corresponding to two measurement settings. {For $n=2$, substituting $g_i$ into Eq.~(26) yields that only the two settings $XX$ and $ZZ$ are required.
In general, $^{11}$} For example, two-measurement EWs for GHZ states can be computed by measuring $X$ on each qubit and $Z$ again on each qubit, i.e., $X^{(1)} \otimes X^{(2)} \otimes \cdots \otimes X^{(n)}$ and $Z^{(1)} \otimes Z^{(2)} \otimes \cdots \otimes Z^{(n)}$. \\

\section{ Multipartite systems}
\label{sec:multi}

In this section, we construct mirrored EWs for multipartite quantum systems and investigate their properties. We consider various mirrored EWs for graph states and compare their separability windows. We also present a mirrored pair of optimal EWs for three-qubit states. 

\subsection{ $n$-partite GHZ states}

Let us begin with $ n$-partite GHZ states, denoted by
\begin{equation}
    \ket{GHZ_n} = \frac{1}{\sqrt{2}} \left( \ket{0}^{\otimes n} + \ket{1}^{\otimes n} \right). \label{eq:nghz}
\end{equation}
which are also instances of two-colorable states. An $n$-partite GHZ state is stabilized by the following set of generators:
\bea
g_1 = X^{(1)} X^{(2)} \cdots X^{(n)}~\mathrm{and}~  g_i = {Z^{(i-1)} Z^{(i)}} \nonumber 
\eea  
for $i=2,\cdots, n-1$. Note that $g \ket{GHZ_n} = \ket{GHZ_n}$ for all generators $g$. Thus, it also follows that 
\bea
    \dyad{GHZ_n} = \prod_{i=1}^{n} \left( \frac{\id^{\otimes n} + g_i}{2}\right).
\eea
The projection onto a GHZ state is uniquely determined by generators. In what follows, we derive mirrored EWs for $n$-partite states and compare separability windows.

\subsubsection{Mirrored EWs}

 Firstly, we derive a mirrored EW for the canonical one from Eq. (\ref{eq:can}), 
\bea
W_c &=& \frac{1}{2^{n-1}-1} \left( \frac{1}{2} \id^{\otimes n} - \dyad{GHZ_n} \right), \nonumber \\
M_c &=& \dyad{GHZ_n}, \nonumber 
\eea
where the mirrored operator is not an EW. Note that a witness {$W_c$} is normalized. The separability window is computed as $\Delta = [0,\mu_c]$ where
\bea
\mu_c = \frac{1}{2^n -2}, \label{eq:delc}
\eea
which quickly converges to $0$ as $n$ tends to be large. 

Secondly, let us consider alternative EWs from Eq. (\ref{altEW}) and compute their mirrored EWs 
\bea
    W_a  &=& \frac{1}{2^n} \left( \id^{\otimes n} - \frac{1}{n-1} \sum^{n}_{i=1} g_i \right),~~\mathrm{and} \nonumber \\
    M_a   &= &\frac{1}{2^n} \left( \id^{\otimes n} + \frac{1}{n-1} \sum^{n}_{i=1} g_i \right), \nonumber
\eea
One can find that the mirrored one {$M_a$  } detects GHZ states equivalent to the state in Eq. (\ref{eq:nghz}) up to local unitaries, 
\bea
    \ket{ GHZ_{n}^{(a)} } = \frac{1}{\sqrt{2}} \bigl( \ket{01010\cdots x} - \ket{10101\cdots (x\oplus 1)} \bigr) ~~~~~~\label{eq:nghza}
\eea
where $x=1$ if $n$ is even and $x=0$ otherwise, and $\oplus$ denotes the bitwise addition. Note that two states $|GHZ_n\rangle$ and {$\ket{ GHZ_{n}^{(a)} }$} are connected by local unitaries, 
\bea
&&    \id^{(1)} \otimes Y^{(2)} \otimes \cdots \otimes Y^{(n-1)} \otimes \id^{(n)} ~~\mathrm{for~n~odd} ~ \nonumber \\
\mathrm{and} &~~& 
    \id^{(1)} \otimes Y^{(2)} \otimes \cdots \otimes \id^{(n-1)} \otimes Y^{(n)} ~~\mathrm{for~n~even}. 
 \nonumber 
\eea
Since a witness {$W_a$} is normalized, we have a separability window $\Delta_c = [0,\mu_a]$ where
\bea
\mu_a = \frac{1}{2^{n-1}} \label{eq:dela}
\eea
which approaches zero rapidly as $n$ increases.

Thirdly, we compute mirrored operators for two-measurement EWs, see Eq. (\ref{2mEW}) {
\bea
W_{2m} &=& \frac{1}{2^n - 2} \left(
            \frac{3}{2} \id^{\otimes n}
            - \frac{\id^{\otimes n} + g_1}{2}
            - \prod_{i=2}^{n} \frac{\id^{\otimes n} + g_i}{2} \right),~\mathrm{and} \nonumber \\
M_{2m} &=& \frac{1}{2^n + 2} \left( \frac{1}{2} \id^{\otimes n} + 
            \frac{\id^{\otimes n} + g_1}{2}
            + \prod_{i=2}^{n} \frac{\id^{\otimes n} + g_i}{2}
        \right) \nonumber 
\eea }
where the mirrored operator is positive semidefinite. Note that a witness {$W_{2m}$} is normalized and thus the separabilty is computed as $\Delta_{2m} = [0,\mu_{2m}]$ where 
{
\bea
\mu_{2m} = \frac{1}{2^{n-1}} \label{eq:mu2m}
\eea
}
which quickly converges to zero as $n$ tends to be large. 
 
\begin{figure}
    \centering
    \includegraphics[width=0.52 \textwidth]{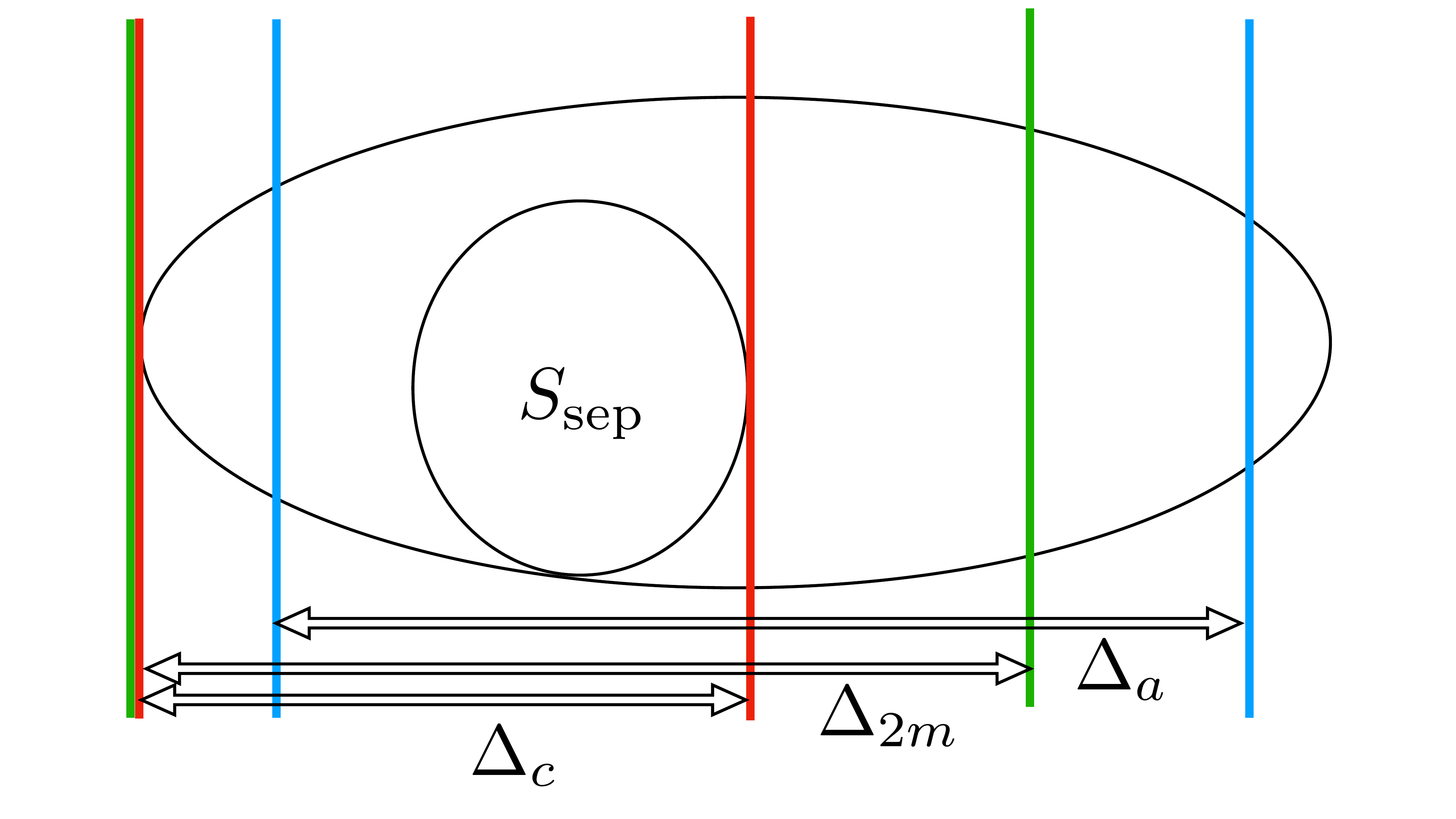}
    \caption{Separability windows of mirrored EWs are compared, see Eq. (\ref{eq:delh}). An optimal EW having its mirror as a positive operator has the smallest separability window (red). Two-measurement (green) and alternative (blue) EWs, which are not optimal, have larger separability windows, $\Delta_{2m}<\Delta_a$. Note that a mirrored pair of EWs can be constructed for alternative EWs. Two alternative EWs in a mirrored pair are equivalent up to local unitaries. }
    \label{fig:windows}
\end{figure}

\subsubsection{ Comparison of the separability window}

Mirrored EWs with a smaller separability window may detect a larger set of entangled states. The three types of EWs above can be compared in terms of the separability window. From Eqs. (\ref{eq:delc}), (\ref{eq:dela}), and (\ref{eq:mu2m}). Overall, the separability windows satisfy the hierarchy
{
\bea
\Delta_c \leq \Delta_{2m} = \Delta_a, \label{eq:delh}
\eea
}
where the equality holds for $n=2$, see Fig. \ref{fig:windows}.

\begin{figure}
    \centering
    \includegraphics[width=0.47 \textwidth]{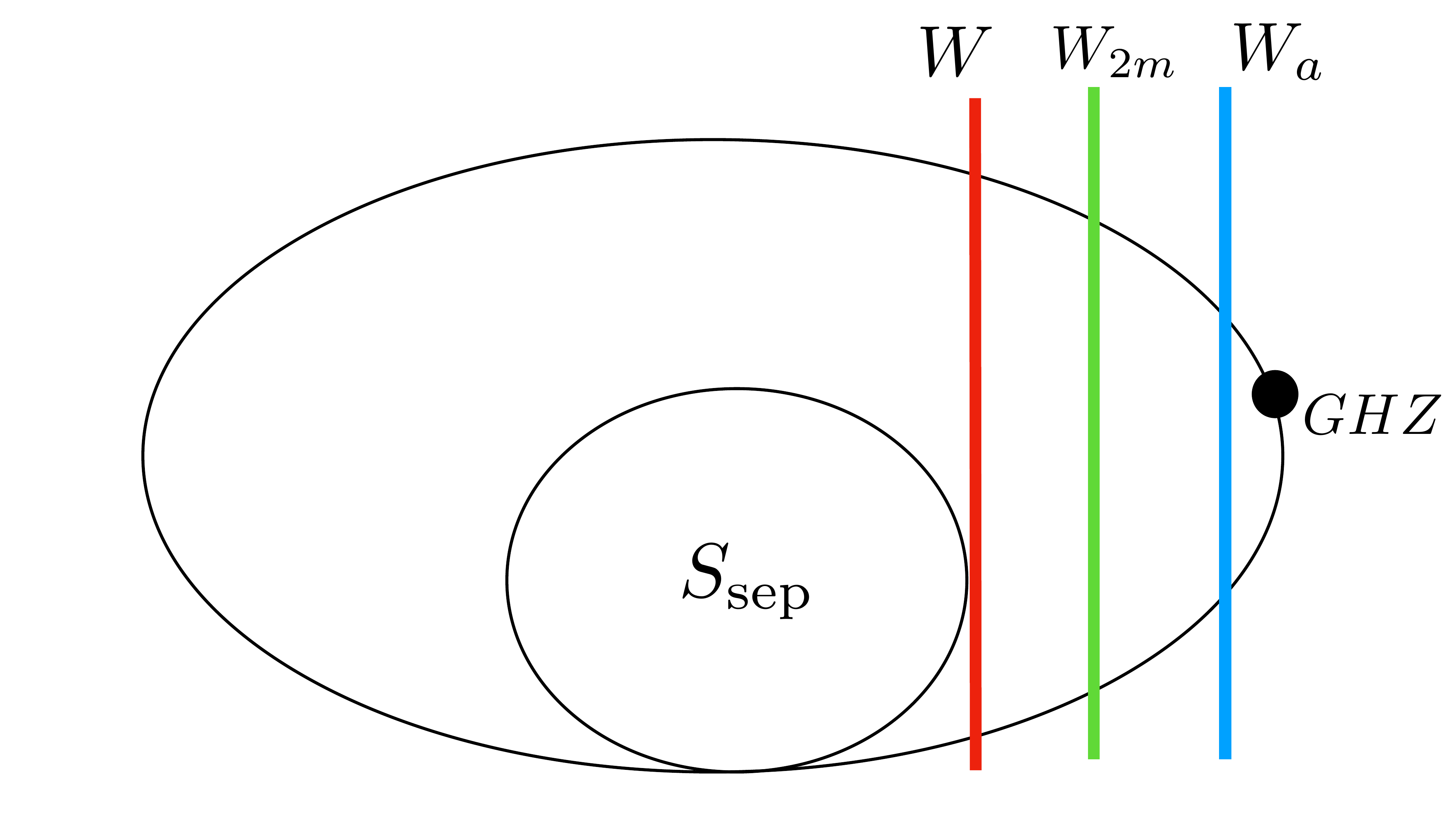}
    \caption{ Negative expectation values show the robustness of EWs against noise on quantum systems, see the comparisons in Eqs. (\ref{eq:noc}), (\ref{eq:noa}), and (\ref{eq:no2m}). }
    \label{fig:ewu}
\end{figure}

\subsubsection{ Mirrored EWs against local unitaries}

Noise on quantum systems generally decreases the entanglement existing in a quantum state. It can be classified into i) noise due to interaction with an environment and ii) unwanted changes of local basis. The former decreases the purity of quantum states, and so does entanglement. The latter, however, keeps the entanglement. Nonetheless, EWs are robust to noise in the first type but not the latter one, to which mirrored EWs provide a partial solution. 

The canonical witness $W_c$ cannot detect a GHZ state in Eq. (\ref{eq:nghza}). Instead, one can take an alternative EW, $W_a$, in Eq. (\ref{altEW}). Estimating an alternative EW also finds its mirrored EW, which can detect GHZ state in a different local basis in Eq. (\ref{eq:nghza}). That is, we have a limited capability with a canonical EW, 
\bea
&& GHZ_n \in D_{W_c}, ~~\mathrm{but} ~~ GHZ_{n}^{(a)}  \notin D_{W_c},
\nonumber
\eea
whereas mirrored EWs exploiting both upper and lower bounds can detect a larger set of entangled states, 
\bea
&&  GHZ_n, ~GHZ_{n}^{(a)} \in D_{W_{a}}^{(U)}\cup D_{W_{a}}^{(L)}. ~ \nonumber 
\eea
Therefore, mirrored EWs offer a strategy of detecting entangled states under rotations of local unitaries. The advantage makes them more suitable for practical entanglement detection in the presence of local unitary noise.

Note also that alternative and two-measurement EWs having non-trivial mirrored operators are less robust against noise due to unwanted interactions with the environment. For an $n$-partite GHZ state in Eq. (\ref{eq:nghz}), we have
\bea
\tr[W_c |GHZ_n\rangle \langle GHZ_n|] &= &  \frac{-1}{2^{n}-2}, \label{eq:noc} 
\eea
which means that $W_c$ provides a negative expectation value up to a noise fraction less than $(2^n-2)^{-1}$. {Additionaly, witness $W_c$ can tolerate unitary rotation error for GHZ state whenever $p < p^{(c)}_U$ where
\bea
    p^{(c)}_U = \frac{1}{2}.
\eea
since
\bea
\| W_c \|_\opn = 1/[2(2^{n-1}-1)].
\eea}
We also find that 
\bea
\tr[W_a |GHZ_{n}\rangle \langle GHZ_{n}|] &= &  \frac{-1}{ (n-1)2^n}, \label{eq:noa} \\
\mathrm{and}~~\tr[W_{2m} |GHZ_{n} \rangle \langle GHZ_{n}  |] &= &   \frac{-1}{2(2^{n}-2)}, \label{eq:no2m} 
\eea
which shows that the other EWs can tolerate smaller fractions of noise, see Fig. \ref{fig:ewu}. It is also straightforward to find that
\bea
\tr[M_a |GHZ_{n}^{(a)}\rangle \langle GHZ_{n}^{(a)} |] 
& = &\frac{-1}{(n-1)2^n}, \nonumber 
\eea
which is equal to Eq. (\ref{eq:noa}).

{For unitary errors, witness $W_a$ and $M_a$ can tolerate such errors for state $\ket{GHZ_n}$ and $\ket{GHZ^{(a)}_n}$, respectively, whenever $p < p^{(a)}_U$ where
\bea
    p^{(a)}_U = \frac{1}{2n}.
\eea
since 
\bea
\| W_a \|_\opn = \| M_a \|_\opn = n/[(n-1)2^n].
\eea 
Witness $W_{2m}$ can tolerate unitary rotation error for GHZ state whenever $p < p^{(2m)}_U$ where
\bea
    p^{(2m)}_U = \frac{1}{6}.
\eea
since
\bea
\| W_a \|_\opn = 3/[2(2^n-2)].
\eea
}

\begin{table*}[t]
{
\centering
\small
\setlength{\tabcolsep}{9pt}
\renewcommand{\arraystretch}{1.5}
\begin{tabular}{l c c c c c}
\hline
\textbf{EW} & \textbf{Mirror} & \textbf{Optimality} & \textbf{Separability window ($\mu$)} & $\boldsymbol{p^\star}$ & $~~\quad\boldsymbol{p_U^\star}$ \\
\hline
\rule{0pt}{6ex}
canonical        & positive            & O \cite{Toth2005} & $\displaystyle \frac{1}{2^{n}-2}$     & $\displaystyle \frac{2^{\,n-1}}{2^{\,n}-1}$            & $~~\quad\displaystyle \frac{1}{2}$ \\
\rule{0pt}{6ex}
alternative      & EW (non-optimal)    & X   & $\displaystyle \frac{1}{2^{\,n-1}}$   & $\displaystyle \frac{1}{n}$                              & $~~\quad\displaystyle \frac{1}{2n}$ \\
\rule{0pt}{6ex}
two-measurement  & positive            & O \cite{Toth2005} & $\displaystyle \frac{1}{2^{\,n-1}}$   & $\displaystyle \frac{2^{\,n-1}}{3\cdot 2^{\,n-1}-2}$    & $~~\quad\displaystyle \frac{1}{6}$ \\
\rule{0pt}{1ex}
& & & & & \\
\hline
\end{tabular}
\caption{ Three constructions of EWs for detecting GHZ states are compared. For instance, a canonical EW has its mirror as a positive operator, i.e., not a valid EW. A separability window is denoted by $\mu$ where $\Delta=[0,\mu]$. Detection by an EW can tolerate depolarization noise up to $p^{\star}$ in Eq. (\ref{eq:white-noise-threshold}) and unitary noise up to $p_{U}^{\star}$ in Eq. (\ref{eq:unitary-threshold}). The optimality of the canonical and two-measurement EWs is shown in Ref. \cite{Toth2005}. 
Among the three constructions, canonical EWs have the smallest separability window and are the most robust against depolarization and unitary noise. Alternative EWs that are not optimal have non-trivial pairs. Two-measurement EWs are also optimal and share similar properties with canonical EWs. }
\label{tab:ghz-witness-summary}}
\end{table*}

\subsection{EWs for graph states}

So far, we have found that alternative EWs are useful to construct mirrored EWs. We further develop mirroring alternative EWs to detect graph states and, in particular, focus on alternative EWs for two-colorable graph states such as cluster states. Let $\ket{G}$ denote an $n$-qubit graph state $\{g_i\}_{i=1}^n$ generators, a product of Pauli operators $X$ and $Z$. An EW for a graph state $|G\rangle$ is given as follows \cite{Toth2005}
\begin{equation} 
    W = (n-1)\, \id^{\otimes n} - \sum_{i=1}^n g_i.\label{eq:graphew}
\end{equation}
Note that an EW above can detect genuinely entangled multipartite states \cite{Toth2005, zhou2019}.

We obtained its mirrored EW as follows,
\begin{equation} 
    M = (n-1)\, \id^{\otimes n} + \sum_{i=1}^n g_i.\label{eq:graphewm}
\end{equation}
The separability window is given by $\Delta = [0, 2^{1-n}]$. In what follows, we show that EWs in Eqs. (\ref{eq:graphew}) and (\ref{eq:graphewm}) are connected by local unitaries. Therfore, entangled states detected by them are local unitarily equivalent. 

\begin{proposition}
Mirrored alternative EWs in Eqs. (\ref{eq:graphew}) and (\ref{eq:graphewm}) are equivalent up to local unitaries. 
\end{proposition}
\begin{proof}
Let $G = (V, E)$ be a two-colorable graph on $n$ vertices. We write by subsets of vertices
\beas
V_{\text{even}} &=& \{ i \in V \mid \abs{N(i)} {\rm \; is \; even} \}, \\ V_{\text{odd}} &=& \{ i \in V \mid \abs{N(i)} {\rm \; is \; odd} \}, \nonumber 
\eeas
and also introduce local unitaries $U$ as follows, 
\be
U = \left( \bigotimes_{i \in V_{\text{even}}} Y^{(i)} \right)\otimes \left( \bigotimes_{j \in V_{\text{odd}}} X^{(j)} \right). \label{eq:graphu} 
\ee
A generator is given as 
\bea
g_i = X^{(i)} \bigotimes_{j \in N(i)} Z^{(j)}. \nonumber 
\eea
Note that $Y Z Y^\dag = -Z$ and $X Z X^\dag = -Z$. It holds that 
\bea
    U g_i U^\dag =  -g_i,\nonumber 
\eea
and consequently we have $U W U^\dag = M$.
\end{proof}
{
A witness $W$ in Eq. \eqref{eq:graphew} tolerates unitary noise for graph state $\ket{G}$ whenever $p_U < 1/(2n-1)$. In the same manner, its mirror $M$ tolerates unitary noise for graph state $\ket{G}$ whenever $p_U < 1/(2n-1)$.
}

Let us present mirrored EWs for instances of graph states. Firstly, we consider linear cluster states $\ket{C}$, for which  generators are given by,
 is given by:
\begin{align*}
    g_1 &= X^{(1)} Z^{(2)}, \\
    g_n &= Z^{(n-1)} X^{(n)}, \\
    g_i &= Z^{(i-1)} X^{(i)} Z^{(i+1)}, \quad \text{for } i = 2, \dots, n-1.
\end{align*}
From Eq. (\ref{eq:graphu}), mirrored EWs are connected by a local unitary transform $U_C$, 
\begin{equation}
    U_C = X^{(1)} \otimes X^{(n)} \bigotimes_{k=2}^{n-1} Y^{(k)}.\nonumber
\end{equation}
Hence, an alternative EW detects a cluster state $|C\rangle$ and then its mirrored EW does its local unitarily equivalent state, $U_C|C\rangle$. 

For instance, for four qubits, the generators are given by 
\bea
g_1 &=& X^{(1)} Z^{(2)},~~    g_2 = Z^{(1)} X^{(2)} Z^{(3)}, \nonumber \\
g_3 &=& Z^{(2)} X^{(3)} Z^{(4)}, ~~\mathrm{and}~~g_4 = Z^{(3)} X^{(4)}. \nonumber 
\eea
from which mirrored EWs are constructed, 
\bea
    W &=& 3\, \id^{\otimes 4} - (g_1 + g_2 + g_3 + g_4), \nonumber \\
    \mathrm{and}~~M &=& 3\, \id^{\otimes 4} + (g_1 + g_2 + g_3 + g_4) \label{eq:4c} 
\eea
and they are connected by a local unitary transform, $    X^{(1)} Y^{(2)} Y^{(3)} X^{(4)}$.

For a four-qubit cluster state $|C_4\rangle $, which can be expressed as follows, 
\[
H^{(1)} H^{(4)} \ket{C_4} = \tfrac{1}{2} (\ket{0000} + \ket{0011} + \ket{1100} - \ket{1111}),
\]
where $H$ denotes a Hadamard gate, its mirrored EW detects a local unitarily equivalent one, 
\bea
 &&   H^{(1)} H^{(4)} ~ X^{(1)} Y^{(2)} Y^{(3)} X^{(4)} \ket{C_4} \nn \\
&=& \frac{1}{2} (\ket{1001} - \ket{0101} - \ket{1010} - \ket{0110}). \nn
\eea
Hence, we have presented mirrored EWs in Eq. (\ref{eq:4c}) detecting local unitarily equivalent cluster states. { The witness pair $(W,M)$ tolerates unitary noise for graph state $\ket{C_4}$ and $X^{(1)} Y^{(2)} Y^{(3)} X^{(4)} \ket{C_4}$ up to $p^\star_U = 1/7$, respectively.}

Secondly, let us revisit alternative EWs for $n$-qubit GHZ states and show local unitaries explicitly. For instance, for $n=3$, EWs are constructed as follows, 
\beas
W &=& 2 \id^{\otimes 3} - (X^{(1)}X^{(2)}X^{(3)} + Z^{(1)}Z^{(2)} + Z^{(2)}Z^{(3)}), \\
M &=& 2 \id^{\otimes 3} + (X^{(1)}X^{(2)}X^{(3)} + Z^{(1)}Z^{(2)} + Z^{(2)}Z^{(3)}),
\eeas
where two EWs are related by a local unitary transform, $Z^{(1)}Y^{(2)} Z^{(3)}$. { The witness pair $(W,M)$ tolerates unitary noise for graph state $\ket{GHZ_3}$ and $Z^{(1)}Y^{(2)} Z^{(3)} \ket{GHZ_3}$ up to $p^\star_U = 1/5$, respectively.}

We have shown that mirrored EWs connected by local unitaries detect entangled states that are local unitarily equivalent. 

\subsection{ Mirrored pair of optimal EWs }

In this section, we explore the existence of a pair of optimal EWs in multipartite systems. In particular, we present mirrored pairs of \emph{optimal} EWs, each of which detects PPT entangled states. 

\subsubsection{Optimality of mirrored EWs} 

We begin by examining the optimality of mirrored EWs in a special class of operators known as \emph{X-shaped} EWs. For $n$ qubits, let $\mathbf{i}$ denote a bitstring of length $n$, indexing basis elements. A Hermitian operator $W$ acting on $n$ qubits is said to be \emph{X-shaped} if its non-zero matrix entries $W_{\mathbf{i},\mathbf{j}}$ satisfy
\[
W_{\mathbf{i},\mathbf{j}} \neq 0 \quad \text{only when} \quad \mathbf{i} = \mathbf{j} \quad \text{or} \quad \mathbf{i} = \bar{\mathbf{j}},
\]
where $\bar{\mathbf{j}}$ is the bitwise complement of $\mathbf{j}$, i.e., $\bar{\mathbf{j}} = \mathbf{j} \oplus \mathbf{1}$. An X-shaped Hermitian operator can typically be put in the following matrix form
\[
\left(
\begin{matrix}
s_1 &&&&&&& u_1\\
& s_2 &&&&& u_2 & \\
&& \ddots &&& \iddots &&\\
&&& s_{2^{n-1}}&u_{2^{n-1}} &&&\\
&&& \bar u_{2^{n-1}}&t_{2^{n-1}}&&&\\
&& \iddots &&& \ddots &&\\
& \bar u_2 &&&&& t_2 &\\
\bar u_1 &&&&&&& t_1
\end{matrix}
\right).
\]
Let us recall results from~\cite{Han2016} for the optimality of X-shaped EWs. 
\begin{theorem} \label{thm:genopt-EW}
Let $W$ be an X-shaped $n$-qubit genuine EW. The following statements are equivalent:
\begin{enumerate}
    \item[(i)] $W$ is an optimal genuine EW;
    \item[(ii)] $W$ satisfies the spanning property;
    \item[(iii)] There exists a bitstring $\mathbf{i}$ and a real number $r > 0$ such that:
    \begin{itemize}
        \item $W_{\mathbf{i},\mathbf{i}} = W_{\bar{\mathbf{i}},\bar{\mathbf{i}}} = 0$ and $\abs{W_{\mathbf{i},\bar{\mathbf{i}}}} = r$;
        \item For every $\mathbf{j} \neq \mathbf{i}$, one has $\sqrt{W_{\mathbf{j},\mathbf{j}} W_{\bar{\mathbf{j}},\bar{\mathbf{j}}}} = r$ and $W_{\mathbf{j},\bar{\mathbf{j}}} = 0$.
    \end{itemize}
\end{enumerate}
\end{theorem}

An example of an optimal genuine EW for a three-qubit system is given by the following matrix:
\[
W_{opt} = \left(
\begin{matrix}
\cdot & \cdot & \cdot & \cdot & \cdot & \cdot & \cdot & e^{i\theta} \\
\cdot & s_2 & \cdot & \cdot & \cdot & \cdot & \cdot & \cdot \\
\cdot & \cdot & s_3 & \cdot & \cdot & \cdot & \cdot & \cdot \\
\cdot & \cdot & \cdot & s_4 & \cdot & \cdot & \cdot & \cdot \\
\cdot & \cdot & \cdot & \cdot & \frac{1}{s_4} & \cdot & \cdot & \cdot \\
\cdot & \cdot & \cdot & \cdot & \cdot & \frac{1}{s_3} & \cdot & \cdot \\
\cdot & \cdot & \cdot & \cdot & \cdot & \cdot & \frac{1}{s_2} & \cdot \\
e^{-i\theta} & \cdot & \cdot & \cdot & \cdot & \cdot & \cdot & \cdot
\end{matrix}
\right).
\]
It can be readily verified that this matrix satisfies all the conditions stated in the theorem. Its mirror is given as
\[
M_{opt} = \beta \id - W_{opt}
\]
where $\beta \geq \max \set{s_i, 1 \slash s_i}^4_{i=2} \geq 1$ to ensure $M_{opt}$ is non-negative for product states. Then we see that $M_{opt}$ is actually positive. We {can generalize} this observation in the following proposition: 

\begin{proposition}
    Let $W$ be a multiqubit X-shaped genuine EW satisfying $W + M = \alpha \id$ for some $\alpha > 0$. If $W$ is optimal, then the mirror operator $M$ is positive semidefinite.
\end{proposition}

\begin{proof} Let $W$ be an X-shaped optimal genuine EW with $\mathbf{i}$ as specified in (iii). By definition, the only nonzero off-diagonal elements of $W$ occur at positions $(\mathbf{i},\bar{\mathbf{i}})$ and $(\bar{\mathbf{i}},\mathbf{i})$. We define the corresponding mirror operator $M$ by
\[
M = \alpha \id - W.
\]
This operator retains the X-shaped structure of $W$ and has off-diagonal elements
\[
M_{\mathbf{i},\bar{\mathbf{i}}} = -r e^{i\theta}, \qquad M_{\bar{\mathbf{i}},\mathbf{i}} = -r e^{-i\theta},
\]
for some real number $r \in \mathbb{R}$ and angle $\theta \in [0,2\pi)$.

To determine whether $M$ is positive, we examine its submatrices. For any bitstring $\mathbf{j} \neq \mathbf{i}$, the corresponding off-diagonal entry vanishes, i.e., $M_{\mathbf{j},\bar{\mathbf{j}}} = 0$. The diagonal entries are given by
\[
M_{\mathbf{j},\mathbf{j}} = \alpha - W_{\mathbf{j},\mathbf{j}}, \quad M_{\bar{\mathbf{j}},\bar{\mathbf{j}}} = \alpha - W_{\bar{\mathbf{j}},\bar{\mathbf{j}}}.
\]
These are non-negative provided that $\alpha$ is greater than or equal to the largest diagonal entry of $W$. Due to condition (iii), we have $\sqrt{W_{\mathbf{j},\mathbf{j}} W_{\bar{\mathbf{j}},\bar{\mathbf{j}}}} = r$, meaning that at least one of the diagonal entries must be no less than $r$. Therefore, choosing $\alpha \geq r$ shows the non-negativity of these entries.

Now consider the $2 \times 2$ principal submatrix corresponding to indices $\mathbf{i}$ and $\bar{\mathbf{i}}$:
\[
\begin{pmatrix}
    \alpha & -r e^{i\theta} \\
    -r e^{-i\theta} & \alpha
\end{pmatrix},
\]
which has eigenvalues
\[
\lambda_\pm = \alpha \pm r.
\]
Both eigenvalues are non-negative for $\alpha \geq r$, implying that this submatrix is positive semi-definite. As a diagonal entry $M_{\mathbf{j},\mathbf{j}}$ is always positive for any bitstring $\mathbf{j} \neq \mathbf{i}$, the operator $M$ itself is positive. This establishes that, for any optimal X-shaped genuine EW $W$, its mirror $M = \alpha \id - W$ is positive semi-definite. 
\end{proof}

{ 
We revisit the $n$-qubit \emph{canonical} GHZ witness, since it fits the (respective) conditions for optimality mentioned in Theorem~\ref{thm:genopt-EW}; hence, it is a genuine multipartite EW and its mirror is positive. \\

Recall that the canonical EW for $n$-qubit GHZ state reads
\bea
W_c \;=\; \frac{1}{2^{n-1}-1} \Big( \tfrac{1}{2}\,\id^{\otimes n} - \dyad{GHZ_n} \Big),
\eea
In the computational basis $\{\ket{\mathbf{j}}\}_{\mathbf{j}\in\{0,1\}^n}$, $\dyad{GHZ_n}$ has only the two off-diagonal entries
$-\tfrac12$ at $\mathbf{j}=0^n,1^n$ and the diagonals $\tfrac12$ between $0^n$ and $1^n$.
Hence $W_c$ is X-shaped: its only nonzero off-diagonal entries lie on the anti-diagonal pairing
$\mathbf{i}=0^n$ with $\bar{\mathbf{i}}=1^n$. Writing explicitly in X-form gives
\bea
W_c \;=\; \alpha \left(
\begin{matrix}
0 &&&&&&& -1/2 \\
& 1/2 &&&&& 0 & \\
&& \ddots &&& \iddots &&\\
&&& 1/2 & 0 &&&\\
&&& 0 & 1/2 &&&\\
&& \iddots &&& \ddots &&\\
& 0 &&&&& 1/2 &\\
-1/2 &&&&&&& 0
\end{matrix}
\right),
\eea
with $\alpha$ is a normalization factor $\alpha = 1/(2^{n-1}-1)$. It is thus straightforward to see that $W_c$ satisfies the optimality conditions. Recall that its mirror is
\bea
M_c \;=\; \dyad{GHZ_n},
\eea
which is manifestly positive semi-definite.
}

\begin{figure}
    \centering
    \includegraphics[width=0.52 \textwidth]{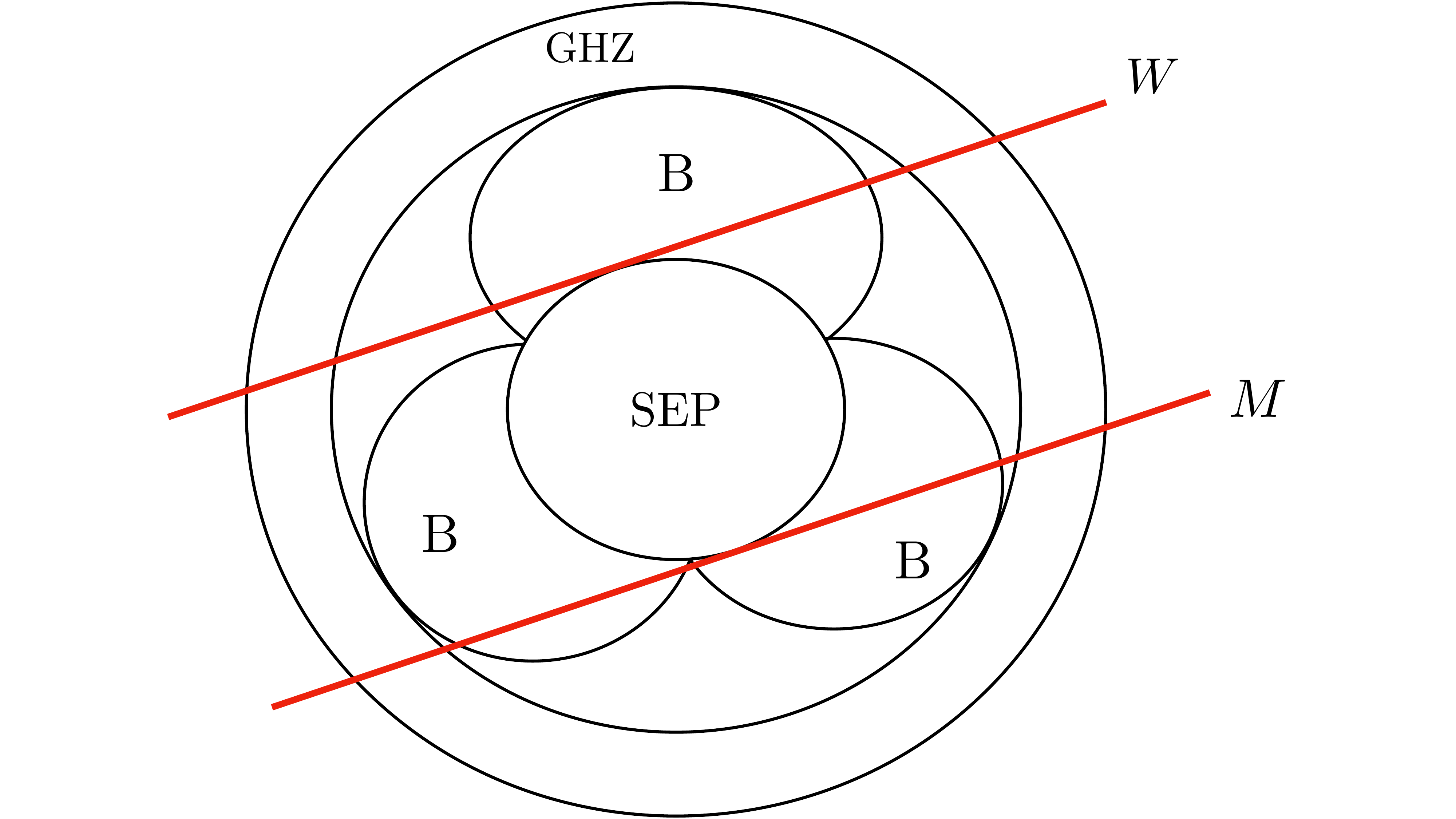}
    \caption{{Three-qubit entangled states have an onion structure \cite{PhysRevLett.87.040401}, where B denotes the set of bi-separable states. Three-qubit EWs in Eqs. (\ref{eq:3w}) and (\ref{eq:3m})) are optimal and detect entangled states which are PPT in all bipartite splittings. Optimal three-qubit EWs can be mirrored with each other.}}
    \label{fig:3opt}
\end{figure}

\begin{table*}[t]
{
\centering
\small
\setlength{\tabcolsep}{12pt}
\renewcommand{\arraystretch}{1.2}
\begin{tabular}{l c c c}
\hline
\textbf{EW} 
& \textbf{Mirror} 
& \textbf{Optimality} 
& $\boldsymbol{p_U^\star}$ \\
&  
& \textbf{(EW and Mirror)}
&  \\
\hline
\rule{0pt}{5ex}
alternative ($n$-graph state) 
& EW (LU-equiv.) 
& X 
& $\displaystyle \frac{1}{2n-1}$ \\
\rule{0pt}{5ex}
alternative ($4$-cluster state) 
& EW (LU-equiv.) 
& X
& $\displaystyle \frac{1}{7}$ \\
\rule{0pt}{5ex}
alternative ($3$-GHZ state) 
& EW (LU-equiv.) 
& X 
& $\displaystyle \frac{1}{5}$ \\
\rule{0pt}{6ex}
$W[i_1,i_2,i_3]$ ($3$-qubit state) 
& EW (optimal, LU-equiv.) 
& O 
& N. A. \\
\rule{0pt}{6ex}
$W[0,1,0]$ 
& EW (optimal, LU-equiv.) 
& O
& $\displaystyle \frac{\big|x+y+z^{-1}-3\big|}{1+\sqrt{17}}$ \\
\rule{0pt}{6ex}
$W[1,1,0]$ 
& EW (optimal, LU-equiv.) 
& O
& $\displaystyle \frac{\big|c-3\big|}{1+\sqrt{17}}$ \\
\rule{0pt}{0.5ex}
& & &  \\
\hline
\end{tabular}
\caption{ EWs for multipartite qubit states are compared. In all cases, EWs and their mirrored pairs are LU-equivalent. Mirrored pairs of optimal EWs can be constructed in many cases, in contrast to bipartite quantum systems. 
}
\label{tab:graph-witness-summary}
}
\end{table*}

\subsubsection{A pair of optimal mirrored witnesses detecting PPT states in three-qubit systems}

We present a family of pairs of optimal mirrored EWs in three-qubit systems. We first consider the following class of three-qubit X-shaped density matrices:
\begin{equation}
\rho_{\rm ppt} = \frac{1}{\sum_{i=1}^4 (s_i + t_i)} \left(%
 \begin{array}{cccccccc} \label{3qX}
   s_1 & \cdot & \cdot & \cdot & \cdot & \cdot & \cdot & u_1 \\
   \cdot & s_2 & \cdot & \cdot & \cdot & \cdot & u_2 & \cdot \\
   \cdot & \cdot & s_3 & \cdot & \cdot & u_3 & \cdot & \cdot \\
   \cdot & \cdot & \cdot & s_4 & u_4 & \cdot & \cdot & \cdot \\
   \cdot & \cdot & \cdot & \bar{u}_4 & t_4 & \cdot & \cdot & \cdot \\
   \cdot & \cdot & \bar{u}_3 & \cdot & \cdot & t_3 & \cdot & \cdot \\
   \cdot & \bar{u}_2 & \cdot & \cdot & \cdot & \cdot & t_2 & \cdot \\
   \bar{u}_1 & \cdot & \cdot & \cdot & \cdot & \cdot & \cdot & t_1 \nonumber
 \end{array}
\right)
\end{equation}
which is PPT in all bipartite splittings \cite{Jafarizadeh2008, PhysRevLett.87.040401} if 
\begin{align*}
    & s_i, t_i > 0 {\rm \quad and \quad} s_i t_i = {\mu},  \\
    & \abs{u_i} \leq 1 {\rm \quad \quad \quad for \; all \;} i = {\mu}, \ldots, 4. 
\end{align*}
Note that the state above reproduces a family of PPT states previously studied in~\cite{PhysRevLett.87.040401, Hyllus2004, Kye2015}. 

Let us consider {a three-parameter family of entanglement witnesses parametrized by $i_1,i_2,i_3$ $^{21}$} that detects PPT states above,
\bea
W[i_1,i_2,i_3]  
&=& \id^{\otimes 3} 
- Z^{(1)} Z^{(2)} Z^{(3)} 
- (-1)^{i_1} X^{(1)} X^{(2)} X^{(3)} \nonumber \\
&&
- (-1)^{i_2} X^{(1)} Y^{(2)} Y^{(3)} 
- (-1)^{i_3} Y^{(1)} X^{(2)} Y^{(3)} \nonumber \\
&& - (-1)^{i_1+i_2+i_3+1} Y^{(1)} Y^{(2)} X^{(3)}, \label{eq:3w}  
\eea
which has been studied in \cite{Jafarizadeh2008}. We compute its mirrored operator, 
\bea
M[i_1,i_2,i_3] 
&=& \id^{\otimes 3}
+ Z^{(1)} Z^{(2)} Z^{(3)} 
+ (-1)^{i_1} X^{(1)} X^{(2)} X^{(3)} \nonumber \\
&&
+ (-1)^{i_2} X^{(1)} Y^{(2)} Y^{(3)} 
+ (-1)^{i_3} Y^{(1)} X^{(2)} Y^{(3)} \nonumber \\
&& + (-1)^{i_1+i_2+i_3+1} Y^{(1)} Y^{(2)} X^{(3)}.\label{eq:3m} \eea
In what follows, we summarize the properties of two EWs.

Firstly, two EWs are equivalent up to local unitaries $Y^{(123)}:=Y^{(1)}Y^{(2)}Y^{(3)}$:
\bea
 M[i_1, i_2, i_3] =  Y^{(123)} W[i_1, i_2, i_3] (Y^{(123)})^{\dagger} \nonumber 
\eea
More generally, it holds that 
    \begin{align*}
        W[i_1',i_2',i_3'] &= U W[i_1, i_2, i_3] U^\dagger,~~ \mathrm{and}\\
        M[i_1',i_2',i_3'] &= U M[i_1, i_2, i_3] U^\dagger,
    \end{align*}
    for some local unitary $U$, e.g.,
    \begin{align*}
        W[1,1,1] &= (Z^{(1)}Z^{(2)}Z^{(3)}) \,W[0,0,0] \,(Z^{(1)}Z^{(2)}Z^{(3)})^\dagger, \\
        M[1,1,1] &= (Z^{(1)}Z^{(2)}Z^{(3)})\, M[0,0,0] \,(Z^{(1)}Z^{(2)}Z^{(3)})^\dagger.
    \end{align*}
A complete characterization of the local equivalences between witness operators is detailed in the Appendix. 

Secondly, both EWs are optimal from the spanning property \cite{Jafarizadeh2008}. For instance, $W[0,0,0]$ has the following eight product states, all of which yield zero expectation values:
\bea
&&    |000\rangle,  ~|011\rangle, ~ |101\rangle,  ~|110\rangle, \nonumber \\
&&   |+++\rangle,~ | +--\rangle,~|-+-\rangle,~ |--+\rangle.\nonumber
\eea
That is, the spanning property is fulfilled. Since two EWs are equivalent up to local unitaries, the spanning property also holds. 

Therefore, the EWs can detect entangled states that remain positive under partial transposition with respect to all bipartite splittings. We show that they detect PPT entangled states studied in~\cite{PhysRevLett.87.040401, Hyllus2004, Kye2015}. 

\begin{itemize}
    \item An EW $W[0,1,0]$ detects a known class of three-qubit bound entangled states in Refs. \cite{PhysRevLett.87.040401, Hyllus2004},  
    \[
    \rho(x,y,z) = \frac{1}{\nu}
    \begin{pmatrix}
    1 & \cdot & \cdot & \cdot & \cdot & \cdot & \cdot & 1 \\
    \cdot & x & \cdot & \cdot & \cdot & \cdot & \cdot & \cdot \\
    \cdot & \cdot & y & \cdot & \cdot & \cdot & \cdot & \cdot \\
    \cdot & \cdot & \cdot & z & \cdot & \cdot & \cdot & \cdot \\
    \cdot & \cdot & \cdot & \cdot & z^{-1} & \cdot & \cdot & \cdot \\
    \cdot & \cdot & \cdot & \cdot & \cdot & y^{-1} & \cdot & \cdot \\
    \cdot & \cdot & \cdot & \cdot & \cdot & \cdot & x^{-1} & \cdot \\
    1 & \cdot & \cdot & \cdot & \cdot & \cdot & \cdot & 1
    \end{pmatrix},
    \]
    with parameters $x,y,z >0$ and $\nu = 2 + x+y+z+x^{-1} +y^{-1} +z^{-1}$ We can see that $\rho(x,y,z)$ above can be reproduced by $\rho_{\rm ppt}$ in Eq.~\eqref{3qX}. The state is thus positive under partial transposition with respect to all bipartitions whenever $x,y,z >0$. We have
    \[
    \tr \big[ W[0,1,0] \, \rho \, (x,y,z) \big] = \frac{2(x+y+z^{-1} -3)}{\nu},
    \]
    which becomes negative when $0< x,y < 1$ and $z > 1$. This shows that $\rho \, (x,y,z)$ is a PPT-entangled state. {$W[1,1,0]$ also tolerates unitary rotation for state $\rho \,(x,y,z)$ whenever
    \bea
        p_U < \frac{\abs{x+y+z^{-1}-3}}{1+\sqrt{17}}
    \eea since $\| W[i_1, i_2, i_3] \| = 1+\sqrt{17}$. The operator norm calculation is provided in the Appendix. }

  {Note also that the locally Pauli-conjugated state $^{22}$}
    \begin{align*}
Y^{(123)}\rho_{\rm  }(x,y,z) ( Y^{(123)} )^\dagger 
    \end{align*}
can be detected by the mirrored one 
\bea
    M[0,1,0] = Y^{(123)} W[0,1,0]  (Y^{(123)})^\dagger,
\eea
Both the witness $W[0,1,0]$ and its mirrored EW $M[0,1,0]$ detect bound entangled states in Ref. \cite{PhysRevLett.87.040401}.
    
    \item An EW $W[1,1,0]$ detects another known class of three-qubit bound entangled states in Ref. \cite{Kye2015}
    \bea
&&    \rho_{\rm  }(b, c) = \nonumber \\
&&    \frac{1}{6 + b + c} \label{eq:3qKye}
    \begin{pmatrix}
    1 & \cdot & \cdot & \cdot & \cdot & \cdot & \cdot & -1 \\
    \cdot & 1 & \cdot & \cdot & \cdot & \cdot & -1 & \cdot \\
    \cdot & \cdot & 1 & \cdot & \cdot & 1 & \cdot & \cdot \\
    \cdot & \cdot & \cdot & b & -1 & \cdot & \cdot & \cdot \\
    \cdot & \cdot & \cdot & -1 & c & \cdot & \cdot & \cdot \\
    \cdot & \cdot & 1 & \cdot & \cdot & 1 & \cdot & \cdot \\
    \cdot & -1 & \cdot & \cdot & \cdot & \cdot & 1 & \cdot \\
    -1 & \cdot & \cdot & \cdot & \cdot & \cdot & \cdot & 1
    \end{pmatrix},\nonumber
\eea
    with parameters $b, c > 0$ such that $bc \geq 1$. It is straightforward to check that $\rho_{\rm  }(b, c)$ can be reproduced by $\rho_{\rm ppt}$ in Eq.~\eqref{3qX}. The state is thus positive under partial transposition with respect to all bipartitions whenever $bc = {\mu}$.
    
    \medskip
    The witness $W[1,1,0]$ detects this PPT entangled state for a range of parameters. Specifically, we compute
    \[
    \tr \big[ W[1,1,0] \, \rho (b,c)\big] = \frac{2(c - 3)}{b + c + 6},
    \]
    which is negative for all $0 < c < 3$, thus confirming it is a PPT-entangled state. The calculation is detailed in the Appendix.
    
    Furthermore, the locally Pauli-conjugated state is given,
 \bea
 && (Y^{(123)} ) \rho_{ }(b, c) (Y^{(123)} )^\dagger \nonumber \\
&=& \frac{1}{6 + b + c}
    \begin{pmatrix}
    1 & \cdot & \cdot & \cdot & \cdot & \cdot & \cdot & 1 \\
    \cdot & 1 & \cdot & \cdot & \cdot & \cdot & 1 & \cdot \\
    \cdot & \cdot & 1 & \cdot & \cdot & -1 & \cdot & \cdot \\
    \cdot & \cdot & \cdot & c & 1 & \cdot & \cdot & \cdot \\
    \cdot & \cdot & \cdot & 1 & b & \cdot & \cdot & \cdot \\
    \cdot & \cdot & -1 & \cdot & \cdot & 1 & \cdot & \cdot \\
    \cdot & 1 & \cdot & \cdot & \cdot & \cdot & 1 & \cdot \\
    1 & \cdot & \cdot & \cdot & \cdot & \cdot & \cdot & 1
    \end{pmatrix}.\nonumber 
\eea
Due to the equivalence by local unitaries, the state above is also PPT in all bipartite splittings. Note that
    \[
    M[1,1,0] = (Y^{(123)} )  W[1,1,0]  (Y^{(123)} )^\dagger,
    \]
    and therefore, for $0<c<3$,
    \bea
&&    \tr \big[ M[1,1,0] (Y^{(123)} ) \rho(b,c) (Y^{(123)} )^{\dagger} \big] \nonumber \\
& = & \tr \big[ W[1,1,0] \rho (b,c) \big] < 0, \nonumber  
    \eea
Hence, both EWs, $W[1,1,0]$ and $M[1,1,0]$, detect the PPT entangled states. 
\end{itemize}

\section{Bipartite systems: optimality}
\label{sec:bio}

In this section, we investigate mirrored EWs for bipartite systems $n\otimes n$ that contain symmetries, such as Bell diagonal and covariant with respect to a maximal commutative subgroup of the unitary group $U(n)$. 

Let us summarize the construction of EWs, and then discuss the properties of their mirrored ones. A Weyl operator $U_{mk}$ can be defined as \cite{Bell1,Bell2}
\begin{equation}
  U_{mk} |\ell\> = \omega^{m\ell} |\ell+k\>~(\mbox{mod}~n),\label{weyl}
\end{equation}
with $\omega = e^{2 \pi i/n}$. Weyl operators satisfy relations as follows 
\bea
&&  U_{k\ell} U_{rs} = \omega^{ks} U_{k+r,\ell + s},~~  U_{k\ell}^* = U_{-k\ell} , \nonumber\\
&&   U^\dagger_{k\ell} = \omega^{k\ell} U_{-k,-\ell}, ~~\mathrm{and}   ~~\tr[  U_{k\ell} U_{rs}^\dagger] = n\, \delta_{kr}\delta_{\ell s}\nonumber
\eea
where addition and subtraction are in modulo $n$. The generalized Bell states are denoted by \begin{equation}\label{gBell}
  |\psi_{k\ell}\> = \I \otimes U_{k\ell} |\psi_{00}^{}\rangle,
\end{equation}
where $ |\psi_{00} \> = {\mu}/\sqrt{n} \sum_{k=0}^{n-1} |k k\>$. A Bell-diagonal operator can be written as 
  $X = \sum_{k,\ell=0}^{n-1} x_{k\ell} P_{k\ell},$
with $P_{k\ell} =|\psi_{k\ell}\>\<\psi_{k\ell}|$. Let us consider a maximal commutative subgroup of a unitary group $U(n)$
\begin{equation}
  T(n) = \{ \ U \in U(n)\, |\, U = \sum_{k=0}^{n-1} e^{i \phi_k} |k\>\<k| \ \}. \nonumber 
\end{equation}
with $\phi_k \in \mathbb{R}$. Moreover, a bipartite operator $X$ is said to be $T \otimes T^*$-covariant whenever
\begin{equation}\label{}
  U \otimes U^* X (U \otimes U^*)^\dagger = X ,
\end{equation}
for any $U \in T(n)$. Then, any covariant operator has the following structure 
\bea
  X = \sum_{k,\ell=0}^{n-1} A_{kl} |k\>\<k| \otimes |\ell\>\<\ell| + \sum_{k\neq \ell=0}^{n-1} B_{kl} |k\>\<\ell| \otimes |k\>\<\ell|,~~~~~~ \label{XAB}
\eea
with complex parameters $A_{k\ell}$ and $B_{k\ell}$. Note that for such operators the Hermitian matrix $A_{k\ell}$ is circulant, i.e. $A_{k\ell} = \alpha_{k-\ell}$ for some real vector $(\alpha_0,\alpha_1,\ldots,\alpha_{n-1})$ and $B_{kl}$ is a constant, i.e. $B_{k\ell} = \beta \in \mathbb{R}$. 

\begin{proposition}\cite{OSID} A symmetric operator in the following is an EW, 
\bea
  W = \sum_{k,\ell=0}^{n-1} \alpha_{k-l} |k\>\<k| \otimes |\ell\>\<\ell| - \sum_{k\neq \ell=0}^{n-1} |k\>\<\ell| \otimes |k\>\<\ell|,~~~~\label{Wit_gen}
\eea
when a circulant matrix $A_{k\ell}$ satisfies the following constraints
\begin{equation}\label{I0}
\alpha_0 + \alpha_1 + \ldots + \alpha_{n-1} = n-1 ,
\end{equation}
and $AA^T = \mathbb{I} + (n-2) \mathbb{J}$ where $\mathbb{J}_{k\ell}= {\mu}$.
\end{proposition}

\subsection{ EWs in $3\otimes 3$ }

Let us consider a class of EWs that generalize the well-known Choi EWs \cite{Korea-1992}
\begin{eqnarray}
  W[a,b,c] &=&\sum_{i=0}^2 \Big[\, a\, |ii\>\<ii| + b\, |i,i+1\>\<i,i+1| \nonumber\\ &+&  c\, |i,i+2\>\<i,i+2|\, \Big] 
   - \sum_{i\neq j}  |ii\>\<jj|\      \label{W-abc}
\end{eqnarray}
with parameters $a,b,c\geq 0$  satisfying
\begin{itemize}
\item $a+b+c\geq 2$, and 
\item if $a \leq 1$, then $bc\geq (1-a)^2$. 
\end{itemize}
The well-known Choi EWs are intances, $W[1,1,0]$ or $W[1,0,1]$ \cite{Choi-2,Choi-3,Choi-4}. Choi EWs are extremal \cite{Choi-5,Ha-extr} and hence also optimal. Note, however, that though being optimal, Choi EWs do not satisfy the spanning property \cite{S71,S72}. Note that $W[a,b,c]$ are optimal when 
$$ a+b+c=2 ~~\mathrm{and}~~ a^2+b^2+c^2=2 , $$ 
if and only if $a \in [0,1]$, and  extremal if and only if $a \in (0,1]$  \cite{Kossak, Filip1}). Note also that $W[a,b,c]$ is decomposable only if $b=c=1$. 

Let us introduce parameterization as follows,
\begin{eqnarray}
  a &=&\frac 23 ( 1 + \cos \phi) , \ \ \ 
   b = \frac 13 ( 2 - \cos \phi - \sqrt{3} \sin\phi) , \nonumber\\
    c &=& \frac 13 ( 2 - \cos \phi + \sqrt{3} \sin\phi) , \label{param3}
\end{eqnarray}
that is, one has a 1-parameter family of EWs $W(\phi)$ for $\phi \in [0,2\pi)$. $W(\pi/3)$ and $W(5\pi/3)$ correspond to a pair of Choi witnesses and $W(\pi)$ corresponds to EW defined via the reduction map. $W(\phi)$ is optimal iff $\phi \in [\pi/3,5\pi/3]$.  

Then, mirrored operators are characterized as follows. \\
   
\begin{proposition} \label{P-abc} Mirrored operators to EWs in Eq. (\ref{W-abc}) are positive, i.e., which can be interpreted as quantum states, for $a \in [0,1/3]$ and decomposable EWs for $a \in (1/3,4/3]$.
\end{proposition}

\begin{figure}
    \centering
    \includegraphics[width=0.55 \textwidth]{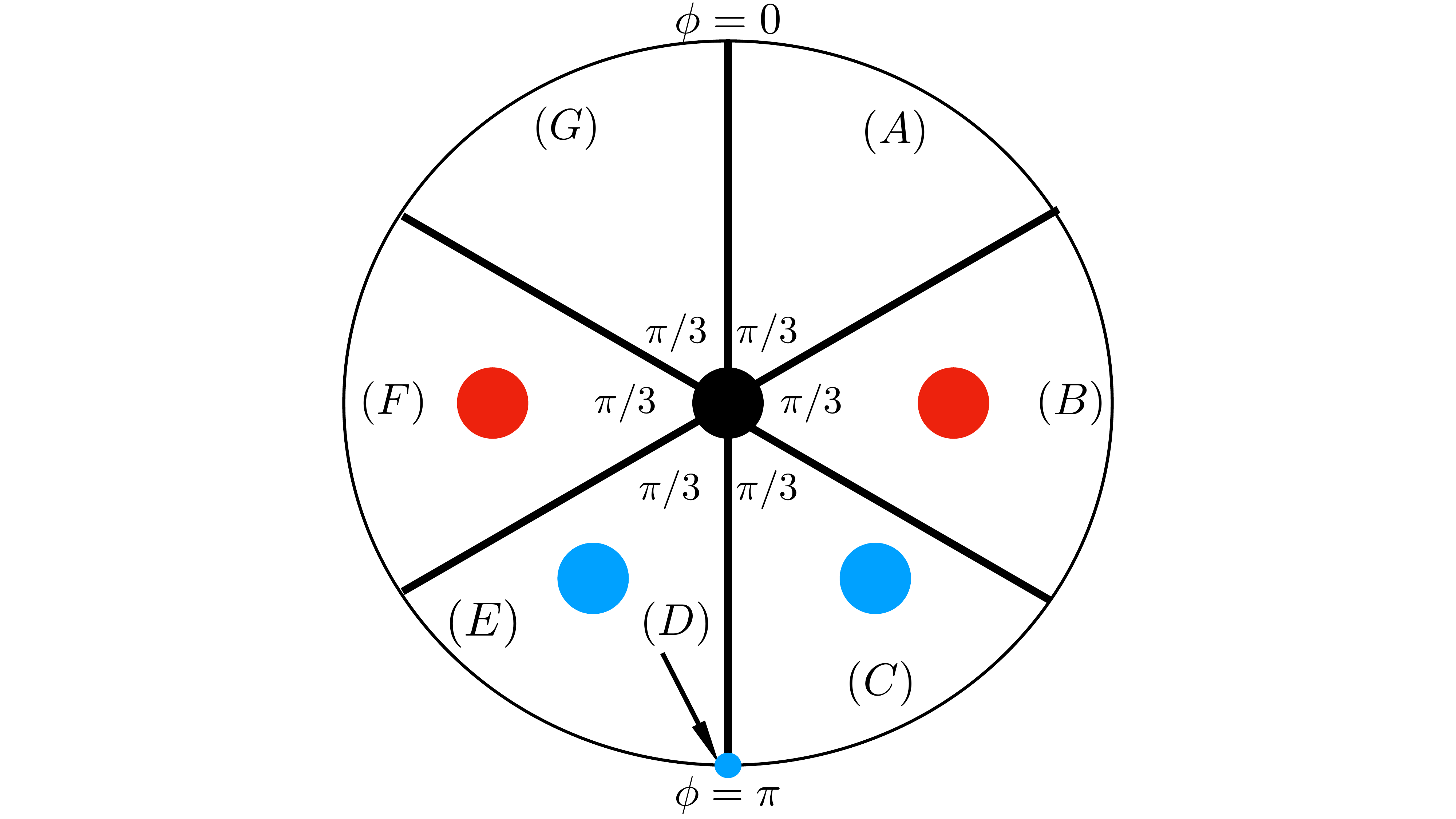}
    \caption{ EWs in Eq. (\ref{W-abc}) are optimal for $\phi\in [\pi/3, 5\pi/3]$, see the parameterization in Eq. (\ref{param3}). For (A) and (G), EWs are not optimal. For (D), i.e., $\phi=\pi$, the EW is decomposable and not optimal, and its mirror is positive semidefinite. For (B) and (F), EWs are non-decomposable and optimal, and their mirrored operators are decomposable EWs. For (C) and (E), EWs are non-decomposable and optimal, and their mirrored operators are positive definite.  }
    \label{fig:phi}
\end{figure}

\subsection{ EWs in $4 \otimes 4 $ }

For $n=4$, we consider instances of EWs in Eq.~(\ref{Wit_gen}) with four parameters \cite{PRA-2022} 
\bea
  W[a,b,c,d] &=& \sum_{i=0}^3 \Big[\, a\, |ii\>\<ii| + b\, |i,i+1\>\<i,i+1| \nonumber\\ 
  && \quad + c\, |i,i+2\>\<i,i+2| +  d\, |i,i+3\>\<i,i+3|\, \Big] \nn \\ 
  && \quad - \sum_{i\neq j}  |ii\>\<jj|\, \label{W-abcd} 
\eea
with $a,b,c,d\geq 0$ satisfying 
\begin{eqnarray}  
  a+b+c+d &=& a^2+b^2+c^2+d^2 = 3, \label{II-1}\\
  ac+bd &=& 1 , \ \ \ \  (a+c)(b+d) =  2 .\nonumber 
\end{eqnarray}
There are two solutions to the above set of equations \cite{PRA-2022}. From four parameters and four equations, we have two variables, called classes I and II:
\bea
\mathrm{ I}:&~&  a+c=2 \ , \ \ b+d=1 \ \nonumber \\
\mathrm{II}: &~&  a+c=1 \ , \ \ b+d=2 \ \nonumber 
\eea
EWs from class I are not optimal, whereas those from class II are optimal \cite{PRA-2022}. 

\subsubsection{Class I: non-optimal non-decomposable EWs}

EWs may be expressed by a single parameter $\theta$, 
\begin{eqnarray}
a &=& \frac{1}{2} (2-\sin\theta)  \ , \ b = \frac{1}{2} (1+\cos\theta) \ , \nonumber\\ \ c &=& 2-a \ , ~\mathrm{and}~\ d = {\mu}- b,
\label{eq:theta_class1}
\end{eqnarray}
In particular, $W(0) = W[1,1,1,0]$ and $W(\pi) = W[1,0,1,1]$ are Choi-like EWs, and 
\bea
W( \frac{\pi}{2}) = W[\frac{1}{2},\frac{1}{2},\frac{3}{2},\frac{1}{2}] \nonumber 
\eea
is the only decomposable EW in the class. The mirrored operator is computed for a witness $W (0)= W[1,1,1,0]$,
\begin{equation}
    M[1,1,1,0] = \frac 43\, \oper_4 \otimes \oper_4 - W[1,1,1,0].  \nonumber
\end{equation}
The mirrored operator can detect PPT entangled states in the following, 
\begin{eqnarray}
\label{xrho}
  \rho_x &\propto &\sum_{i=0}^3 \Big[\, 3\, |ii\>\<ii| + x\, |i,i+1\>\<i,i+1| + \, |i,i+2\>\<i,i+2| \nonumber\\ &+&  \frac 1x\, |i,i+3\>\<i,i+3|\, \Big]
   - \sum_{i\neq j}  |ii\>\<jj|\  ,
\end{eqnarray}
which is PPT for $x > 0$. It follows that 
\begin{equation}
    \tr [  M[1,1,1,0]  \rho_x ] = \frac{4}{3x}(x^2-5x +4),\nonumber 
\end{equation}
which is negative for $x \in (1,4)$. Hence, the mirrored operator $M[1,1,1,0]$ is a non-decomposable EW. 

We have shown that a pair of non-decomposable EWs can be mirrored with each other. 

\subsubsection{Class II: optimal non-decomposable EWs}

We introduce a parameterization as follows:
\begin{eqnarray}
a &=&  \frac{1}{2} (1+\cos\theta) \ , \ b = \frac{1}{2} (2-\sin\theta) \ ,\nonumber\\ \ c &=& 1-a\ , \ d = 2-b ,\label{eq:theta_class2}
\end{eqnarray}
with $\theta \in[0,\pi]$. For $\theta=0$ and $\theta=\pi$, EWs are decomposable. Otherwise, i.e., $\theta\in(0,\pi)$, they are non-decomposable and optimal. 

It turns out that mirrored operators are positive semidefinite for EWs with $\theta\in (\pi/2,\pi)$ and decomposable EWs for $ \theta \in \in (0,\pi/2)$. Hence, for optimal EWs in class II, their mirrored operators cannot be non-decomposable, i.e., either decomposable EWs or positive semidefinite \cite{arxiv25}.

\begin{figure}
    \centering
    \includegraphics[width=0.47 \textwidth]{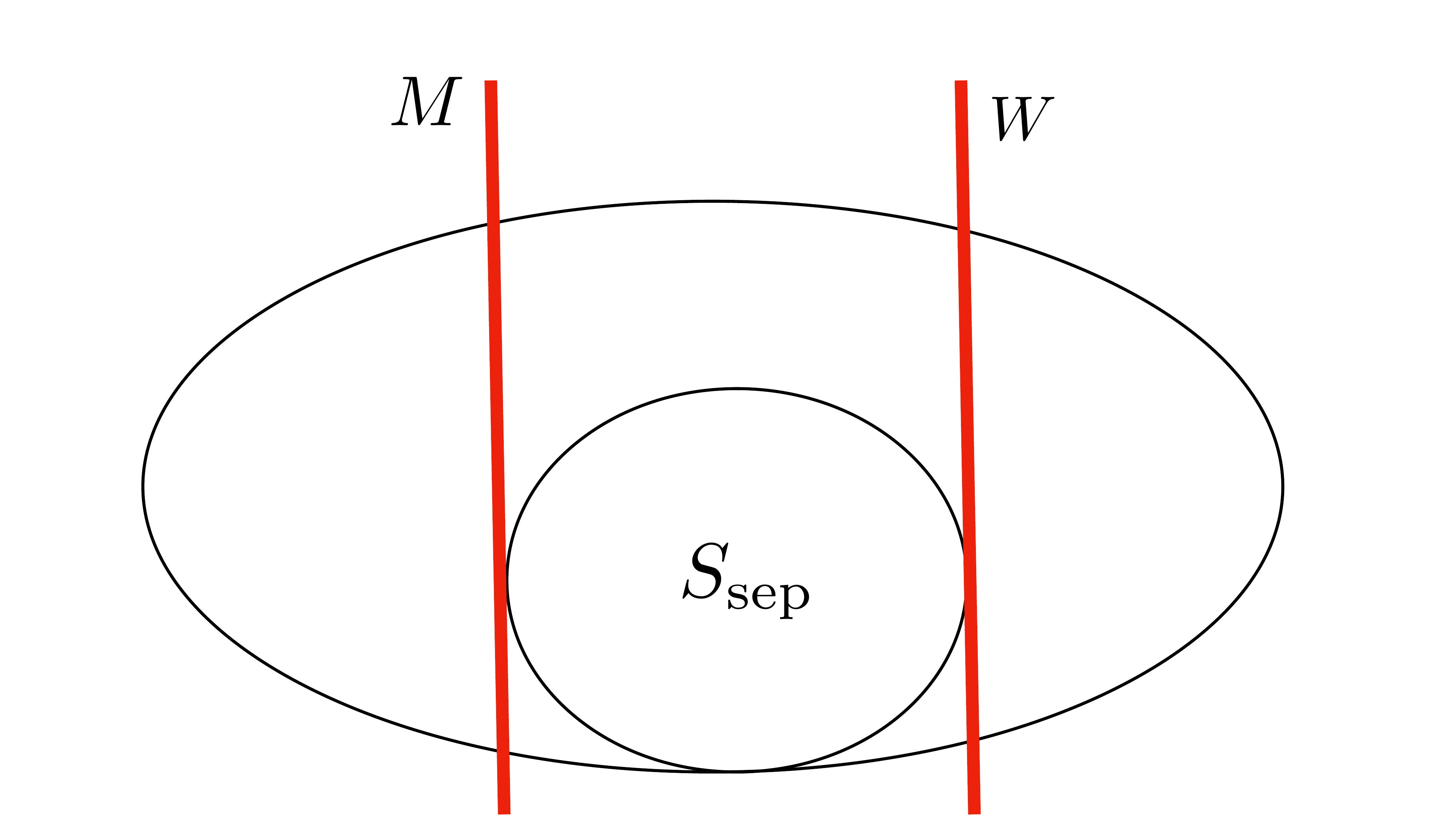}
    \caption{ A pair of optimal EWs in Eqs. (\ref{W&W}) can be mirrored with each other. Both are also non-decomposable. Two EWs are related by local unitaries. }
    \label{fig:opt3}
\end{figure}

\subsection{ Mirrored pair of optimal non-decomposable 
EWs}

So far, we have considered mirrored operators of non-decomposable EWs and found that the mirrored ones are not EWs detecting PPT entangled states; they are either decomposable EWs or positive semidefinite operators. Recently, an intriguing instance of mirrored EWs has been presented in $3\otimes 3$ systems \cite{arxiv25}, where both EWs are optimal, non-decomposable, and equivalent up to local unitaries. We here summarize the pair of optimal non-decomposable EWs.

Firstly, two EWs in the following are non-decomposable, \bea
{W} &=&
\left[\begin{array}{c c c|c c c|c c c}
\cdot & \cdot & \cdot & \cdot & 1 & \cdot & \cdot & \cdot & 1 \\
\cdot & 3 & \cdot & \cdot & \cdot & -2 & -2 & \cdot & \cdot \\
\cdot & \cdot & 3 & -2 & \cdot & \cdot & \cdot & -2 & \cdot \\
\hline
\cdot & \cdot & -2 & 3 & \cdot & \cdot & \cdot & -2 & \cdot \\
1 & \cdot & \cdot & \cdot & \cdot & \cdot & \cdot & \cdot & 1 \\
\cdot & -2 & \cdot & \cdot & \cdot & 3 & -2 & \cdot & \cdot \\
\hline
\cdot & -2 & \cdot & \cdot & \cdot & -2 & 3 & \cdot & \cdot \\
\cdot & \cdot & -2 & -2 & \cdot & \cdot & \cdot & 3 & \cdot \\
1 & \cdot & \cdot & \cdot & 1 & \cdot & \cdot & \cdot & \cdot
\end{array}\right],\nonumber\\
~\mathrm{and}~~{M} & = &
\left[\begin{array}{c c c|c c c|c c c}
4 & \cdot & \cdot & \cdot & -1 & \cdot & \cdot & \cdot & -1 \\
\cdot & 1 & \cdot & \cdot & \cdot & 2 & 2 & \cdot & \cdot \\
\cdot & \cdot & 1 & 2 & \cdot & \cdot & \cdot & 2 & \cdot \\
\hline
\cdot & \cdot & 2 & 1 & \cdot & \cdot & \cdot & 2 & \cdot \\
-1 & \cdot & \cdot & \cdot & 4 & \cdot & \cdot & \cdot & -1 \\
\cdot & 2 & \cdot & \cdot & \cdot & 1 & 2 & \cdot & \cdot \\
\hline
\cdot & 2 & \cdot & \cdot & \cdot & 2 & 1 & \cdot & \cdot \\
\cdot & \cdot & 2 & 2 & \cdot & \cdot & \cdot & 1 & \cdot \\
-1 & \cdot & \cdot & \cdot & -1 & \cdot & \cdot & \cdot & 4
\end{array}\right] \label{W&W}
\eea
as they detect PPT entangled states as follows, 
\bea
\rho_{W} &= & \frac{1}{15}
\left[\begin{array}{c c c|c c c|c c c}
3 & \cdot & \cdot & \cdot & \cdot & \cdot & \cdot & \cdot & \cdot \\
\cdot & 1 & \cdot & \cdot & \cdot & 1 & 1 & \cdot & \cdot \\
\cdot & \cdot & 1 & 1 & \cdot & \cdot & \cdot & 1 & \cdot \\
\hline
\cdot & \cdot & 1 & 1 & \cdot & \cdot & \cdot & 1 & \cdot \\
\cdot & \cdot & \cdot & \cdot & 3 & \cdot & \cdot & \cdot & \cdot \\
\cdot & 1 & \cdot & \cdot & \cdot & 1 & 1 & \cdot & \cdot \\
\hline
\cdot & 1 & \cdot & \cdot & \cdot & 1 & 1 & \cdot & \cdot \\
\cdot & \cdot & 1 & 1 & \cdot & \cdot & \cdot & 1 & \cdot \\
\cdot & \cdot & \cdot & \cdot & \cdot & \cdot & \cdot & \cdot & 3
\end{array}\right]~~\mathrm{and} \nonumber\\
\rho_{M} & = & \frac{1}{15} \left[
\begin{array}{ccc|ccc|ccc}
 1 & \cdot & \cdot & \cdot & 1 & \cdot & \cdot & \cdot & 1 \\
 \cdot & 2 & \cdot & \cdot & \cdot & -1 & -1 & \cdot & \cdot \\
 \cdot & \cdot & 2 & -1 & \cdot & \cdot & \cdot & -1 & \cdot \\\hline
 \cdot & \cdot & -1 & 2 & \cdot & \cdot & \cdot & -1 & \cdot \\
 1 & \cdot & \cdot & \cdot & 1 & \cdot & \cdot & \cdot & 1 \\
 \cdot & -1 & \cdot & \cdot & \cdot & 2 & -1 & \cdot & \cdot \\\hline
 \cdot & -1 & \cdot & \cdot & \cdot & -1 & 2 & \cdot & \cdot \\
 \cdot & \cdot & -1 & -1 & \cdot & \cdot & \cdot & 2& \cdot \\
 1 & \cdot & \cdot & \cdot & 1 & \cdot & \cdot & \cdot & 1 \\
\end{array} \nonumber
\right],
\eea
such that
\bea
   \tr [  W  \rho_W ] =  \tr [M  \rho_M ]  = - \frac{2}{5}.\nonumber 
\eea
Note that EWs in Eq. (\ref{W&W}) are constructed from positive maps via depolarization with mutually unbiased bases \cite{arxiv25}. 

Both EWs in Eq. (\ref{W&W}) are optimal as they fulfill the spanning property. In this case, it also holds that they are connected by local unitaries
\bea
M = U \otimes U^*  W  U^\dagger\otimes U^{*\dagger} \nonumber 
\eea
where 
\bea
U=\frac{1}{\sqrt{3}}\left(
\begin{array}{ccc}
 1 & 1 & \omega\\
 1 & \omega & 1 \\
 \omega^* & \omega & \omega \end{array}
\right)
\eea
and $\omega = \exp[ 2\pi i /3]$.

\section{Bipartite systems: generalization} 
\label{sec:big}

We here generalize mirrored EWs in Eq. (\ref{eq:mr}) such that they are connected by some other observable, not an identity, 
\bea
W + M = \mathbbm{K},\label{eq:gmeew_form}
\eea
where it follows that $\mathbbm{K}$ block-positive, i.e., it is nonnegative for all separable states since $W$ and $M$ are EWs. Generalized mirrored EWs allow one to approach the following problems. Firstly, a pair of EWs may be constructed to detect a larger set of entangled states such that there may be entangled states detected by both EWs in common, i.e., 
\bea
D_W\cap D_M \neq \emptyset. \label{eq:emptyg}
\eea
Secondly, one may attempt to characterize non-extremal EWs composed of optimal EWs and construct EWs that are efficient in the detection of entangled states. 

Let us exploit EWs of class II in Eq.~\eqref{W-abcd}, which are optimal and non-decomposable EWs. We then apply a Weyl operator $U_{jk}$ to construct its generalized mirrored one
\bea
M = (I \otimes U_{jk}) W (I \otimes U_{jk})^{\dag}.\nonumber
\eea
Note that two EWs, $W$ and $M$, are unitarily equivalent; hence, both are optimal and non-decomposable, and the two sets $D_W$ and $D_M$ are unitarily equivalent.

\begin{figure}
    \centering
    \includegraphics[width=0.48 \textwidth]{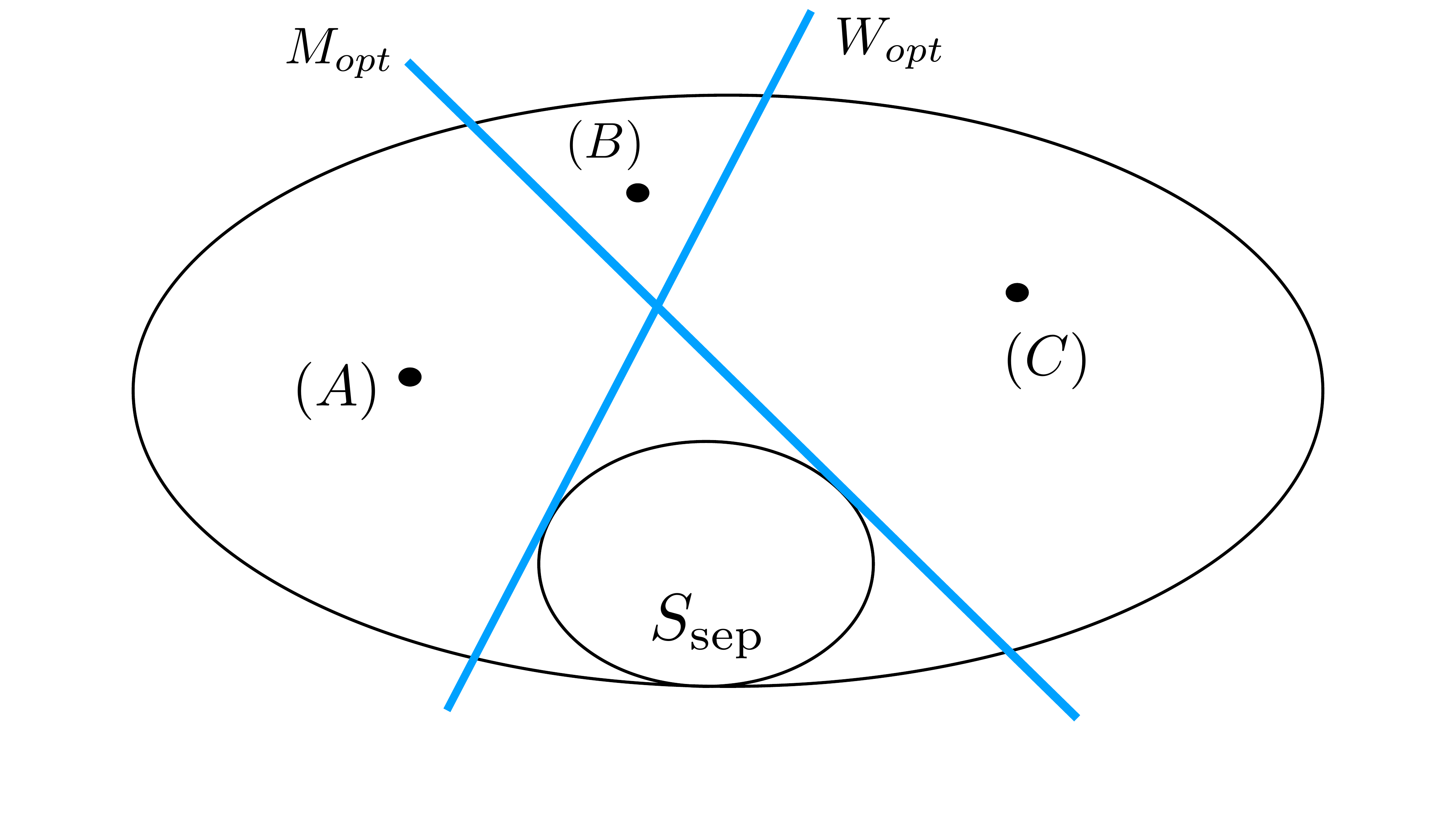}
    \caption{ Entangled states (A) and (C) are detected by a witness $W$ and its mirrored one $M$, respectively, where both EWs are non-decomposable and optimal. Both EWs detect an entangled state (B) in common, see Eq. (\ref{eq:emptyg}), and $W+M$ also defines an EW, see Eq. (\ref{eq:gmeew_form}). }
    \label{fig:gmew}
\end{figure}

A witness $W$ detects a family of PPT entangled states $\rho_x$ in Eq.~\eqref{xrho} since, 
\bea
    \tr [W \rho_x] &=&  a + bx + c+ \frac{d}{x} - 3 \nonumber \\
    &=& b (x-1) + d \left( \frac{1}{x} - 1\right),\nonumber 
\eea
which is negative for 
\bea
\frac{1}{2b} \left( b+d - \abs{b-d} \right) < x <\frac{1}{2b} \left( b+d + \abs{b-d} \right). \nonumber 
\eea
We have $\tr [W \rho_x] < 0$ for $b\neq d$. Since $W$ and $M$ are unitarily equivalent, $M$ detects entangled states$(I \otimes U_{jk}) \rho_x (I \otimes U_{jk})^{\dag}$, which are PPT; $M$ is non-decomposable.  

With two optimal EWs, we construct an operator 
\bea
\mathbbm{K} = W +(I \otimes U_{jk})  W (I \otimes U_{jk})^{\dagger}. \nonumber
\eea
It is clear that $\mathbbm{K}$ is non-negative for all separable states. Let us consider an entangled state 
\bea
\tau = \frac{1}{2}( P_{00} + P_{j1}) \nonumber 
\eea
where $P_{kl} = |\psi_{kl}\rangle \langle \psi_{kl}|$ in Eq. (\ref{gBell}). Note that the state above does not remain positive after the partial transpose. From the parametrization of $a,b,c,d$ in Eq. (\ref{eq:theta_class2}), 
\bea
    \tr [W \tau] &=& \frac{1}{4} (  \cos \theta - \sin \theta -3), \nonumber \\
  \mathrm{and} ~~ \tr [M \tau] &=& \frac{1}{4} (  \cos \theta + \sin \theta -3), \nonumber
\eea
both of which are negative. By construction, it follows that
\be
    \tr [\mathbbm{K} \tau] = \tr [W \tau] + \tr [M \tau] < 0,\nonumber  
\ee
and each one $W$ and $M$ detects entangled states $\tau$ in common, see Fig.~\ref{fig:gmew}. Hence, we have shown that optimal and non-decomposable EWs can be mirrored with each other by another EW.

\section{Open problems}
\label{sec:open}

Mirrored EWs pave the pathway to increasing the usefulness of an experimental realization of EWs. Once an EW is estimated, the estimation can be exploited to detect distinct types of entangled states. Our results open an avenue to develop EWs with intriguing questions in the following. 
\begin{enumerate}
    \item For bipartite systems, we have found that for two-qubit states, optimal EWs cannot be mirrored with each other. One example of a mirrored pair of optimal EWs have been obtained later for systems $3\otimes 3$ \cite{arxiv25}.  However, it is not straightforward to extend the result to higher-dimensional systems.  It remains open to find if a mirrored pair of optimal EWs exists for higher dimensions. It would also be interesting to seek such pairs with covariant Bell-diagonal EWs. 
   \item Optimal EWs can be mirrored via another EW, as shown in Eq. (\ref{eq:gmeew_form}). Since optimal EWs are generally not extremal, yet, it remains unanswered whether a pair of optimal EWs can be mirrored through another optimal EWs. 
    \item Mutually unbiased bases (MUBs) are a measurement setting for quantum state verification. They can also be used for constructing EWs \cite{Bae_2019, Hiesmayr_2021, Bae_2022}, in which extendible and unextendible MUBs, which are inequivalent, provide different efficiencies in detecting entangled states. For practical applications, it would be interesting to find mirrored EWs for EWs constructed from MUBs and investigate how efficient a subset of MUBs is in entanglement detection before a quantum state is fully verified.   
    \item Throughout, we have investigated mirrored linear EWs. In future directions, it is worth investigating the properties of mirrored nonlinear EWs. 
    \item All instances provided here about mirrored pairs of optimal EWs satisfy the equivalence up to local unitaries. That is, mirrored EWs detect distinct sets of entangled states in which states of one set can be transformed the others by local untaries. It remains open if there could be a mirrored pair of optimal EWs that cannot be transformed by local unitaries.  
    \end{enumerate}

\section{Conclusion}
\label{sec:con}

We have developed mirrored EWs for multipartite and high-dimensional quantum systems. For multipartite systems, we have investigated mirrored EWs for $n$-partite GHZ states, graph states, and multipartite bound entangled states.  A mirrored pair of optimal EWs is also shown for three-qubit systems. For bipartite systems, we have considered mirroring for existing optimal EWs and found that the resulting mirrored ones are either decomposable EWs or positive semidefinite. A mirrored pair of optimal EWs is then constructed in a three-dimensional bipartite system. Hence, we have shown that mirrored pairs of optimal EWs, which do not exist for two-qubit states, appear in higher dimensions and multipartite systems. The framework of mirrored EWs has been generalized by taking an EW, not an identity, to which two EWs are mirrored with each other. Our results shed new light on the theory of EWs and also on the efficient detection of entangled states in practical scenarios.  

{On the theoretical aspect, we stress that EWs present the fundamental framework of detecting entangled states. The framework haas been applied to developing measurement-device-independent EWs that relax assumptions on experimental implementations and do not ask one to trust a measurement setting \cite{PhysRevLett.110.060405}. In future directions, it would be interesting to apply the framework of mirrored EWs to scenarios where an implementation cannot be fully trusted, such as semi-device-independent and device-independent EWs. Moreover, as noted, the mirroring framework of a separating hyperplane applies to resource theories in general. The framework of mirrored EW can also apply and be used to detect useful resources \cite{PhysRevX.9.031053}. }

For practical applications, our results are readily applied to realistic scenarios of quantum communication, e.g., fibre-optic communication, where rotations of local bases of qubit states are a critical source of errors. Entanglement is also a precondition in secure quantum communication \cite{PhysRevLett.92.217903, PhysRevLett.94.020501}. Mirrored pairs of EWs considered in the present work, when both are EWs, are connected by local unitaries. Consequently, entangled states detected by each EW of a pair are local unitarily equivalent. Therefore, mirrored EWs are robust against the noise from rotations of local bases. We envisage that mirrored EWs can enhance the detection of entangled states in realistic entanglement-based quantum communication protocols.

\section*{Acknowledgments}
J.S. and J.B. are supported by the National Research Foundation of Korea (Grant No. NRF-2021R1A2C2006309, NRF-2022M1A3C2069728, RS-2024-00408613, RS-2023-00257994), and the Institute for Information \& Communication Technology Promotion (IITP) (RS-2023-00229524, RS-2025-02304540). 

A.B. acknowledges the support received for this research from the research grant sanctioned by the National Board for Higher Mathematics (NBHM), Department of Atomic Energy (DAE), Government of India, under sanction letter no: 02011/32/2025/NBHM(R.P)/R$\&$D II/9677.

B.C.H. acknowledges gratefully that this research was funded in whole, or in part, by the Austrian Science Fund (FWF) project P36102-N (Grant DOI: 10.55776/P36102). For the purpose of open access, the author has applied a CC BY public copyright license to any Author Accepted Manuscript version arising from this submission. The funder played no role in study design, data collection, analysis and interpretation of data, or the writing of this manuscript.

D.C. was supported by the Polish National Science Centre project No. 2024/55/B/ST2/01781.

\section*{Appendix}

{\subsection{ Robustness inequalities}
In this subsection, we provide rigorous proofs of the inequalities employed in Section~E, Chapter~I (\emph{Error Robustness}). 

We collect and prove several basic identities and bounds used in our robustness analysis. 
Throughout, $\| \cdot \|_\opn$ denotes the operator norm, where 
\begin{equation}
    \|A\|_\opn := \sup_{\|\psi\|=1} \|A\psi\| = \sup_{\|\psi\|=1} |\langle \psi | A | \psi \rangle|.
\end{equation}
Now we state the first lemma:
\begin{lemma}[Operator exponential]
\label{lem:exp-FTC}
Let $H$ be a bounded operator on a finite-dimensional Hilbert space.
Then
\bea
e^{iH}-\id \;=\; i \int_0^1 e^{itH}\,H\,dt.
\eea
\end{lemma}

\begin{proof}
Define $f(t):=e^{itH}$ for $t\in[0,1]$. Then $f$ is differentiable and, by the power series of the matrix exponential,
\bea
f'(t) \;=\; iHe^{itH} \;=\; i\,e^{itH}H,
\eea
since $H$ commutes with any polynomial in $H$. By the fundamental theorem of calculus,
\bea
e^{iH}-\id \;=\; f(1)-f(0) \;=\; \int_0^1 f'(t)\,dt
\;=\; \int_0^1 i\,e^{itH}H\,dt.
\eea
\end{proof}

\begin{corollary}[A norm bound for exponentials]
\label{cor:exp-bound}
For bounded $H$,
\bea
\|e^{iH}-\id\|_\opn \;\le\; \|H\|_\opn.
\eea
\end{corollary}

\begin{proof}
By Lemma~\ref{lem:exp-FTC} and $\|e^{itH}\|_\opn=1$,
\bea
\|e^{iH}-\id\|_\opn &\;=\;& \left\| i \int_0^1 e^{itH}H\,dt \right\|_\opn \nonumber \\
&\;\le\;& \int_0^1 \|e^{itH}\|_\opn \, \|H\|_\opn \,dt \;=\; \|H\|_\opn.
\eea
\end{proof}

\begin{lemma}[Expectation bound]
\label{lem:simple-exp-bound}
For any bounded $A$, state $\rho$, and unitary $U$,
\bea
\big|\tr\!\big[ A\rho\big]\big| \;\le\; \| A \|_\opn.
\eea
\end{lemma}

\begin{proof}
Let $\rho=\sum_i p_i \dyad{\psi_i}$ be a density matrix with $p_i\ge 0$, $\sum_i p_i=1$. Then
\bea
\tr\!\big[ A \rho\big]
&=& \sum_i p_i \bra{\psi_i}A\ket{\psi_i} \nonumber \\
&\le& \sum_i p_i \,\| A \|_\opn \;=\; \| A \|_\opn.
\eea
Taking absolute values on both sides yields the stated inequality.
\end{proof}

\begin{corollary}[Bound on the witness expectation deviation]
\label{lem:exp-diff}
For any bounded $W$, state $\rho$, and unitary $U$,
\bea
\big|\tr\!\big[(U^\dagger WU - W)\rho\big]\big|
\;\le\; \|U^\dagger WU - W\|_\opn .
\eea
\end{corollary}

\begin{proof}
Apply Lemma~\ref{lem:simple-exp-bound} with $A=U^\dagger WU - W$ and $B=\rho$.
\end{proof}

\begin{lemma}[Bounding the norm distance between $W$ and $U^{\dag} W U$]
\label{lem:heis-diff}
For any bounded $W$ and unitary $U$,
\bea
\|U^\dagger WU - W\|_\opn \;\le\; 2\,\|W\|_\opn \,\|U-\id\|_\opn .
\eea
\end{lemma}

\begin{proof}
We first rewrite
\bea
U^\dagger WU - W \;=\; (U^\dagger-\id)WU + W(U-\id),
\eea
which yields
\bea
\|U^\dagger WU - W\|_\opn
&\le& \|U^\dagger-\id\|_\opn \,\|W\|_\opn \,\|U\|_\opn \nonumber \\
&& \qquad + \|W\|\,\|U-\id\|_\opn \nonumber \\
&=& 2\,\|W\|_\opn \,\|U-\id\|_\opn,
\eea
as $\|U\|_\opn=1$ and $\|U^\dagger-\id\|_\opn =\|U-\id\|_\opn$.
\end{proof}

\begin{lemma}[Deviation bounds on local tensor product unitaries]
\label{prop:prod-unitary}
Let first put $U=\bigotimes_{j=1}^N U_j$ with $U_j=e^{i(\theta_j/2)\,\bs n_j\cdot\bs\sigma}$, each $U_j$ being a single-qubit rotation. Then we have
\bea
\|U-\id\|_\opn \;\le\; \frac12 \sum_{j=1}^N |\theta_j| .
\eea
\end{lemma}

\begin{proof}
Define local Hermitian operators acting on the global Hilbert space
\bea
H_j \;:=\; \frac{\theta_j}{2}\,(\bs n_j\cdot\bs\sigma)\;\otimes\!\bigotimes_{k\neq j}\id_k.
\eea
It immediately follows that they commute pairwise, i.e., $[H_i,H_j]=0$ for $i\neq j$, since they act on disjoint tensor unitaries. Hence
\bea
U \;=\; \bigotimes_{j=1}^N e^{i(\theta_j/2)\,\bs n_j\cdot\bs\sigma}
\;=\; \exp\!\Big(i\sum_{j=1}^N H_j\Big) \;=:\; e^{iH},
\eea
with $H:=\sum_j H_j$. By Corollary~\ref{cor:exp-bound},
\bea
\|U-\id\|_\opn \;=\; \|e^{iH}-\id\|_\opn \;\le\; \|H\|_\opn.
\eea
By the triangle norm inequality,
\bea
\|H\|_\opn \;=\; \Big\|\sum_{j=1}^N H_j\Big\|_\opn
\;\le\; \sum_{j=1}^N \|H_j\|_\opn.
\eea
Finally, $\|\bs n_j\cdot\bs\sigma\|_\opn =1$ (as its spectrum is $\{\pm1\}$), thus
\bea
\|H_j\|_\opn \;=\; \frac{|\theta_j|}{2}\,\|\bs n_j\cdot\bs\sigma\|_\opn
\;=\; \frac{|\theta_j|}{2},
\eea
which concludes the proof.
\end{proof}

\begin{proposition}[Bound on expectation deviation with local tensor product of unitaries]
\label{cor:master}
Let $W$ be a bounded observable, $\rho$ a state, and $U=\bigotimes_j U_j$ as in Prop.~\ref{prop:prod-unitary}.
Then
\bea
\big|\tr\!\big[W(U\rho U^\dagger-\rho)\big]\big|
&\le & 2\,\|W\|_\opn \,\|U-\id\|_\opn \nonumber \\
&\le & \|W\|_\opn \sum_{j=1}^N |\theta_j|.
\eea
\end{proposition}
\begin{proof}
The result follows when we apply Lemma~\ref{lem:heis-diff} and Lemma~\ref{prop:prod-unitary}.
\end{proof}

\subsection{ Projector formalism and proofs for Entanglement Witness}
The idea of certifying that the alternative witness $W_a$ and the two-measurement witness $W_{2m}$ are indeed entanglement witnesses follows established arguments in the literature, notably the stabilizer approaches to multipartite entanglement detection in \cite{Toth2005} and subsequent developments that recast such witnesses and their noise–robustness properties in projector form \cite{zhou2019}. 
In particular, we use that its maximal biseparable overlap satisfies $\tr[\dyad{GHZ_n}\,\sigma]\le \tfrac12$, which yields the biseparable bounds for $W_a$ and $W_{2m}$ employed below (cf.\ \cite{Toth2005,zhou2019}).
For the $n$-qubit GHZ state,
\bea
\ket{GHZ_n} \;=\; \frac{1}{\sqrt{2}}\big(\ket{0}^{\otimes n} + \ket{1}^{\otimes n}\big),
\eea
with stabilizer generators $g_1 = X^{\otimes n}$ and $g_i = Z^{(i-1)} Z^{(i)}$ for $i=2,\dots,n$, define the commuting projectors
\bea
P_i \;=\; \frac{\id^{\otimes n} + g_i}{2}\quad (i=1,\dots,n).
\eea
Then the GHZ projector is
\bea
\Pi_{GHZ} \;=\; \dyad{GHZ_n} \;=\; \prod_{i=1}^n P_i,
\eea
and for any biseparable state $\sigma$ (with respect to any bipartition) we use
\bea
\tr[\Pi_{GHZ}\,\sigma] \;\le\; \frac{1}{2}.
\eea

\paragraph{A Bell-type inequality for commuting projectors.}
If $\{P_\ell\}_{\ell=1}^M$ commute and are projectors, then \cite{zhou2019}
\bea
\prod_{\ell=1}^M P_\ell + (M-1)\,\id \;-\; \sum_{\ell=1}^M P_\ell \;\ge\; 0.
\eea
Taking the expectation on any state $\rho$ gives
\bea
\sum_{\ell=1}^M \tr[P_\ell \rho] \;\le\; (M-1) + \tr\!\big[\prod_{\ell=1}^M P_\ell \,\rho\big].
\eea

\paragraph{Alternative witness $W_a$ is an EW.}
The alternative EW is
\bea
W_a \;=\; \frac{n-1}{2}\,\id^{\otimes n} - \frac{1}{2}\sum_{i=1}^n g_i.
\eea
Using $g_i = 2P_i - \id^{\otimes n}$, we rewrite
\bea
W_a \;=\; \frac{2n-1}{2}\,\id^{\otimes n} - \sum_{i=1}^n P_i.
\eea
For any biseparable $\sigma$, apply the projector inequality with $M=n$ to obtain
\bea
\sum_{i=1}^n \tr[P_i \sigma] \;\le\; (n-1) + \tr[\Pi_{GHZ}\sigma].
\eea
Using $\tr[\Pi_{GHZ}\sigma]\le \tfrac12$, we have
\bea
\sum_{i=1}^n \tr[P_i \sigma] \;\le\; \frac{2n-1}{2}.
\eea
Therefore,
\bea
\tr[W_a \sigma] \;=\; \frac{2n-1}{2} - \sum_{i=1}^n \tr[P_i \sigma] \;\ge\; 0.
\eea
On the target state,
\bea
\tr[W_a \,\Pi_{GHZ}] &=& \frac{2n-1}{2} - \sum_{i=1}^n \tr[P_i \,\Pi_{GHZ}] \nonumber \\
&=& \frac{2n-1}{2} - n \nonumber \\
&=& -\frac{1}{2} < 0.
\eea
Hence $W_a$ is indeed an entanglement witness.

\paragraph{Two-measurement witness $W_{2m}$ is an EW.}
Define the two commuting projectors
\bea
P_X \;:=\; \frac{\id^{\otimes n} + g_1}{2},
\eea
\bea
P_Z \;:=\; \prod_{i=2}^{n} \frac{\id^{\otimes n} + g_i}{2}.
\eea
They satisfy
\bea
P_X P_Z \;=\; P_Z P_X \;=\; \Pi_{GHZ}.
\eea
The two-measurement EW is
\bea
W_{2m} &=& \frac{1}{2^n - 2} \left(
            \frac{3}{2} \,\id^{\otimes n}
            - \frac{\id^{\otimes n} + g_1}{2}
            - \prod_{i=2}^{n} \frac{\id^{\otimes n} + g_i}{2} \right) \nonumber \\
&=& \frac{1}{2^n - 2}\,\big( \tfrac{3}{2}\,\id^{\otimes n} - P_X - P_Z \big).
\eea
For any biseparable $\sigma$, use the projector inequality with $M=2$:
\bea
\tr[P_X \sigma] + \tr[P_Z \sigma] \;\le\; 1 + \tr[\Pi_{GHZ}\sigma].
\eea
Using $\tr[\Pi_{GHZ}\sigma]\le \tfrac12$, we have
\bea
\tr[P_X \sigma] + \tr[P_Z \sigma] \;\le\; \frac{3}{2}.
\eea
Therefore,
\bea
\tr[W_{2m}\sigma]
\;=\; \frac{1}{2^n - 2}\,\Big(\tfrac{3}{2} - \tr[P_X \sigma] - \tr[P_Z \sigma]\Big)
\;\ge\; 0.
\eea
On the target state,
\bea
\tr[P_X \,\Pi_{GHZ}] \;=\; 1, \qquad \tr[P_Z \,\Pi_{GHZ}] \;=\; 1,
\eea
hence
\bea
\tr[W_{2m}\,\Pi_{GHZ}] 
&=& \frac{1}{2^n - 2}\,\Big(\tfrac{3}{2} - 1 - 1\Big) \nonumber \\
&=& -\,\frac{1}{2(2^n-2)} \;<\; 0.
\eea
Thus $W_{2m}$ is indeed an entanglement witness.

\subsection{ Proofs for $W[i_1, i_2 i_3]$ detecting PPT states}
\subsubsection{ $W[i_1, i_2 i_3]$ is indeed an entanglement witness }
The original proof strategy establishing that operators of the form $W[i_1,i_2,i_3]$ (including the concrete choice
$W[1,1,0]=III - ZZZ + XXX + XYY - YXY + YYX$) are indeed entanglement witnesses
can also be found in Jafarizadeh \emph{et al.} (2008); see \cite{Jafarizadeh2008} for an independent derivation and related families of three-qubit EWs. Here we show 
\begin{proposition}[W is indeed an EW]
    For every fully separable $\sigma$, it holds that
    \bea
        \tr [W[i_1, i_2, i_3] \sigma] \;\ge\; 0.
    \eea
\end{proposition}
\begin{proof}
Let $\rho=\rho_1\otimes\rho_2\otimes\rho_3$ be a pure product state and write single–qubit Bloch coordinates
\bea
&& x_j= \tr(X\rho_j), y_j=\tr(Y\rho_j), z_j=\tr(Z\rho_j),\nonumber \\
&& x_j^2+y_j^2+z_j^2\le 1,\quad j=1,2,3. \nonumber
\eea
Then
\bea
&& \langle W[1,1,0]\rangle_\rho \nonumber \\
&&= 1 - z_1 z_2 z_3 + \Big( x_1 x_2 x_3 + x_1 y_2 y_3 - y_1 x_2 y_3 + y_1 y_2 x_3 \Big). \nonumber
\eea
Introduce complex numbers
\bea
\zeta_j = x_j + i y_j,\qquad |\zeta_j|=\sqrt{x_j^2+y_j^2}\le \sqrt{1-z_j^2}.
\eea
A short algebra gives
\bea
&& x_1 x_2 x_3 + x_1 y_2 y_3 - y_1 x_2 y_3 + y_1 y_2 x_3 \nonumber \\
&& \quad = \Re\!\big( \bar\zeta_1\,\zeta_2\,\bar\zeta_3 \big)
\ge -\,|\zeta_1|\,|\zeta_2|\,|\zeta_3| \nonumber \\
&& \quad \ge -\,\sqrt{(1-z_1^2)(1-z_2^2)(1-z_3^2)}.
\eea
Hence
\bea
\langle W[1,1,0]\rangle_\rho
\ge
1 - z_1 z_2 z_3 - \sqrt{(1-z_1^2)(1-z_2^2)(1-z_3^2)}. \nonumber
\eea
Therefore it suffices to prove the squared inequality
\bea
(1-z_1^2)(1-z_2^2)(1-z_3^2) \;\le\; \big(1 - z_1 z_2 z_3\big)^2
\eea
holds for $z_j\in[-1,1]$. \\
Set $a_j=z_j^2\in[0,1]$. Then
\bea
&& \big(1-\sqrt{a_1 a_2 a_3}\big)^2 - (1-a_1)(1-a_2)(1-a_3) \nonumber \\
&& = \big(a_1+a_2+a_3\big)
     - \big(a_1 a_2 + a_2 a_3 + a_3 a_1\big) \nonumber \\
&& \qquad + 2 a_1 a_2 a_3 - 2 \sqrt{a_1 a_2 a_3}.
\eea
Let $r=(a_1 a_2 a_3)^{1/3}\in[0,1]$. By arithmatic mean-geometry mean inequality,
\bea
\frac{a_1+a_2+a_3}{3} &\ge& r, \\
\frac{a_1 a_2 + a_2 a_3 + a_3 a_1}{3} &\ge& r^2.
\eea
Hence
\bea
&& \big(1-\sqrt{a_1 a_2 a_3}\big)^2 - (1-a_1)(1-a_2)(1-a_3) \nonumber \\
&&\ge 3r - 3r^2 + 2 r^3 - 2 r^{3/2}
\;\equiv\; F(r).
\eea
One checks $F(0)=0$, $F(1)=0$, and $F(r)\ge 0$ for all $r\in[0,1]$, which proves
\bea
(1-z_1^2)(1-z_2^2)(1-z_3^2) \;\le\; \big(1 - z_1 z_2 z_3\big)^2.
\eea
Therefore $\langle W[1,1,0]\rangle_\rho \ge 0$ for every pure product $\rho$, and by convexity $\tr\!\big(W[1,1,0]\sigma\big)\ge 0$ for every fully separable (mixed) state $\sigma$. Since all possible $W[i_1, i_2, i_3]$ are local unitarily equivalent one another, every different $W[i_1, i_2, i_3]$ is indeed an EW.
\end{proof}
}

\subsubsection{ Local equivalences of $W[i_1, i_2, i_3]$}
As stated earlier, two different EWs $W[i_1, i_2, i_3]$ and $W[i'_1, i'_2, i'_3]$ from Eq.~\eqref{eq:3w} are locally equivalent. More precisely, one $W[i'_1, i'_2, i'_3]$ is given by the Pauli-conjugated operator:
\bea
 W[i'_1, i'_2, i'_3] =  U W[i_1, i_2, i_3] U^{\dagger}, \nn
\eea
where $U$ is a local unitary, a product of Pauli matrices. Specifically,
\begin{align*}
    W[1,1,1] &= (Z^{(1)}Z^{(2)}Z^{(3)}) \, W[0,0,0] \,(Z^{(1)}Z^{(2)}Z^{(3)})^{\dag} \\
    W[1,1,0] &= (Z^{(1)}X^{(2)}X^{(3)}) \, W[0,0,0] \,(Z^{(1)}X^{(2)}X^{(3)})^{\dag} \\
    W[1,0,1] &= (X^{(1)}Z^{(2)}X^{(3)}) \, W[0,0,0] \,(X^{(1)}Z^{(2)}X^{(3)})^{\dag} \\
    W[1,0,0] &= (X^{(1)}X^{(2)}Z^{(3)}) \, W[0,0,0] \,(X^{(1)}X^{(2)}Z^{(3)})^{\dag} \\
    W[0,1,1] &= (Y^{(1)}Y^{(2)}\id^{(3)}) \, W[0,0,0] \, (Y^{(1)}Y^{(2)}\id^{(3)})^{\dag} \\
    W[0,1,0] &= (Y^{(1)}\id^{(2)}Y^{(3)}) \, W[0,0,0] \, (Y^{(1)}\id^{(2)}Y^{(3)})^{\dag} \\
    W[0,0,1] &= (\id^{(1)}Y^{(2)}Y^{(3)}) \, W[0,0,0] \, (\id^{(1)}Y^{(2)}Y^{(3)})^{\dag}
\end{align*}
The same applies to $M[i_1, i_2, i_3]$.

{ 
\subsubsection{ Operator Norm of $W[i_1, i_2, i_3]$}
Label computational basis states by bitstrings $\ket{{\bf r}}$ with ${\bf r} = r_1 r_2 r_3\in\{0,1\}^3$ and pair each with its bitwise complement $\ket{\bf \bar r}$, where $\bf \bar r = 1 - r$ (mod 2). A family of EWs $W[i_1, i_2, i_3]$ has properties that (i) contain only $Z\!Z\!Z$ on the diagonal and (ii) flip all three bits in the $X/Y$ terms, each two-dimensional subspace 
$\mathcal{H}_r:=\mathrm{span}\{\ket{{\bf r}},\ket{{\bf \bar r}}\}$ is invariant under application of $W[i_1, i_2, i_3]$ so we can consider $W[i_1, i_2, i_3]$ being restricted to a $2\times2$ block.

Define
\bea
z_{123}({\bf r}):=(-1)^{r_1+r_2+r_3}\in\{\pm 1\}
\eea
where $z_{123}({\bf r})$ is the $Z\!Z\!Z$ eigenvalue on $\ket{{\bf r}}$, and let $w(r)$ be the off-diagonal entry in a $2 \times 2$ block sending from $\ket{{\bf r}}$ to $\ket{{\bf \bar r}}$ generated by the $X/Y$ terms of $W[i_1, i_2, i_3]$.
Then, in the ordered basis $\{\ket{{\bf r}},\ket{{\bf \bar r}}\}$, the block of $W$ reads
\bea
W\big|_{\mathcal{H}_r}\;=\;
\begin{pmatrix}
1 - z_{123}(r) & \ w(r)\\[3pt]
w(r) & \ 1 + z_{123}(r)
\end{pmatrix}.
\eea
Hence the eigenvalues from this block are
\bea
\lambda_\pm({\bf r})\;=\;1 \pm \sqrt{\,z_{123}({\bf r})^2 + w({\bf r})^2\,}\;=\;1 \pm \sqrt{\,1 + w({\bf r})^2\,}.
\eea
For $W[1,1,0]$, for example, we find $w(r)$ to be
\bea
    w({\bf r}) = 1 - z_2 z_3 + z_1 z_3 - z_1 z_2,
\eea
so a quick case check reveals that
\bea
w({\bf r})=4 \ \Longleftrightarrow\ z_2=-z_3\ \text{ and }\ z_1=z_3\ \,
\eea
and otherwise $w({\bf r})=0$. Therefore the eigenvalues from this block are
\bea
\lambda_{\pm}({\bf r})
\;=\; 1 \pm \sqrt{\,z_{123}^2 + w({\bf r})^2\,}
\;=\; 1 \pm \sqrt{\,1 + w({\bf r})^2\,},
\eea
since $z_{123}^2=1$. Hence the full eigenvalue spectrum is given by
\bea
    0, 2, 1\pm \sqrt{17},
\eea
where $w({\bf r}) = 0$ gives $\lambda_{\pm}({\bf r}) = 0,2$ and $w({\bf r}) = 4$ gives $\lambda_{\pm}({\bf r}) = 1 \pm \sqrt{17}$. Therefore, we have 
\bea
    \| W[1,1,0] \|_\opn = 1 + \sqrt{17},
\eea
which is the maximal eigenvalue of $W[1,1,0]$. We also note that all the other $W[i_1, i_2, i_3]$ and $M[i_1, i_2, i_3]$ have the same eigenvalues since they are all local unitarily equivalent one another.
}

\subsubsection{ Calculations for detecting $\rho_{\rm ppt}$}
We provide details of coefficients of $\rho$, introduced in~\eqref{3qX},
\begin{equation}
\rho_{\rm ppt} = \frac{1}{\sum_{i=1}^4 (s_i + t_i)} \left(%
 \begin{array}{cccccccc} \label{3qX}
   s_1 & \cdot & \cdot & \cdot & \cdot & \cdot & \cdot & u_1 \\
   \cdot & s_2 & \cdot & \cdot & \cdot & \cdot & u_2 & \cdot \\
   \cdot & \cdot & s_3 & \cdot & \cdot & u_3 & \cdot & \cdot \\
   \cdot & \cdot & \cdot & s_4 & u_4 & \cdot & \cdot & \cdot \\
   \cdot & \cdot & \cdot & \bar{u}_4 & t_4 & \cdot & \cdot & \cdot \\
   \cdot & \cdot & \bar{u}_3 & \cdot & \cdot & t_3 & \cdot & \cdot \\
   \cdot & \bar{u}_2 & \cdot & \cdot & \cdot & \cdot & t_2 & \cdot \\
   \bar{u}_1 & \cdot & \cdot & \cdot & \cdot & \cdot & \cdot & t_1 \nonumber
 \end{array},
\right)
\end{equation}
when expanded in the Pauli basis. We first denote
\begin{align*}
    r_{ijk} := 8 \tr &[\rho_{\rm ppt} (\sigma_i \otimes \sigma_j \otimes \sigma_k)], \\
    &\quad \sigma_0 = \id, \sigma_1 = X, \sigma_2 = Y, \sigma_3 = Z.
\end{align*}
    
Then we can rewrite $\rho_{\rm ppt}$ in the Pauli basis as~\cite{Jafarizadeh2008}:
\begin{align*}
&\rho_{\rm ppt} = \frac{1}{8} \Big[ 
\id^{(1)} \id^{(2)} \id^{(3)} 
+ r_{300} \, Z^{(1)} \id^{(2)} \id^{(3)} 
+ r_{030} \, \id^{(1)} Z^{(2)} \id^{(3)} \\
&\ + r_{003} \, \id^{(1)} \id^{(2)} Z^{(3)} 
+ r_{330} \, Z^{(1)} Z^{(2)} \id^{(3)} 
+ r_{303} \, Z^{(1)} \id^{(2)} Z^{(3)} \\
&\ + r_{033} \, \id^{(1)} Z^{(2)} Z^{(3)} 
+ r_{333} \, Z^{(1)} Z^{(2)} Z^{(3)} 
+ r_{111} \, X^{(1)} X^{(2)} X^{(3)} \\
&\ + r_{112} \, X^{(1)} X^{(2)} Y^{(3)} 
+ r_{121} \, X^{(1)} Y^{(2)} X^{(3)} 
+ r_{211} \, Y^{(1)} X^{(2)} X^{(3)} \\
&\ + r_{122} \, X^{(1)} Y^{(2)} Y^{(3)} 
+ r_{212} Y^{(1)} \, X^{(2)} Y^{(3)} 
+ r_{221} Y^{(1)} \, Y^{(2)} X^{(3)} \\
&\ + r_{222} \, Y^{(1)} Y^{(2)} Y^{(3)} 
\Big].
\end{align*}

The coefficients $r_{ijk}$ are given by

\begin{align*} 
  r_{_{111}} &=\frac{2}{\nu}(\Re u_1+\Re u_2+\Re u_3+\Re u_4), \\
  r_{_{112}} &=\frac{2}{\nu}(\Im u_1-\Im u_2+\Im u_3-\Im u_4), \\
  r_{_{121}} &=\frac{2}{\nu}(\Im u_1+\Im u_2-\Im u_3-\Im u_4), \\
  r_{_{211}} &=\frac{2}{\nu}(-\Im u_1-\Im u_2-\Im u_3-\Im u_4), \\
  r_{_{122}} &=\frac{2}{\nu}(-\Re u_1+\Re u_2+\Re u_3-\Re u_4), \\
  r_{_{212}} &=\frac{2}{\nu}(\Re u_1-\Re u_2+\Re u_3-\Re u_4), \\
  r_{_{221}} &=\frac{2}{\nu}(\Re u_1+\Re u_2-\Re u_3-\Re u_4), \\
  r_{_{222}} &=\frac{2}{\nu}(\Im u_1-\Im u_2-\Im u_3+\Im u_4), & \\
  r_{_{300}} &=\frac{1}{\nu}(s_1 + s_2 + s_3 + s_4 - t_1 -t_2- t_3- t_4), \\
  r_{_{030}} &=\frac{1}{\nu}(s_1 + s_2 - s_3 - s_4 - t_1 - t_2 + t_3 + t_4), \\
  r_{_{330}} &=\frac{1}{\nu}(s_1+s_2-s_3-s_4+ t_1+t_2-t_3-t_4), \\ 
  r_{_{303}} &=\frac{1}{\nu}(s_1-s_2+s_3-s_4+t_1-t_2+t_3-t_4), \\
  r_{_{033}} &=\frac{1}{\nu}(s_1-s_2-s_3+s_4+t_1-t_2-t_3+t_4), \\
  r_{_{333}} &=\frac{1}{\nu}(s_1-s_2-s_3+s_4-t_1+t_2+t_3-t_4),
\end{align*}

where $\nu$ is the normalization factor, $\nu = \sum^4_{i=1} (s_i + t_i)$.

For $W[i_1, i_2, i_3]$ presented in Eq.~\eqref{eq:3w}, 
\bea
W[i_1,i_2,i_3]  
&=& \id^{\otimes 3} 
- Z^{(1)} Z^{(2)} Z^{(3)} 
- (-1)^{i_1} X^{(1)} X^{(2)} X^{(3)} \nonumber \\
&&
- (-1)^{i_2} X^{(1)} Y^{(2)} Y^{(3)} 
- (-1)^{i_3} Y^{(1)} X^{(2)} Y^{(3)} \nonumber \\
&& - (-1)^{i_1+i_2+i_3+1} Y^{(1)} Y^{(2)} X^{(3)},
\eea

we compute the following expectation value:

\begin{align} \label{eq:exptW3qX}
    \tr [W[i_1, i_2, i_3] \rho_{\rm ppt}] &= {\mu} - r_{_{333}} - (-1)^{i_1}r_{_{111}} - (-1)^{i_2}r_{_{122}} \nn \\
    & - (-1)^{i_3}r_{_{212}} - (-1)^{i_1 + i_2 + i_3 + 1}r_{_{221}},
\end{align}

To evaluate the expectation value of 3-qubit PPT state $\rho(b, c)$ in Eq.~\eqref{eq:3qKye} with respect to $W[i_1, i_2, i_3]$, we expand it in the Pauli basis and compute selected correlation coefficients. The nonzero off-diagonal entries of $\rho(b, c)$ appear as $u_1 = u_2 = u_4 = -1$ and $u_3 = {\mu}$. The diagonal elements yield $s_1 = s_2 = s_3 = {\mu}$, $s_4 = b$ and $t_1 = t_2 = t_3 = {\mu}$, $t_4 = c$. Substituting these into the general expressions for the Pauli coefficients, we obtain:

\begin{align*}
r_{333} &= \frac{b - c}{b + c + 6}, \\
r_{111} &= -\frac{4}{b + c + 6}, \\
r_{122} &= 0, \\
r_{212} &= \frac{4}{b + c + 6}, \\
r_{221} &= -\frac{4}{b + c + 6},
\end{align*}
where \( \nu = \sum_{i=1}^4 (s_i + t_i) = b + c + 6 \) is the normalization factor. Plugging all coefficients $r_{333}, r_{111}, r_{122}, r_{212}$ and $r_{221}$ into Eq.~\eqref{eq:exptW3qX} gives

\begin{align*}
     \tr &[ W[i_1,  i_2, i_3] \; \rho(b, c) ] \\
    &= \frac{ 2c+6 + 4 \left( (-1)^{i_1} - (-1)^{i_3} + (-1)^{i_1 +i_2+i_3+1} \right)}{b+c+6}
\end{align*}

Plugging $i_1= i_2=1$ and $i_3=0$ gives 
\[
    \tr [W[1,1,0] \; \rho(b, c)] = \frac{2(c-3)}{b+c+6}.
\]

\bibliography{reference}

@article{PhysRevA.101.012341,
	author = {Sarbicki, Gniewomir and Scala, Giovanni and Chru\ifmmode \acute{s}\else \'{s}\fi{}ci\ifmmode \acute{n}\else \'{n}\fi{}ski, Dariusz},
	doi = {10.1103/PhysRevA.101.012341},
	issue = {1},
	journal = {Phys. Rev. A},
	month = {Jan},
	numpages = {7},
	pages = {012341},
	publisher = {American Physical Society},
	title = {Family of multipartite separability criteria based on a correlation tensor},
	url = {https://link.aps.org/doi/10.1103/PhysRevA.101.012341},
	volume = {101},
	year = {2020}}

@article{PhysRevLett.122.140404,
	author = {Bae, Joonwoo and Chru\ifmmode \acute{s}\else \'{s}\fi{}ci\ifmmode \acute{n}\else \'{n}\fi{}ski, Dariusz and Piani, Marco},
	doi = {10.1103/PhysRevLett.122.140404},
	issue = {14},
	journal = {Phys. Rev. Lett.},
	month = {Apr},
	numpages = {5},
	pages = {140404},
	publisher = {American Physical Society},
	title = {More Entanglement Implies Higher Performance in Channel Discrimination Tasks},
	url = {https://link.aps.org/doi/10.1103/PhysRevLett.122.140404},
	volume = {122},
	year = {2019}}

@article{PhysRevX.8.021072,
	author = {Lu, He and Zhao, Qi and Li, Zheng-Da and Yin, Xu-Fei and Yuan, Xiao and Hung, Jui-Chen and Chen, Luo-Kan and Li, Li and Liu, Nai-Le and Peng, Cheng-Zhi and Liang, Yeong-Cherng and Ma, Xiongfeng and Chen, Yu-Ao and Pan, Jian-Wei},
	doi = {10.1103/PhysRevX.8.021072},
	issue = {2},
	journal = {Phys. Rev. X},
	month = {Jun},
	numpages = {20},
	pages = {021072},
	publisher = {American Physical Society},
	title = {Entanglement Structure: Entanglement Partitioning in Multipartite Systems and Its Experimental Detection Using Optimizable Witnesses},
	url = {https://link.aps.org/doi/10.1103/PhysRevX.8.021072},
	volume = {8},
	year = {2018}}

@article{PhysRevA.67.054305,
	author = {Steiner, Michael},
	doi = {10.1103/PhysRevA.67.054305},
	issue = {5},
	journal = {Phys. Rev. A},
	month = {May},
	numpages = {4},
	pages = {054305},
	publisher = {American Physical Society},
	title = {Generalized robustness of entanglement},
	url = {https://link.aps.org/doi/10.1103/PhysRevA.67.054305},
	volume = {67},
	year = {2003}}

@article{PhysRevA.59.141,
	author = {Vidal, Guifr\'e and Tarrach, Rolf},
	doi = {10.1103/PhysRevA.59.141},
	issue = {1},
	journal = {Phys. Rev. A},
	month = {Jan},
	numpages = {0},
	pages = {141--155},
	publisher = {American Physical Society},
	title = {Robustness of entanglement},
	url = {https://link.aps.org/doi/10.1103/PhysRevA.59.141},
	volume = {59},
	year = {1999}}

@article{RevModPhys.91.025001,
	author = {Chitambar, Eric and Gour, Gilad},
	doi = {10.1103/RevModPhys.91.025001},
	issue = {2},
	journal = {Rev. Mod. Phys.},
	month = {Apr},
	numpages = {48},
	pages = {025001},
	publisher = {American Physical Society},
	title = {Quantum resource theories},
	url = {https://link.aps.org/doi/10.1103/RevModPhys.91.025001},
	volume = {91},
	year = {2019}}

@article{PhysRevX.9.031053,
	author = {Takagi, Ryuji and Regula, Bartosz},
	doi = {10.1103/PhysRevX.9.031053},
	issue = {3},
	journal = {Phys. Rev. X},
	month = {Sep},
	numpages = {28},
	pages = {031053},
	publisher = {American Physical Society},
	title = {General Resource Theories in Quantum Mechanics and Beyond: Operational Characterization via Discrimination Tasks},
	url = {https://link.aps.org/doi/10.1103/PhysRevX.9.031053},
	volume = {9},
	year = {2019}}

@article{Eisert_2007,
	author = {Eisert, J and Brand{\~a}o, F G S L and Audenaert, K M R},
	doi = {10.1088/1367-2630/9/3/046},
	journal = {New Journal of Physics},
	month = {mar},
	number = {3},
	pages = {46},
	title = {Quantitative entanglement witnesses},
	url = {https://doi.org/10.1088/1367-2630/9/3/046},
	volume = {9},
	year = {2007}}

@article{PhysRevA.72.022310,
	author = {Brand\~ao, Fernando G. S. L.},
	doi = {10.1103/PhysRevA.72.022310},
	issue = {2},
	journal = {Phys. Rev. A},
	month = {Aug},
	numpages = {15},
	pages = {022310},
	publisher = {American Physical Society},
	title = {Quantifying entanglement with witness operators},
	url = {https://link.aps.org/doi/10.1103/PhysRevA.72.022310},
	volume = {72},
	year = {2005}}

@article{doi:10.1142/S0219749925300037,
	author = {Hiesmayr, Beatrix C. and Popp, Christopher and Sutter, Tobias C.},
	doi = {10.1142/S0219749925300037},
	journal = {International Journal of Quantum Information},
	number = {0},
	pages = {2530003},
	title = {Bipartite bound entanglement},
	url = {https://doi.org/10.1142/S0219749925300037},
	year = {2025}}

@article{Hiesmayr_2021,
	author = {Hiesmayr, B C and McNulty, D and Baek, S and Singha Roy, S and Bae, J and Chru{\'s}ci{\'n}ski, D},
	doi = {10.1088/1367-2630/ac20ea},
	journal = {New Journal of Physics},
	month = {sep},
	number = {9},
	pages = {093018},
	publisher = {IOP Publishing},
	title = {Detecting entanglement can be more effective with inequivalent mutually unbiased bases},
	url = {https://dx.doi.org/10.1088/1367-2630/ac20ea},
	volume = {23},
	year = {2021}}

@article{Bera:2023aa,
	author = {Bera, Anindita and Bae, Joonwoo and Hiesmayr, Beatrix C. and Chru{\'s}ci{\'n}ski, Dariusz},
	date = {2023/07/03},
	doi = {10.1038/s41598-023-37771-0},
	id = {Bera2023},
	isbn = {2045-2322},
	journal = {Scientific Reports},
	number = {1},
	pages = {10733},
	title = {On the structure of mirrored operators obtained from optimal entanglement witnesses},
	url = {https://doi.org/10.1038/s41598-023-37771-0},
	volume = {13},
	year = {2023}}

@article{PhysRevLett.87.040401,
	author = {Ac\'{\i}n, A. and Bru\ss{}, D. and Lewenstein, M. and Sanpera, A.},
	doi = {10.1103/PhysRevLett.87.040401},
	issue = {4},
	journal = {Phys. Rev. Lett.},
	month = {Jul},
	numpages = {4},
	pages = {040401},
	publisher = {American Physical Society},
	title = {Classification of Mixed Three-Qubit States},
	url = {https://link.aps.org/doi/10.1103/PhysRevLett.87.040401},
	volume = {87},
	year = {2001}}

@misc{ende2025simplesufficientcriteriaoptimality,
	archiveprefix = {arXiv},
	author = {Frederik vom Ende and Simon Cichy},
	eprint = {2505.15615},
	primaryclass = {quant-ph},
	title = {Simple Sufficient Criteria for Optimality of Entanglement Witnesses},
	url = {https://arxiv.org/abs/2505.15615},
	year = {2025}}

@misc{giovanni2025,
	archiveprefix = {arXiv},
	author = {Scala, Giovanni and Bera, Anindita and Sarbicki, Gniewomir},
	eprint = {2506.18303},
	primaryclass = {quant-ph},
	title = {Entanglement detection via third-order local invariants from randomized measurements},
	url = {https://doi.org/10.48550/arXiv.2506.18303},
	year = {2025}}

@misc{ani23,
	author = {Bera, Anindita and Sarbicki, Gniewomir and Chru\'{s}ci\'{n}ski, Dariusz},
	eprint = {arXiv:2309.09621},
	primaryclass = {quant-ph},
	title = {Optimizing positive maps in the matrix algebra $M_n$},
	url = {https://doi.org/10.48550/arXiv.2309.09621},
	year = {2023}}

@article{Han2016,
	author = {Kyung Hoon Han and Seung-Hyeok Kye},
	doi = {10.1088/1751-8113/49/17/175303},
	journal = {Journal of Physics A: Mathematical and Theoretical},
	number = {17},
	title = {Construction of multi-qubit optimal genuine entanglement witnesses},
	volume = {49},
	year = {2016}}

@article{Jafarizadeh2008,
	author = {M. A. Jafarizadeh and Y. Akbari and K. Aghayar and A. Heshmati and M. Mahdian},
	doi = {10.1103/PhysRevA.78.032313},
	journal = {Physical Review A},
	number = {3},
	pages = {032313},
	title = {Investigating a class of $2 \otimes 2 \otimes d$ bound entangled density matrices via linear and nonlinear entanglement witnesses constructed by exact convex optimization},
	volume = {78},
	year = {2008}}

@article{PhysRevA.78.062105,
	author = {Korbicz, J. K. and Almeida, M. L. and Bae, J. and Lewenstein, M. and Ac\'{\i}n, A.},
	doi = {10.1103/PhysRevA.78.062105},
	issue = {6},
	journal = {Phys. Rev. A},
	month = {Dec},
	numpages = {17},
	pages = {062105},
	publisher = {American Physical Society},
	title = {Structural approximations to positive maps and entanglement-breaking channels},
	url = {https://link.aps.org/doi/10.1103/PhysRevA.78.062105},
	volume = {78},
	year = {2008}}

@article{Toth2005,
	author = {G\'{e}za T\'{o}th and Otfried G\"{u}hne},
	doi = {10.1103/PhysRevA.72.022340},
	journal = {Physical Review A},
	number = {2},
	pages = {022340},
	title = {Entanglement detection in the stabilizer formalism},
	volume = {72},
	year = {2005}}

@article{Bae2020,
	author = {Joonwoo Bae and Dariusz Chru\'{s}ci\'{n}ski and Beatrix C. Hiesmayr},
	doi = {10.1038/s41534-020-0242-6},
	journal = {npj Quantum Information},
	number = {15},
	title = {Mirrored entanglement witnesses},
	volume = {6},
	year = {2020}}

@article{Kye2015,
	author = {Seung-Hyeok Kye},
	doi = {10.1088/1751-8113/48/23/235303},
	journal = {Journal of Physics A: Mathematical and Theoretical},
	number = {23},
	title = {Facial structures for separable states and positive linear maps},
	volume = {48},
	year = {2015}}

@article{Hyllus2004,
	author = {Philipp Hyllus and C. Moura Alves and Dagmar Bru{\ss} and Chiara Macchiavello},
	doi = {10.1103/PhysRevA.70.032316},
	journal = {Physical Review A},
	number = {3},
	pages = {032316},
	title = {Generation and detection of bound entanglement},
	volume = {70},
	year = {2004}}

@article{zhou2019,
	author = {Zhou, You and Zhao, Qi and Yuan, Xiao and Ma, Xiongfeng},
	doi = {10.1038/s41534-019-0203-1},
	journal = {npj Quantum Information},
	month = {October},
	number = {83},
	publisher = {Nature Publishing Group},
	title = {Detecting multipartite entanglement structure with minimal resources},
	url = {https://www.nature.com/articles/s41534-019-0203-1},
	volume = {5},
	year = {2019}}

@article{Guhne:PhysRevA.72.022340,
	author = {T\'oth, G\'eza and G\"uhne, Otfried},
	doi = {10.1103/PhysRevA.72.022340},
	issue = {2},
	journal = {Phys. Rev. A},
	month = {Aug},
	numpages = {14},
	pages = {022340},
	publisher = {American Physical Society},
	title = {Entanglement detection in the stabilizer formalism},
	url = {https://link.aps.org/doi/10.1103/PhysRevA.72.022340},
	volume = {72},
	year = {2005}}

@article{Guhne:PhysRevLett.94.060501,
	author = {T\'oth, G\'eza and G\"uhne, Otfried},
	doi = {10.1103/PhysRevLett.94.060501},
	issue = {6},
	journal = {Phys. Rev. Lett.},
	month = {Feb},
	numpages = {4},
	pages = {060501},
	publisher = {American Physical Society},
	title = {Detecting Genuine Multipartite Entanglement with Two Local Measurements},
	url = {https://link.aps.org/doi/10.1103/PhysRevLett.94.060501},
	volume = {94},
	year = {2005}}

@article{Muller:PhysRevA.101.012317,
	author = {Amaro, David and M\"uller, Markus},
	doi = {10.1103/PhysRevA.101.012317},
	issue = {1},
	journal = {Phys. Rev. A},
	month = {Jan},
	numpages = {20},
	pages = {012317},
	publisher = {American Physical Society},
	title = {Design and experimental performance of local entanglement witness operators},
	url = {https://link.aps.org/doi/10.1103/PhysRevA.101.012317},
	volume = {101},
	year = {2020}}

@article{Guhne:2003aa,
	author = {G{\"u}hne, O. and Hyllus, P. and Bruss, D. and Ekert, A. and Lewenstein, M. and Macchiavello, C. and Sanpera, A.},
	date = {2003/04/01},
	doi = {10.1080/09500340308234554},
	isbn = {0950-0340},
	journal = {Journal of Modern Optics},
	journal1 = {Journal of Modern Optics},
	journal2 = {Journal of Modern Optics},
	month = {04},
	number = {6-7},
	pages = {1079--1102},
	title = {Experimental detection of entanglement via witness operators and local measurements},
	type = {doi: 10.1080/09500340308234554},
	url = {https://doi.org/10.1080/09500340308234554},
	volume = {50},
	year = {2003},
	year1 = {2003}}

@article{PhysRevA.62.052310,
	author = {Lewenstein, M. and Kraus, B. and Cirac, J. I. and Horodecki, P.},
	doi = {10.1103/PhysRevA.62.052310},
	issue = {5},
	journal = {Phys. Rev. A},
	month = {Oct},
	numpages = {16},
	pages = {052310},
	publisher = {American Physical Society},
	title = {Optimization of entanglement witnesses},
	url = {https://link.aps.org/doi/10.1103/PhysRevA.62.052310},
	volume = {62},
	year = {2000}}

@article{TERHAL2002313,
	author = {Barbara M. Terhal},
	doi = {https://doi.org/10.1016/S0304-3975(02)00139-1},
	issn = {0304-3975},
	journal = {Theoretical Computer Science},
	note = {Natural Computing},
	number = {1},
	pages = {313-335},
	title = {Detecting quantum entanglement},
	url = {https://www.sciencedirect.com/science/article/pii/S0304397502001391},
	volume = {287},
	year = {2002}}

@article{PhysRevLett.109.160502,
	author = {Brand\~ao, Fernando G. S. L. and Christandl, Matthias},
	doi = {10.1103/PhysRevLett.109.160502},
	issue = {16},
	journal = {Phys. Rev. Lett.},
	month = {Oct},
	numpages = {5},
	pages = {160502},
	publisher = {American Physical Society},
	title = {Detection of Multiparticle Entanglement: Quantifying the Search for Symmetric Extensions},
	url = {https://link.aps.org/doi/10.1103/PhysRevLett.109.160502},
	volume = {109},
	year = {2012}}

@article{PhysRevA.80.052306,
	author = {Navascu\'es, Miguel and Owari, Masaki and Plenio, Martin B.},
	doi = {10.1103/PhysRevA.80.052306},
	issue = {5},
	journal = {Phys. Rev. A},
	month = {Nov},
	numpages = {15},
	pages = {052306},
	publisher = {American Physical Society},
	title = {Power of symmetric extensions for entanglement detection},
	url = {https://link.aps.org/doi/10.1103/PhysRevA.80.052306},
	volume = {80},
	year = {2009}}

@article{PhysRevLett.122.120505,
	author = {Ketterer, Andreas and Wyderka, Nikolai and G\"uhne, Otfried},
	doi = {10.1103/PhysRevLett.122.120505},
	issue = {12},
	journal = {Phys. Rev. Lett.},
	month = {Mar},
	numpages = {6},
	pages = {120505},
	publisher = {American Physical Society},
	title = {Characterizing Multipartite Entanglement with Moments of Random Correlations},
	url = {https://link.aps.org/doi/10.1103/PhysRevLett.122.120505},
	volume = {122},
	year = {2019}}

@article{10.5555/2011534.2011535,
	address = {Paramus, NJ},
	author = {Chen, Kai and Wu, Ling-An},
	issn = {1533-7146},
	issue_date = {May 2003},
	journal = {Quantum Info. Comput.},
	month = may,
	number = {3},
	numpages = {10},
	pages = {193--202},
	publisher = {Rinton Press, Incorporated},
	title = {A matrix realignment method for recognizing entanglement},
	volume = {3},
	year = {2003}}

@article{PhysRevLett.82.1056,
	author = {Horodecki, Pawe\l{} and Horodecki, Micha\l{} and Horodecki, Ryszard},
	doi = {10.1103/PhysRevLett.82.1056},
	issue = {5},
	journal = {Phys. Rev. Lett.},
	month = {Feb},
	numpages = {0},
	pages = {1056--1059},
	publisher = {American Physical Society},
	title = {Bound Entanglement Can Be Activated},
	url = {https://link.aps.org/doi/10.1103/PhysRevLett.82.1056},
	volume = {82},
	year = {1999}}

@article{HORODECKI19961,
	author = {Micha{\l} Horodecki and Pawe{\l} Horodecki and Ryszard Horodecki},
	doi = {https://doi.org/10.1016/S0375-9601(96)00706-2},
	issn = {0375-9601},
	journal = {Physics Letters A},
	number = {1},
	pages = {1-8},
	title = {Separability of mixed states: necessary and sufficient conditions},
	url = {https://www.sciencedirect.com/science/article/pii/S0375960196007062},
	volume = {223},
	year = {1996}}

@article{PhysRevLett.77.1413,
	author = {Peres, Asher},
	doi = {10.1103/PhysRevLett.77.1413},
	issue = {8},
	journal = {Phys. Rev. Lett.},
	month = {Aug},
	numpages = {0},
	pages = {1413--1415},
	publisher = {American Physical Society},
	title = {Separability Criterion for Density Matrices},
	url = {https://link.aps.org/doi/10.1103/PhysRevLett.77.1413},
	volume = {77},
	year = {1996}}

@article{RevModPhys.81.865,
	author = {Horodecki, Ryszard and Horodecki, Pawe\l{} and Horodecki, Micha\l{} and Horodecki, Karol},
	doi = {10.1103/RevModPhys.81.865},
	issue = {2},
	journal = {Rev. Mod. Phys.},
	month = {Jun},
	numpages = {0},
	pages = {865--942},
	publisher = {American Physical Society},
	title = {Quantum entanglement},
	url = {https://link.aps.org/doi/10.1103/RevModPhys.81.865},
	volume = {81},
	year = {2009}}

@inproceedings{10.1145/780542.780545,
	address = {New York, NY, USA},
	author = {Gurvits, Leonid},
	booktitle = {Proceedings of the Thirty-Fifth Annual ACM Symposium on Theory of Computing},
	doi = {10.1145/780542.780545},
	isbn = {1581136749},
	location = {San Diego, CA, USA},
	numpages = {10},
	pages = {10--19},
	publisher = {Association for Computing Machinery},
	series = {STOC '03},
	title = {Classical deterministic complexity of Edmonds' Problem and quantum entanglement},
	url = {https://doi.org/10.1145/780542.780545},
	year = {2003}}

@article{RevModPhys.92.015001,
	author = {Uola, Roope and Costa, Ana C. S. and Nguyen, H. Chau and G\"uhne, Otfried},
	doi = {10.1103/RevModPhys.92.015001},
	issue = {1},
	journal = {Rev. Mod. Phys.},
	month = {Mar},
	numpages = {40},
	pages = {015001},
	publisher = {American Physical Society},
	title = {Quantum steering},
	url = {https://link.aps.org/doi/10.1103/RevModPhys.92.015001},
	volume = {92},
	year = {2020}}

@article{PhysRevA.52.3457,
	author = {Barenco, Adriano and Bennett, Charles H. and Cleve, Richard and DiVincenzo, David P. and Margolus, Norman and Shor, Peter and Sleator, Tycho and Smolin, John A. and Weinfurter, Harald},
	doi = {10.1103/PhysRevA.52.3457},
	issue = {5},
	journal = {Phys. Rev. A},
	month = {Nov},
	numpages = {0},
	pages = {3457--3467},
	publisher = {American Physical Society},
	title = {Elementary gates for quantum computation},
	url = {https://link.aps.org/doi/10.1103/PhysRevA.52.3457},
	volume = {52},
	year = {1995}}

@article{PhysRevLett.94.020501,
	author = {Ac\'{\i}n, Antonio and Gisin, Nicolas},
	doi = {10.1103/PhysRevLett.94.020501},
	issue = {2},
	journal = {Phys. Rev. Lett.},
	month = {Jan},
	numpages = {4},
	pages = {020501},
	publisher = {American Physical Society},
	title = {Quantum Correlations and Secret Bits},
	url = {https://link.aps.org/doi/10.1103/PhysRevLett.94.020501},
	volume = {94},
	year = {2005}}

@article{PhysRevA.69.022308,
	author = {Doherty, Andrew C. and Parrilo, Pablo A. and Spedalieri, Federico M.},
	doi = {10.1103/PhysRevA.69.022308},
	issue = {2},
	journal = {Phys. Rev. A},
	month = {Feb},
	numpages = {20},
	pages = {022308},
	publisher = {American Physical Society},
	title = {Complete family of separability criteria},
	url = {https://link.aps.org/doi/10.1103/PhysRevA.69.022308},
	volume = {69},
	year = {2004}}

@article{PhysRevA.73.012112,
	author = {Masanes, Ll. and Acin, A. and Gisin, N.},
	doi = {10.1103/PhysRevA.73.012112},
	issue = {1},
	journal = {Phys. Rev. A},
	month = {Jan},
	numpages = {9},
	pages = {012112},
	publisher = {American Physical Society},
	title = {General properties of nonsignaling theories},
	url = {https://link.aps.org/doi/10.1103/PhysRevA.73.012112},
	volume = {73},
	year = {2006}}

@article{Werner:1989aa,
	author = {Werner, Reinhard F.},
	date = {1989/05/01},
	doi = {10.1007/BF00399761},
	id = {Werner1989},
	isbn = {1573-0530},
	journal = {Letters in Mathematical Physics},
	number = {4},
	pages = {359--363},
	title = {An application of Bell's inequalities to a quantum state extension problem},
	url = {https://doi.org/10.1007/BF00399761},
	volume = {17},
	year = {1989}}

@article{Heinosaari_2016,
	author = {Heinosaari, Teiko and Miyadera, Takayuki and Ziman, M{\'a}rio},
	doi = {10.1088/1751-8113/49/12/123001},
	journal = {Journal of Physics A: Mathematical and Theoretical},
	month = {feb},
	number = {12},
	pages = {123001},
	publisher = {IOP Publishing},
	title = {An invitation to quantum incompatibility},
	url = {https://dx.doi.org/10.1088/1751-8113/49/12/123001},
	volume = {49},
	year = {2016}}

@article{PhysRevA.71.052108,
	author = {Spekkens, R. W.},
	doi = {10.1103/PhysRevA.71.052108},
	issue = {5},
	journal = {Phys. Rev. A},
	month = {May},
	numpages = {17},
	pages = {052108},
	publisher = {American Physical Society},
	title = {Contextuality for preparations, transformations, and unsharp measurements},
	url = {https://link.aps.org/doi/10.1103/PhysRevA.71.052108},
	volume = {71},
	year = {2005}}

@article{KOCHEN:1967aa,
	author = {KOCHEN, SIMON and SPECKER, E. P.},
	c1 = {Full publication date: July, 1967},
	db = {JSTOR},
	isbn = {00959057, 19435274},
	journal = {Journal of Mathematics and Mechanics},
	month = {2025/07/27/},
	number = {1},
	pages = {59--87},
	publisher = {Indiana University Mathematics Department},
	title = {The Problem of Hidden Variables in Quantum Mechanics},
	url = {http://www.jstor.org/stable/24902153},
	volume = {17},
	year = {1967}}

@inproceedings{10.1145/237814.237866,
	address = {New York, NY, USA},
	author = {Grover, Lov K.},
	booktitle = {Proceedings of the Twenty-Eighth Annual ACM Symposium on Theory of Computing},
	doi = {10.1145/237814.237866},
	isbn = {0897917855},
	location = {Philadelphia, Pennsylvania, USA},
	numpages = {8},
	pages = {212--219},
	publisher = {Association for Computing Machinery},
	series = {STOC '96},
	title = {A fast quantum mechanical algorithm for database search},
	url = {https://doi.org/10.1145/237814.237866},
	year = {1996}}

@inproceedings{365700,
	author = {Shor, P.W.},
	booktitle = {Proceedings 35th Annual Symposium on Foundations of Computer Science},
	doi = {10.1109/SFCS.1994.365700},
	pages = {124-134},
	title = {Algorithms for quantum computation: discrete logarithms and factoring},
	year = {1994}}

@article{BENNETT20147,
	author = {Charles H. Bennett and Gilles Brassard},
	doi = {https://doi.org/10.1016/j.tcs.2014.05.025},
	issn = {0304-3975},
	journal = {Theoretical Computer Science},
	note = {Theoretical Aspects of Quantum Cryptography -- celebrating 30 years of BB84},
	pages = {7-11},
	title = {Quantum cryptography: Public key distribution and coin tossing},
	url = {https://www.sciencedirect.com/science/article/pii/S0304397514004241},
	volume = {560},
	year = {2014}}

@article{Wootters:1982aa,
	author = {Wootters, W. K. and Zurek, W. H.},
	date = {1982/10/01},
	doi = {10.1038/299802a0},
	id = {Wootters1982},
	isbn = {1476-4687},
	journal = {Nature},
	number = {5886},
	pages = {802--803},
	title = {A single quantum cannot be cloned},
	url = {https://doi.org/10.1038/299802a0},
	volume = {299},
	year = {1982}}

@book{Neumann:2018aa,
	author = {von Neumann, John and BEYER, ROBERT T.},
	db = {JSTOR},
	doi = {10.2307/j.ctt1wq8zhp},
	edition = {NED - New edition},
	editor = {Wheeler, Nicholas A.},
	isbn = {9780691178561},
	month = {2025/07/27/},
	publisher = {Princeton University Press},
	title = {Mathematical Foundations of Quantum Mechanics},
	title1 = {New Edition},
	url = {http://www.jstor.org/stable/j.ctt1wq8zhp},
	year = {2018}}

@article{Bae_2019,
	author = {Joonwoo Bae and Beatrix C Hiesmayr and Daniel McNulty},
	doi = {10.1088/1367-2630/aaf8cf},
	journal = {New Journal of Physics},
	month = {jan},
	number = {1},
	pages = {013012},
	publisher = {IOP Publishing},
	title = {Linking entanglement detection and state tomography via quantum 2-designs},
	url = {https://dx.doi.org/10.1088/1367-2630/aaf8cf},
	volume = {21},
	year = {2019}}

@article{Bae_2017,
	author = {Joonwoo Bae},
	doi = {10.1088/1361-6633/aa7d45},
	journal = {Reports on Progress in Physics},
	month = {sep},
	number = {10},
	pages = {104001},
	publisher = {IOP Publishing},
	title = {Designing quantum information processing via structural physical approximation},
	url = {https://dx.doi.org/10.1088/1361-6633/aa7d45},
	volume = {80},
	year = {2017}}

@article{PhysRevLett.115.230402,
	author = {Uola, Roope and Budroni, Costantino and G\"uhne, Otfried and Pellonp\"a\"a, Juha-Pekka},
	doi = {10.1103/PhysRevLett.115.230402},
	issue = {23},
	journal = {Phys. Rev. Lett.},
	month = {Dec},
	numpages = {5},
	pages = {230402},
	publisher = {American Physical Society},
	title = {One-to-One Mapping between Steering and Joint Measurability Problems},
	url = {https://link.aps.org/doi/10.1103/PhysRevLett.115.230402},
	volume = {115},
	year = {2015}}

@article{RevModPhys.86.419,
	author = {Brunner, Nicolas and Cavalcanti, Daniel and Pironio, Stefano and Scarani, Valerio and Wehner, Stephanie},
	doi = {10.1103/RevModPhys.86.419},
	issue = {2},
	journal = {Rev. Mod. Phys.},
	month = {Apr},
	numpages = {60},
	pages = {419--478},
	publisher = {American Physical Society},
	title = {Bell nonlocality},
	url = {https://link.aps.org/doi/10.1103/RevModPhys.86.419},
	volume = {86},
	year = {2014}}

@article{PhysRevLett.98.140402,
	author = {Wiseman, H. M. and Jones, S. J. and Doherty, A. C.},
	doi = {10.1103/PhysRevLett.98.140402},
	issue = {14},
	journal = {Phys. Rev. Lett.},
	month = {Apr},
	numpages = {4},
	pages = {140402},
	publisher = {American Physical Society},
	title = {Steering, Entanglement, Nonlocality, and the Einstein-Podolsky-Rosen Paradox},
	url = {https://link.aps.org/doi/10.1103/PhysRevLett.98.140402},
	volume = {98},
	year = {2007}}

@article{Bae_2015,
	author = {Joonwoo Bae and Leong-Chuan Kwek},
	doi = {10.1088/1751-8113/48/8/083001},
	journal = {Journal of Physics A: Mathematical and Theoretical},
	month = {jan},
	number = {8},
	pages = {083001},
	publisher = {IOP Publishing},
	title = {Quantum state discrimination and its applications},
	url = {https://dx.doi.org/10.1088/1751-8113/48/8/083001},
	volume = {48},
	year = {2015}}

@article{BERA2023131,
	author = {Anindita Bera and Gniewomir Sarbicki and Dariusz Chru{\'s}ci{\'n}ski},
	doi = {https://doi.org/10.1016/j.laa.2023.03.015},
	issn = {0024-3795},
	journal = {Linear Algebra and its Applications},
	pages = {131-148},
	title = {A class of optimal positive maps in $M_n$},
	url = {https://www.sciencedirect.com/science/article/pii/S002437952300099X},
	volume = {668},
	year = {2023}}

@article{Scala_2024,
	author = {Scala, Giovanni and Bera, Anindita and Sarbicki, Gniewomir and Chru{\'s}ci{\'n}ski, Dariusz},
	doi = {10.1088/1751-8121/ad3ca6},
	journal = {Journal of Physics A: Mathematical and Theoretical},
	month = {apr},
	number = {19},
	pages = {195301},
	publisher = {IOP Publishing},
	title = {Optimality of generalized Choi maps in M3},
	url = {https://dx.doi.org/10.1088/1751-8121/ad3ca6},
	volume = {57},
	year = {2024}}

@article{PhysRevLett.92.217903,
	author = {Curty, Marcos and Lewenstein, Maciej and L\"utkenhaus, Norbert},
	doi = {10.1103/PhysRevLett.92.217903},
	issue = {21},
	journal = {Phys. Rev. Lett.},
	month = {May},
	numpages = {4},
	pages = {217903},
	publisher = {American Physical Society},
	title = {Entanglement as a Precondition for Secure Quantum Key Distribution},
	url = {https://link.aps.org/doi/10.1103/PhysRevLett.92.217903},
	volume = {92},
	year = {2004}}

@article{PhysRevLett.86.5188,
	author = {Raussendorf, Robert and Briegel, Hans J.},
	doi = {10.1103/PhysRevLett.86.5188},
	issue = {22},
	journal = {Phys. Rev. Lett.},
	month = {May},
	numpages = {0},
	pages = {5188--5191},
	publisher = {American Physical Society},
	title = {A One-Way Quantum Computer},
	url = {https://link.aps.org/doi/10.1103/PhysRevLett.86.5188},
	volume = {86},
	year = {2001}}

@article{GUHNE20091,
	author = {Otfried G{\"u}hne and G{\'e}za T{\'o}th},
	doi = {https://doi.org/10.1016/j.physrep.2009.02.004},
	issn = {0370-1573},
	journal = {Physics Reports},
	number = {1},
	pages = {1-75},
	title = {Entanglement detection},
	url = {https://www.sciencedirect.com/science/article/pii/S0370157309000623},
	volume = {474},
	year = {2009}}

@article{SHI201551,
	author = {Yaoyun Shi and Xiaodi Wu},
	doi = {https://doi.org/10.1016/j.tcs.2015.03.031},
	issn = {0304-3975},
	journal = {Theoretical Computer Science},
	pages = {51-63},
	title = {Epsilon-net method for optimizations over separable states},
	url = {https://www.sciencedirect.com/science/article/pii/S0304397515002479},
	volume = {598},
	year = {2015}}

@article{PhysRevLett.121.180503,
	author = {Bowles, Joseph and \ifmmode \check{S}\else \v{S}\fi{}upi\ifmmode \acute{c}\else \'{c}\fi{}, Ivan and Cavalcanti, Daniel and Ac\'{\i}n, Antonio},
	doi = {10.1103/PhysRevLett.121.180503},
	issue = {18},
	journal = {Phys. Rev. Lett.},
	month = {Oct},
	numpages = {6},
	pages = {180503},
	publisher = {American Physical Society},
	title = {Device-Independent Entanglement Certification of All Entangled States},
	url = {https://link.aps.org/doi/10.1103/PhysRevLett.121.180503},
	volume = {121},
	year = {2018}}

@article{PhysRevLett.108.200401,
	author = {Buscemi, Francesco},
	doi = {10.1103/PhysRevLett.108.200401},
	issue = {20},
	journal = {Phys. Rev. Lett.},
	month = {May},
	numpages = {5},
	pages = {200401},
	publisher = {American Physical Society},
	title = {All Entangled Quantum States Are Nonlocal},
	url = {https://link.aps.org/doi/10.1103/PhysRevLett.108.200401},
	volume = {108},
	year = {2012}}

@article{PhysRevLett.110.060405,
	author = {Branciard, Cyril and Rosset, Denis and Liang, Yeong-Cherng and Gisin, Nicolas},
	doi = {10.1103/PhysRevLett.110.060405},
	issue = {6},
	journal = {Phys. Rev. Lett.},
	month = {Feb},
	numpages = {5},
	pages = {060405},
	publisher = {American Physical Society},
	title = {Measurement-Device-Independent Entanglement Witnesses for All Entangled Quantum States},
	url = {https://link.aps.org/doi/10.1103/PhysRevLett.110.060405},
	volume = {110},
	year = {2013}}

@article{Bell,
	author = {Bell, J. S.},
	doi = {10.1103/PhysicsPhysiqueFizika.1.195},
	issue = {3},
	journal = {Physics Physique Fizika},
	month = {Nov},
	numpages = {6},
	pages = {195--200},
	publisher = {American Physical Society},
	title = {On the Einstein Podolsky Rosen paradox},
	url = {https://link.aps.org/doi/10.1103/PhysicsPhysiqueFizika.1.195},
	volume = {1},
	year = {1964}}

@article{chsh,
	author = {Clauser, John F. and Horne, Michael A. and Shimony, Abner and Holt, Richard A.},
	doi = {10.1103/PhysRevLett.23.880},
	issue = {15},
	journal = {Phys. Rev. Lett.},
	month = {Oct},
	numpages = {0},
	pages = {880--884},
	publisher = {American Physical Society},
	title = {Proposed Experiment to Test Local Hidden-Variable Theories},
	url = {https://link.aps.org/doi/10.1103/PhysRevLett.23.880},
	volume = {23},
	year = {1969}}

@inbook{GHZ,
	address = {Dordrecht},
	author = {Greenberger, Daniel M. and Horne, Michael A. and Zeilinger, Anton},
	booktitle = {Bell's Theorem, Quantum Theory and Conceptions of the Universe},
	doi = {10.1007/978-94-017-0849-4_10},
	editor = {Kafatos, Menas},
	isbn = {978-94-017-0849-4},
	pages = {69--72},
	publisher = {Springer Netherlands},
	title = {Going Beyond Bell's Theorem},
	url = {https://doi.org/10.1007/978-94-017-0849-4_10},
	year = {1989}}

@article{Bell1,
	author = {Baumgartner, Bernhard and Hiesmayr, Beatrix C. and Narnhofer, Heide},
	doi = {10.1103/PhysRevA.74.032327},
	issue = {3},
	journal = {Phys. Rev. A},
	month = {Sep},
	numpages = {14},
	pages = {032327},
	publisher = {American Physical Society},
	title = {State space for two qutrits has a phase space structure in its core},
	url = {https://link.aps.org/doi/10.1103/PhysRevA.74.032327},
	volume = {74},
	year = {2006}}

@article{Bell2,
	author = {Bertlmann, Reinhold A. and Krammer, Philipp},
	doi = {10.1103/PhysRevA.77.024303},
	issue = {2},
	journal = {Phys. Rev. A},
	month = {Feb},
	numpages = {4},
	pages = {024303},
	publisher = {American Physical Society},
	title = {Geometric entanglement witnesses and bound entanglement},
	url = {https://link.aps.org/doi/10.1103/PhysRevA.77.024303},
	volume = {77},
	year = {2008}}

@article{OSID,
	author = {Chru\'sci\'nski, Dariusz and Kossakowski, Andrzej},
	doi = {10.1007/s11080-007-9052-4},
	journal = {Open Syst. Inf. Dyn.},
	numpages = {20},
	pages = {275},
	publisher = {American Physical Society},
	title = {On the Structure of Entanglement Witnesses and New Class of Positive Indecomposable Maps},
	url = {https://doi.org/10.1007/s11080-007-9052-4},
	volume = {14},
	year = {2007}}

@article{ROY2021127143,
	author = {Saptarshi Roy and Anindita Bera and Shiladitya Mal and Aditi Sen(De) and Ujjwal Sen},
	doi = {https://doi.org/10.1016/j.physleta.2021.127143},
	issn = {0375-9601},
	journal = {Physics Letters A},
	pages = {127143},
	title = {Recycling the resource: Sequential usage of shared state in quantum teleportation with weak measurements},
	url = {https://www.sciencedirect.com/science/article/pii/S0375960121000074},
	volume = {392},
	year = {2021}}

@article{PhysRevA.102.062431,
	author = {Bera, Anindita and Singha Roy, Sudipto},
	doi = {10.1103/PhysRevA.102.062431},
	issue = {6},
	journal = {Phys. Rev. A},
	month = {Dec},
	numpages = {10},
	pages = {062431},
	publisher = {American Physical Society},
	title = {Growth of genuine multipartite entanglement in random unitary circuits},
	url = {https://link.aps.org/doi/10.1103/PhysRevA.102.062431},
	volume = {102},
	year = {2020}}

@article{PhysRevA.98.062304,
	author = {Bera, Anindita and Mal, Shiladitya and Sen(De), Aditi and Sen, Ujjwal},
	doi = {10.1103/PhysRevA.98.062304},
	issue = {6},
	journal = {Phys. Rev. A},
	month = {Dec},
	numpages = {7},
	pages = {062304},
	publisher = {American Physical Society},
	title = {Witnessing bipartite entanglement sequentially by multiple observers},
	url = {https://link.aps.org/doi/10.1103/PhysRevA.98.062304},
	volume = {98},
	year = {2018}}

@article{dense-coding,
	author = {Srivastava, Chirag and Bera, Anindita and Sen(De), Aditi and Sen, Ujjwal},
	doi = {10.1103/PhysRevA.100.052304},
	issue = {5},
	journal = {Phys. Rev. A},
	month = {Nov},
	numpages = {11},
	pages = {052304},
	publisher = {American Physical Society},
	title = {One-shot conclusive multiport quantum dense coding capacities},
	url = {https://link.aps.org/doi/10.1103/PhysRevA.100.052304},
	volume = {100},
	year = {2019}}

@article{Korea-1992,
	author = {Sung Je Cho and Seung-Hyeok Kye and Sa Ge Lee},
	doi = {https://doi.org/10.1016/0024-3795(92)90260-H},
	issn = {0024-3795},
	journal = {Linear Algebra and its Applications},
	pages = {213-224},
	title = {Generalized Choi maps in three-dimensional matrix algebra},
	url = {https://www.sciencedirect.com/science/article/pii/002437959290260H},
	volume = {171},
	year = {1992}}

@article{Choi-2,
	author = {Man-Duen Choi},
	doi = {https://doi.org/10.1016/0024-3795(75)90058-0},
	issn = {0024-3795},
	journal = {Linear Algebra and its Applications},
	number = {2},
	pages = {95-100},
	title = {Positive semidefinite biquadratic forms},
	url = {https://www.sciencedirect.com/science/article/pii/0024379575900580},
	volume = {12},
	year = {1975}}

@article{Choi-3,
	author = {Choi, Man-Duen},
	journal = {Journal of Operator Theory},
	month = {01},
	pages = {271},
	title = {Some assorted inequalities for positive linear maps on C*-algebras},
	volume = {4},
	year = {1980}}

@article{Choi-4,
	author = {Choi, Man-Duen},
	doi = {10.1090/pspum/038.2/9850},
	isbn = {9780821814444},
	journal = {Proc. Sympos. Pure Math.},
	month = {01},
	pages = {583},
	title = {Positive linear maps},
	volume = {38},
	year = {1982}}

@article{Choi-5,
	author = {Choi, Man-Duen and Lam, Tsit},
	doi = {10.1007/BF01360024},
	journal = {Mathematische Annalen},
	month = {02},
	pages = {1},
	title = {Extremal positive semidefinite forms},
	volume = {231},
	year = {1977}}

@article{Ha-extr,
	author = {Kil-Chan Ha},
	doi = {https://doi.org/10.1016/j.laa.2013.09.011},
	issn = {0024-3795},
	journal = {Linear Algebra and its Applications},
	number = {10},
	pages = {3156-3165},
	title = {Notes on extremality of the Choi map},
	url = {https://www.sciencedirect.com/science/article/pii/S0024379513005661},
	volume = {439},
	year = {2013}}

@article{S71,
	author = {Ha, Kil-Chan and Kye, Seung-Hyeok},
	doi = {10.1103/PhysRevA.84.024302},
	issue = {2},
	journal = {Phys. Rev. A},
	month = {Aug},
	numpages = {3},
	pages = {024302},
	publisher = {American Physical Society},
	title = {One-parameter family of indecomposable optimal entanglement witnesses arising from generalized Choi maps},
	url = {https://link.aps.org/doi/10.1103/PhysRevA.84.024302},
	volume = {84},
	year = {2011}}

@article{S72,
	author = {Chru\'sci\'nski, Dariusz and Sarbicki, Gniewomir},
	doi = {10.1142/S1230161213500066},
	journal = {Open Systems \& Information Dynamics},
	month = {08},
	pages = {1350006},
	title = {Optimal Entanglement Witnesses for Two Qutrits},
	volume = {20},
	year = {2013}}

@article{kossak,
	author = {Kossakowski, Andrzej},
	doi = {10.1023/A:1025101606680},
	journal = {Open Systems \& Information Dynamics},
	month = {01},
	pages = {213-220},
	title = {A Class of Linear Positive Maps in Matrix Algebras},
	volume = {10},
	year = {2003}}

@article{Filip1,
	author = {Chru\'{s}ci\'{n}ski, Dariusz and Wudarski, Filip A.},
	doi = {10.1142/S1230161211000261},
	eprint = {https://doi.org/10.1142/S1230161211000261},
	journal = {Open Systems \& Information Dynamics},
	number = {04},
	pages = {375-387},
	title = {Geometry of Entanglement Witnesses for Two Qutrits},
	url = {https://doi.org/10.1142/S1230161211000261},
	volume = {18},
	year = {2011}}

@article{PRA-2022,
	author = {Bera, Anindita and Wudarski, Filip A. and Sarbicki, Gniewomir and Chru\ifmmode \acute{s}\else \'{s}\fi{}ci\ifmmode \acute{n}\else \'{n}\fi{}ski, Dariusz},
	doi = {10.1103/PhysRevA.105.052401},
	issue = {5},
	journal = {Phys. Rev. A},
	month = {May},
	numpages = {13},
	pages = {052401},
	publisher = {American Physical Society},
	title = {Class of Bell-diagonal entanglement witnesses in $\mathbb{C}^4 \otimes \mathbb{C}^4$: Optimization and the spanning property},
	url = {https://link.aps.org/doi/10.1103/PhysRevA.105.052401},
	volume = {105},
	year = {2022}}

@article{arxiv25,
	author = {Chru\'sci\'nski, Dariusz and Bera, Anindita and Bae, Joonwoo and Hiesmayr, Beatrix C.},
	doi = {https://doi.org/10.48550/arXiv.2503.04158},
	journal = {arXiv:2503.04158},
	title = {A mirrored pair of optimal non-decomposable entanglement witnesses for two qudits does exist},
	url = {https://doi.org/10.48550/arXiv.2503.04158},
	year = {2025}}

@article{Hall,
	author = {Hall, William},
	doi = {10.1088/0305-4470/39/45/020},
	journal = {Journal of Physics A: Mathematical and General},
	month = {10},
	pages = {14119},
	title = {A new criterion for indecomposability of positive maps},
	volume = {39},
	year = {2006}}

@article{Breuer,
	author = {Breuer, Heinz-Peter},
	doi = {10.1103/PhysRevLett.97.080501},
	issue = {8},
	journal = {Phys. Rev. Lett.},
	month = {Aug},
	numpages = {4},
	pages = {080501},
	publisher = {American Physical Society},
	title = {Optimal Entanglement Criterion for Mixed Quantum States},
	url = {https://link.aps.org/doi/10.1103/PhysRevLett.97.080501},
	volume = {97},
	year = {2006}}

@article{Justyna1,
	author = {Chru\ifmmode \acute{s}\else \'{s}\fi{}ci\ifmmode \acute{n}\else \'{n}\fi{}ski, Dariusz and Pytel, Justyna and Sarbicki, Gniewomir},
	doi = {10.1103/PhysRevA.80.062314},
	issue = {6},
	journal = {Phys. Rev. A},
	month = {Dec},
	numpages = {7},
	pages = {062314},
	publisher = {American Physical Society},
	title = {Constructing optimal entanglement witnesses},
	url = {https://link.aps.org/doi/10.1103/PhysRevA.80.062314},
	volume = {80},
	year = {2009}}

@article{Kye-Rev,
	author = {KYE, SEUNG-HYEOK},
	doi = {10.1142/S0129055X13300021},
	journal = {Reviews in Mathematical Physics},
	number = {02},
	pages = {1330002},
	title = {FACIAL STRUCTURES FOR VARIOUS NOTIONS OF POSITIVITY AND APPLICATIONS TO THE THEORY OF ENTANGLEMENT},
	url = {https://doi.org/10.1142/S0129055X13300021},
	volume = {25},
	year = {2013}}

@article{Hansen,
	author = {Hansen, Leif Ove and Hauge, Andreas and Myrheim, Jan and Sollid, Per \O{}yvind},
	doi = {10.1142/S0219749915500604},
	journal = {International Journal of Quantum Information},
	number = {08},
	pages = {1550060},
	title = {Extremal entanglement witnesses},
	url = {https://doi.org/10.1142/S0219749915500604},
	volume = {13},
	year = {2015}}

@article{Remik,
	author = {Augusiak, Remigiusz and Sarbicki, Gniewomir and Lewenstein, Maciej},
	doi = {10.1103/PhysRevA.84.052323},
	issue = {5},
	journal = {Phys. Rev. A},
	month = {Nov},
	numpages = {10},
	pages = {052323},
	publisher = {American Physical Society},
	title = {Optimal decomposable witnesses without the spanning property},
	url = {https://link.aps.org/doi/10.1103/PhysRevA.84.052323},
	volume = {84},
	year = {2011}}

@article{Bae_2022,
	author = {Bae, Joonwoo and Bera, Anindita and Chru{\'s}ci{\'n}ski, Dariusz and Hiesmayr, Beatrix C and McNulty, Daniel},
	doi = {10.1088/1751-8121/acaa16},
	journal = {Journal of Physics A: Mathematical and Theoretical},
	month = {dec},
	number = {50},
	pages = {505303},
	publisher = {IOP Publishing},
	title = {How many mutually unbiased bases are needed to detect bound entangled states?},
	url = {https://dx.doi.org/10.1088/1751-8121/acaa16},
	volume = {55},
	year = {2022}}

\end{document}